\documentclass[11pt,letterpaper]{article}

\usepackage[T1]{fontenc}
\usepackage[utf8]{inputenc}
\usepackage[margin=1in]{geometry}
\usepackage{amsmath,amssymb,amsthm,mathtools}
\usepackage{algorithm}
\usepackage[noend]{algpseudocode}
\usepackage[final]{microtype}
\usepackage{xcolor}
\usepackage[hidelinks]{hyperref}
\usepackage{comment}
\usepackage{enumitem}

\usepackage{hyperref}
\usepackage{cleveref}

\crefname{appendix}{Appendix}{Appendices}
\Crefname{appendix}{Appendix}{Appendices}

\hypersetup{
plainpages = false,
bookmarksopen = true,
colorlinks = true,
citecolor = purple,
linkcolor = teal,
urlcolor  = brown,
}
\usepackage{doi}

\allowdisplaybreaks
\numberwithin{equation}{section}
\numberwithin{algorithm}{section}
\newcommand{\R}{\mathbb{R}}
\newcommand{\C}{\mathbb{C}}
\newcommand{\Sph}{\mathbb{S}}
\newcommand{\Tr}{\operatorname{Tr}}

\newcommand{\val}{\operatorname{val}}
\newcommand{\distF}{\operatorname{dist}_{\mathrm F}}
\newcommand{\poly}{\operatorname{poly}}
\newcommand{\polylog}{\operatorname{polylog}}
\newcommand{\Ot}{\widetilde{O}}
\newcommand{\Thetat}{\widetilde{\Theta}}
\newcommand{\Omegat}{\widetilde{\Omega}}
\newcommand{\Yes}{\mathsf{Yes}}
\newcommand{\No}{\mathsf{No}}
\newcommand{\sep}{\mathsf{sep}}

\newcommand{\OPT}{\operatorname{OPT}}
\newcommand{\SOS}{\operatorname{SOS}}
\newcommand{\sos}{\mathrm{sos}}

\newcommand{\ip}[2]{\langle #1,#2\rangle}

\DeclarePairedDelimiter{\norm}{\lVert}{\rVert}
\DeclareMathOperator*{\argmax}{arg\,max}
\DeclareMathOperator*{\E}{\mathbb{E}}
\DeclareMathOperator*{\wtE}{\widetilde{\mathbb{E}}}

\newcommand{\Dens}{\mathcal{D}}
\newcommand{\QMA}{\mathsf{QMA}}
\newcommand{\HSEP}{h_\mathsf{sep}}
\renewcommand{\epsilon}{\varepsilon}

\theoremstyle{plain}
\newtheorem{theorem}{Theorem}[section]
\newtheorem{lemma}[theorem]{Lemma}
\newtheorem{proposition}[theorem]{Proposition}
\newtheorem{corollary}[theorem]{Corollary}
\newtheorem{fact}[theorem]{Fact}
\theoremstyle{definition}
\newtheorem{definition}[theorem]{Definition}

\theoremstyle{remark}
\newtheorem{remark}[theorem]{Remark}

\crefname{theorem}{Theorem}{Theorems}
\Crefname{theorem}{Theorem}{Theorems}
\crefname{lemma}{Lemma}{Lemmas}
\Crefname{lemma}{Lemma}{Lemmas}
\crefname{proposition}{Proposition}{Propositions}
\Crefname{proposition}{Proposition}{Propositions}
\crefname{corollary}{Corollary}{Corollaries}
\Crefname{corollary}{Corollary}{Corollaries}
\crefname{definition}{Definition}{Definitions}
\Crefname{definition}{Definition}{Definitions}
\crefname{remark}{Remark}{Remarks}
\Crefname{remark}{Remark}{Remarks}
\crefname{fact}{Fact}{Facts}
\Crefname{fact}{Fact}{Facts}

\title{An Argmax Principle for Sum-of-Squares\\
Relaxations on the Sphere}

\author{
Fernando Granha Jeronimo%
\thanks{University of Illinois Urbana--Champaign.
Email: \url{granha@illinois.edu}}
\and
Pei Wu%
\thanks{The Pennsylvania State University.
Email: \url{pei.wu@psu.edu}}
\and
Haochen Xu%
\thanks{The Pennsylvania State University.
Email: \url{hpx5065@psu.edu}}
}
\date{}

\begin{document}
\maketitle
\begin{abstract}
We develop an argmax principle for analyzing sum-of-squares relaxations of optimization problems over the unit sphere.  Given a feasible pseudo-expectation $\wtE$, we form a polynomial in external parameters whose coefficients are high-order pseudo-moments, such as $\Phi_k(u)=\widetilde{\mathbb E}\langle x,u\rangle^{2k}$.  Our guiding principle is that the maximizers of such polynomial are candidates for rounding: their local and global optimality conditions reveal the reweighed pseudo-expectation inequalities that govern SoS convergence.  We show that this viewpoint gives a unified explanation of several SoS approximation problems where the prior work relied on rather different techniques.
In particular, under this umbrella, we obtain the following three results.
\begin{itemize}
    \item First, for the Best Separable State problem, we give a degree-$O(\sqrt {n/\epsilon})$ SoS analysis for approximating $h_{\sep}(P)$ in the perfect-completeness regime, improving and simplifying the  analysis of Barak, Kothari and Steurer (STOC'17). The dependence is essentially tight for inverse-linear gap under Exponential-Time Hypothesis, matching the hardness results from the $\QMA(2)$ protocols.

    \item Second, for the matrix $2\to4$ norm, we show that degree-$O(\sqrt n/\epsilon)$ SoS gives a multiplicative $(1 + \epsilon)$ approximation. Prior to our result, Barak et al. (STOC'12) gave an algorithm to decide if a subspace contains a $\delta$-analytically sparse vector for some constant $\delta$, with a comparable running time $\exp(O_\delta(\sqrt{n}))$. Our result upgrades it to a multiplicative approximation ratio guarantee, and to the more general $p\to q$ norms with even $q$.

    \item Finally, for general degree-$d$ polynomial optimization, we recover the convergence theorem of Bhattiprolu et al. (FOCS'17) with a shorter and more direct proof: For an arbitrary degree-$d$ polynomial, the degree-$k$ SoS based algorithm gives an approximation ratio of $O_d((n/k)^{d/2-1})$.
\end{itemize}

Thus the paper does not introduce a new relaxation.  Instead, it gives a
new way to read an SoS solution: the high-moment argmax provides a common
analytical object that unifies several previously separate convergence
analyses and yields sharper bounds or simpler proofs.
We view these three applications as evidence that this simple argmax principle is widely applicable for understanding SoS convergence beyond the settings treated here.
\end{abstract}

\newpage
\tableofcontents
\newpage

\section{Introduction}
\label{section:intro}

Optimization over the unit sphere is one of the basic continuous optimization problems behind modern algorithmic complexity.
Several central problems in theoretical computer science and quantum
information can be formulated as polynomial optimization over the sphere.  The quadratic case is spectral: maximizing $x^\top A x$ over $\norm{x}_2=1$ is exactly an eigenvalue problem.
Already for quartic and higher-degree objectives, however, the same template contains substantially harder problems.  In this paper, we study three representative
examples: the Best Separable State problem, the \(2\to4\) and $p\to q$  norms, and general
polynomial optimization.

The sum-of-squares (SoS) relaxation is the canonical SDP hierarchy for polynomial optimization problems.
At degree $2k$, it optimizes over pseudo-expectations
that behave like moments of a distribution on the sphere up to degree $2k$;
equivalently, it searches for degree-$2k$ certificates of nonnegativity
modulo the sphere constraint.  This viewpoint goes back to the moment and
semidefinite programming frameworks of Lasserre and
Parrilo~\cite{Lasserre01,Parrilo03,Lau09}.  For a problem in $n$ variables, degree
$2k$ typically leads to an SDP of size $n^{O(k)}$, so the convergence rate
as a function of $k$ is directly an algorithmic running-time question.

There are two rather different traditions for analyzing this convergence.  In
theoretical computer science, the analysis of SoS relaxations is often coupled
to an efficient rounding algorithm: one shows how to convert a feasible pseudo-expectation
into an actual vector, often through iterative reweighing or
conditioning, and apply Gaussian rounding or other problem-specific rounding procedure to the
resulting low-degree moments~\cite{BRS11,BKS14,BKS17,BHK25}.  In the optimization and quantum-information literature,
by contrast, convergence of SoS and related hierarchies are frequently proved by
harmonic analysis, approximation theory, and de Finetti type bounds~\cite{DPS04, Lau09,BCY11,DW12, dKLS17convergence, FF19, LS22}.

This paper studies a more elementary mechanism for analyzing SoS relaxations
in three settings central to theoretical computer science and quantum
complexity: the Best Separable State problem, approximation of the
\(2\!\to\!4\) and \(p\!\to\!q\) norms, and general polynomial optimization.
Given a feasible pseudo-expectation, we construct a high-moment auxiliary
polynomial in external parameters, such as a direction or a fold, and use a
maximizer of this polynomial to identify the structure around which the pseudo-expectation concentrates.  The optimality or maximality properties of this auxiliary polynomial supply analytical tools in the SoS analysis.  This gives a unified principle that yields constant-factor
approximation algorithms and explicit SoS convergence rates, while also
sharpening or simplifying analyses that previously came from rather different
techniques.

\paragraph{Background.}
% In this paper, we apply our unified principle to three central problems, the Best Separable State problem; the $2\to4$ and $p\to q$ norms problem; and the general polynomial optimization problem over the unit sphere.

Best Separable State (BSS) is a central optimization problem in quantum
information.  Given a measurement operator \(0\preceq M\preceq I\), it asks for
\[
    h_\sep(M):=\max_{\rho\otimes\sigma}
    \Tr\!\left(M(\rho\otimes\sigma)\right),
\]
where \(\rho\) and \(\sigma\) range over density operators.  Its role in quantum
complexity comes from \(\mathsf{QMA}(2)\): a verifier receives two unentangled
quantum proofs, and its maximum acceptance probability is the value of a corresponding BSS instance~\cite{KMY03,HM13}. Thus, $\QMA(2)$ protocols lead to instances whose parameters give a lower bounds for BSS problems~\cite{HM13, BBHKSZ12,HNW2019limitations}.
From a broader optimization perspective, BSS is a tensor-optimization problem:
it maximizes a linear operator over product states.  This viewpoint connects
BSS to separability testing, SDP hierarchies such as the DPS hierarchy, and SoS
algorithms for quantum information~\cite{gurvits03sep,DPS04,BCY11,BH13,BKS17}.
Under standard reductions, \(h_\sep\) is closely related to the tripartite
injective tensor norm~\cite{BBHKSZ12}.  BSS therefore belongs to the broader
landscape of tensor-norm certification and polynomial optimization; related
certification problems also arise in average-case statistical models such as
Tensor PCA~\cite{montanari2014statistical}.
% Best Separable State (BSS) is, on its face, a tensor optimization problem: one
% optimizes a measurement operator over product states.  This already places BSS in the
% broader classical landscape of tensor norms and polynomial optimization.  Its
% most prominent role, however, is in quantum complexity.  A
% $\mathsf{QMA}(2)$ verifier receives two unentangled quantum proofs, and its
% acceptance probability is the maximum of a measurement operator over separable
% states~\cite{KMY03,HM13}.  Thus the basic optimization problem is
% \[
%     h_\sep(M):=\max_{\rho\otimes\sigma}\Tr(M\rho\otimes\sigma),
% \]
% where $0\preceq M\preceq I$.  This separability and tensor-optimization
% viewpoint connects BSS to questions about the complexity of detecting quantum
% entanglement, to SDP hierarchies such as the DPS hierarchy, and to SoS
% algorithms for quantum information~\cite{gurvits03sep,DPS04,BCY11,BH13,BKS17}.  The
% connection to classical tensor optimization is quite direct: $h_{\sep}$ is
% roughly equivalent, under standard reductions, to the tripartite injective
% tensor norm~\cite{BBHKSZ12}.  This places BSS in the same landscape as injective tensor-norm certification problems; similar tensor-norm certification questions also arise in average-case statistical models such as Tensor PCA~\cite{montanari2014statistical}.  At the same time, BSS is the optimization
% core of the $\mathsf{QMA}(2)$
% line of work: two-prover protocols lead to instances whose
% parameters give a lower bounds for BSS problems~\cite{BBHKSZ12,HNW2019limitations}.

The projector $2\!\to\!4$ norm is the corresponding classical sparse-vector
problem.  For a subspace $W\subseteq\R^n$,
\[
    \max_{x\in W,\ \norm{x}_2=1}\norm{x}_4
\]
measures whether $W$ contains a vector concentrated on a small set of
coordinates.  This analytic notion of sparsity is closely tied to
hypercontractivity, small-set expansion, and Unique Games.  In particular,
the $2\!\to\!4$ norm and related $p\to q$ norms appear in the
work of Barak, Brand{\~a}o, Harrow, Kelner, Steurer, and Zhou in hypercontractivity
and small-set expansion, in subexponential algorithms for Unique Games, and in
more recent work on polynomial-time approximation guarantees for matrix
$p\to q$ norms~\cite{BBHKSZ12,ABS15subexp,Jal23,GMU23,HT26,ABH26}.
Hardness of approximating $2\to 4$ norm is often viewed as an intermediate
step toward understanding the hardness landscape around small-set expansion and
Unique Games~\cite{BBHKSZ12,BGGLT23,BGLR25}.
The broader $2\!\to\!q$ problem also has a statistical interpretation:
if the rows of a matrix $A$ are data points, then $\|A\|_{2\to q}^q/n$
is the largest empirical $q$-th moment of a one-dimensional projection.
SoS certificates for such directional moment bounds, often phrased as
certifiable bounded moments or certifiable hypercontractivity, are a basic
primitive in SoS-based algorithms for robust moment estimation, robust
mean and covariance estimation, robust regression, and mixture learning~\cite{HL18,KSS18,KKM18,BP21,ST21,HT26}.   We remark that the recent work of~\cite{GMU23, HT26} focuses on the polynomial-time regime, targeting for improved (still polynomial) approximation ratios; while we study the convergence rate and the constant approximation ratio regime. This is
why the subexponential running time regime we focus is not merely a fallback after polynomial time:
it is one of the natural scales at which current algorithms see these problems.

Finally, general polynomial optimization over the sphere is relatively underexplored in TCS compared to polynomial optimization over the simplex or Boolean cube.
For quadratics, it reduces to eigenvalues; for quartics it already contains tensor-norm and sparse-vector phenomena.
Basic hardness enters through reductions from discrete problems such as Maximum Independent Set, etc.~\cite{nesterov2003random,BBHKSZ12}; while existing convergence theorems from optimization and quantum
information often become nontrivial only at levels linear in the dimension,
outside the subexponential regime most relevant to TCS. In parallel with the polynomial-time approximation line for matrix norms,
there are also polynomial-time algorithms for arbitrary degree-$d$ forms over
the sphere, with only polynomially large approximation guarantee, for instance
$O_d(n^{d/2-1})$~\cite{HLZ10,So11}.
Bhattiprolu, Ghosh, Guruswami, Lee, and Tulsiani developed a
powerful technique of weak decoupling and polynomial folds giving explicit $O_d((n/k)^{d/2-1})$
approximation ratio of the degree-${2k}$ SoS for degree-$d$ polynomials over the unit sphere~\cite{BGGLT}.

Taken together, this set of problems comes from
some of the central algorithmic and complexity-theoretic lines around tensor
optimization, quantum separability and $\mathsf{QMA}(2)$,
small-set expansion and Unique Games.

\subsection{Our results and techniques}
In this work, the SoS programs we study are essentially the \emph{tensor-SDP} as in~\cite{BKS14}.
The central conceptual ingredient of this paper is an argmax principle for pseudo-expectations.
Given a feasible pseudo-expectation $\wtE$, we then form a high-moment polynomial in external
parameters and ask where it is maximized.  The parameter may be a direction, as
in BSS and $2\!\to\!4$ norm, or a fold, as in the general polynomial optimization problem.
Its local or global optimality conditions turn pseudo-moment feasibility into usable structural information,
which in our applications yields \emph{convergence bounds} and sometimes \emph{efficient rounding} for the SoS hierarchy.
We summarize our unifying argmax principle underlying our results:
\begin{quote}
   \emph{The argmax of a high-moment polynomial built from a feasible pseudo-expectation is a rounding object for understanding SoS convergence.}
%\emph{The argmax of a high-moment polynomial built from a feasible pseudo-expectation is the rounding object: its optimality conditions turn global pseudo-moment information into local moment inequalities.}
\end{quote}
The three main applications will use this principle in increasingly flexible ways.
For BSS, the relevant argmax is a pair of directions, and only its first- and second-order optimality conditions are needed:
after reweighing, they force a spectral gap in a pseudo-covariance matrix, which yields an approximately rank-one product direction.
This gives a sharp SoS analysis for the perfect-completeness BSS gap problem, improving the parameter dependence of~Barak--Kothari--Steurer~\cite{BKS17}.
For the \(2\to 4\) and more general \(p\to q\) norms, the same directional moment polynomial is used more strongly: global maximality gives an argmax-induced fourth-order pseudo-concentration estimate, which is the key tool to our analysis.
This yields multiplicative approximation algorithms and explicit convergence rates for the corresponding SoS relaxations, strengthening the subexponential algorithm due to Barak, Brand\~ao, Harrow, Kelner, Steurer, and Zhou~\cite{BBHKSZ12}.
For fixed-degree polynomial optimization over the sphere, the argmax object is no longer a direction but a fold; together with pseudo-H\"older and spherical moment identities, the argmax fold witnesses a large value of the given polynomial.
This recovers the Bhattiprolu--Ghosh--Guruswami--Lee--Tulsiani convergence bound with a shorter and more direct analysis~\cite{BGGLT}.

Thus the same argmax principle appears as a spectral-gap argument, a
higher-order pseudo-concentration argument, and finally a fold-selection argument.
For the problems studied here, this viewpoint gives short and
quite transparent convergence analyses. A separate
and more surprising point is that, in some cases, the same analysis also suggests an efficient rounding procedure.

We now elaborate our results.

\paragraph{Best Separable State.}
Our first result concerns the perfect-completeness Best Separable State (BSS)
problem.  The key idea is best illustrated in the clean real symmetric setting. Let $P$ be an orthogonal projector on $\R^n\otimes\R^n$, and set
\[
    h_\sep(P)
    =
    \max_{x\in\Sph^{n-1}}
    \langle x\otimes x,\;P(x\otimes x)\rangle .
\]
The $(1,1-\epsilon)$-BSS gap problem distinguishes
\[
    \Yes:\quad h_\sep(P)=1,
    \qquad
    \No:\quad h_\sep(P)\le 1-\epsilon .
\]
Equivalently, the unit-eigenspace of $P$ either contains a tensor square
$x\otimes x$, or every tensor square has projection at most $1-\epsilon$ onto
this subspace.  The body proves the slightly more general product-vector
formulation $x\otimes y$; however the symmetric case already contains the mechanism.

The SoS relaxation asks for a pseudo-expectation $\wtE$ which
behaves as if
\[
    \|x\|_2^2=1,
    \qquad
    P(x\otimes x)=x\otimes x .
\]
\begin{theorem}[Informal BSS convergence]
\label{thm:bss-convergence-intro}
The degree-$O(\sqrt {n/\epsilon})$ SoS decides $(1, 1-\epsilon)$-BSS.
\end{theorem}

Given a feasible pseudo-expectation $\wtE$, the rounding target is to make the (reweighed) second moment approximately rank one.
Concretely, if $\Sigma=\wtE xx^\top$ has top eigenvector $u$ and $\lambda_{\max}(\Sigma)\ge(1-\epsilon)\|\Sigma\|_{\mathrm F},$
then the squared distance between $u\otimes u$ and the subspace of $\operatorname{im}(P)$ is at most $O(\epsilon)$. Consequently, this direction $u$ satisfies
\[
    \langle u\otimes u,\;P(u\otimes u)\rangle\ge 1-O(\epsilon).
\]
% This is because if
% % $\lambda_{\max}(\wtE xx^T) \ge (1-\epsilon) \|(\wtE xx^T)\|_F,$
% % then the largest eigenvector $u$ of $\wtE xx^T$ satisfies that $\|P(u\otimes u)\|\ge 1-O(\sqrt\epsilon).$
% \[
%     \Sigma=\wtE xx^\top=\lambda uu^\top+B,
%     \qquad Bu=0,
%     \qquad \|B\|_{\mathrm F}\le \epsilon \lambda,
% \]
% % then the subspace constraint implies
% % \[
% %     dist_F(u\otimes u,W)\le O(\alpha),
% % \]
% % and hence
% then
% \[
%     \langle u\otimes u,\;P(u\otimes u)\rangle
%     \ge 1-O(\epsilon^2).
% \]
% Thus it is enough to obtain $\epsilon=O(\sqrt\eta)$.
In their exciting and pathbreaking work, Barak, Kothari, and Steurer obtain such a direction by Gaussian rounding and an iterative progress-or-round argument by conditioning on a cap from the Gaussian rounding~\cite{BKS17}.  We instead choose the direction by looking for the maximizer $u$ of the following polynomial taking advantage of the whole pseudo-moment matrix in one-shot,
\begin{equation}
\label{eq:Phi-k}
    u=\argmax_{v\in\Sph^{n-1}} \Phi_k(v), \qquad \Phi_k(v): =\wtE\langle x,v\rangle^{2k}.
\end{equation}
Consider the pseudo-expectation of $\wtE$ reweighed by $w(x)=\langle x,u\rangle^{2k-2}$, denoted $\wtE_w$. First-order derivative
shows that the maximizer $u$ of $\Phi_k$ is an eigenvector of the reweighed second pseudo-moment matrix $\Sigma_w=\wtE_w xx^\top$; while
the second-order derivative controls
the spectral gap, such that for any $h\in \Sph^{n-1}$ and $h\perp u$,
\begin{equation}
\label{eq:second-order-argmax-comparison}
\wtE_w \langle x,h\rangle^2 \lesssim \frac{\wtE_w \langle x,u\rangle^2}{k}.
\end{equation}
Consequently, all the $(n-1)$ eigenvalues except the largest one contribute at most $(n-1)/k^2$ fraction of the squared Frobenius mass in $\Sigma_w$. Hence, a degree-$O(\sqrt {n/\epsilon})$ SoS is feasible if and only if the input is a yes instance.
Indeed, taking $k\gtrsim \sqrt{n/\epsilon}$, one obtains
$\langle u\otimes u,\;P(u\otimes u)\rangle\ge 1-\epsilon.$
Note that for the gap problem, one never needs to actually compute the maximizer $u$, it is only for the analysis.

\bigskip
\emph{Tightness.} Our analysis improves BKS's degree-$\tilde O(\sqrt n /\epsilon^2)$ SoS on the dependence of $\epsilon$ to essentially tight.\footnote{\cite{BKS17} remarks that the $\sqrt n$ dependence in the exponent has to be there in the current techniques.} This $\epsilon$ parameter is important in view of the quantum complexity. There has been a long line of research on $\QMA(2)$ protocols with short-proofs certifying $\mathsf{NP}$-hard problems~\cite{ABDFS08,CD10, BT12, CF13, LGNN12, NZ23qFreegame, LW26}, which gives the hardness side for the hsep and the BSS problem.
In particular for the regime that matters here, the $(\log n + O(1))$-size proof $\mathsf{QMA}(2)$ protocol due to Le Gall, Nakagawa,
and Nishimura for
3-SAT instances of size $L=\Ot(n)$, reduces to perfect-completeness BSS instances with a local
dimension $n$ and inverse-linear gap $\epsilon=\Thetat(1/n)$~\cite{LGNN12}.  In this regime, our SoS algorithm has degree
$\Ot(n)$ with running time $2^{\Ot(n)}$ that
decides the BSS instances thus decides the 3-SAT problem with running time $2^{\Ot(L)}$. This matches the Exponential-Time Hypothesis (ETH), establishing the conditional tightness of this SoS algorithm.
Similarly, on the protocol side, there is no logarithmic-proof $\QMA(2)$ protocols that can certify 3-SAT with perfect-completeness and a $1/o(n)$ gap under ETH.
%A constant gap protocol, on the other hand, would resolve a major open problem in quantum complexity.

We should remark that the tradeoffs between local dimension $n$ and the gap $\epsilon$ can be delicate. Our tightness statement is specific to the inverse-linear gap regime $\epsilon=\Thetat(1/n)$. It does not preclude a different tradeoff. For example, nothing prevents even a quasi-polynomial time algorithm for BSS when $\epsilon=\Omega(1)$ at the moment~\cite{HM13,JLW26}.

\bigskip
In our view, this new analysis is quite satisfying. It is  elementary and the core calculation is completely automatic once we decide to analyze the derivatives of the polynomial $\Phi_k$. Furthermore, this argmax object of $\Phi_k$ is deterministic, one-shot, and avoids the iterative progress or the cumulants arguments in BKS's analysis, resulting in an ETH-tight bound.
What might be even more surprising, \Cref{thm:bss-convergence-intro} also has a companion efficient algorithmic rounding statement.  The BKS-style
rounding mechanism makes efficiency manifest because it samples Gaussian
directions and iterates.  On the other hand, the maximizer of $\Phi_k$ in our analysis initially looks less constructive:
$\Phi_k$ is high degree, and a global maximizer is not obviously computable
within the moment-SDP time budget.  The derivative proof removes this obstacle.
The analysis only needs the second-order local
optimality behind \eqref{eq:second-order-argmax-comparison}; algorithmically,
a standard smooth optimization routine on the product sphere, such as a
Riemannian trust-region method, finds a positive approximate second-order
local maximizer~\cite{BAC18}.

\medskip
\paragraph{Approximating $2\!\to\!4$ and $p\!\to\!q$ norms.}
Our second result gives a multiplicative approximation for the matrix
$2\!\to\!4$ norm, and more generally $p\!\to\!q$ norm. In the introduction, for simplicity we talk about $2\!\to\!4$ norm, which captures the main idea. We will use $\norm{\cdot}_p$ to denote the $\ell_p$-norm, thus for an operator $P\in \R^{n\times m}$
\[
    \norm P_{2\to 4} := \max_{x\in \Sph^{m-1}} \|Px\|_4 = \max_{x\in \Sph^{m-1}} \left(\sum_i (Px)_i^4\right)^{1/4}.
\]
We further restrict $P\in\R^{n\times n}$ to be a nonzero orthogonal projector and set $W=\operatorname{im}(P)$.\footnote{Projector is not really important, our analysis extends to non-projector operator. For the application to small-set expansion and Unique Games though, we normally care about the projector of top eigenspace of the adjacency matrix from some underlying graph.}
Here $\norm{P}_{2\to 4}$ is an analytic sparsity statistic on $W$: a flat unit vector $v$ on
$\delta n$ coordinates has $\norm{v}_4 ^4 =1/(\delta n)$.  This is why the
$2\!\to\!4$ norm appears naturally around small-set expansion and Unique Games.
\begin{theorem}[Approximating $2\to 4$ norm]
\label{thm:norm-convergence-intro}
Given a projector $P\in \R^{n\times n}$, a degree-$2k$ SoS based algorithm approximates $\|P\|_{2\to4}$ to a multiplicative factor of $1+O(\sqrt n/k + n/k^2)$. In particular,
a degree-$O(\sqrt n /\epsilon) $ SoS based algorithm approximates $\|P\|_{2\to4}$ to the approximation ratio $(1+\epsilon)$.
\end{theorem}

Prior to our result, Barak et al., motivated by applications to Unique Games, gave a $\exp(O_\delta(\sqrt n))$ time algorithm that decides if a given subspace $W$ contains a unit vector $v$ with $\|v\|_4^4 \ge 1/(\delta n)$ for some constant $\delta>0$~\cite{BBHKSZ12}. Their algorithm follows from a simple dimension argument: when the subspace $W$ has dimension at least $\Omega_\delta(\sqrt n)$, then it must contain a $\delta$-analytically sparse vector; in case when the subspace $W$ has dimension $O(\sqrt n )$, one can exhaustive search to find if there is such a sparse vector. To the best of the authors' knowledge, our result is the first multiplicative $(1+\epsilon)$-approximation algorithm with a comparable running time $\exp(\Ot(\sqrt n)).$

Our degree-$2k$ SoS relaxation asks for a
pseudo-expectation $\wtE$, given a threshold $\tau$, behaving as if
\[
    \norm{x}_2^2=1,
    \qquad
    Px=x,
    \qquad
    \norm {x}_4^4\ge \tau .
\]
The rounding again uses
\[
    \Phi_k(v)=\wtE\langle x,v\rangle^{2k},
    \qquad
    v\in \Sph(W).
\]
Let $u\in \Sph(W)$ maximize $\Phi_k$.
Reweigh $\wtE$ by $w(x)=\langle x, u\rangle ^{2k-4}$. We would like to establish that the pseudo-distribution is pseudo-concentrated around the direction $u$ after the reweighing, and consequently $\wtE_w \|u\|^4_4\ge(1-\epsilon) \wtE_w\|x\|_4^4$.
The exact analytical tool
is the following inequality:
for every $h\in \Sph(W)\cap u^\perp$,
\begin{equation}
\label{eq:fourth-order-argmax-comparison}
    \wtE_w \langle x,h\rangle^4
    \lesssim
    \frac{\wtE_w \langle x,u\rangle^4}{k^2}.
\end{equation}
In the second-moment setting, the same argmax comparison gives a spectral gap for the reweighed pseudo-covariance matrix~\eqref{eq:second-order-argmax-comparison}, here it gives a higher-moment analogue: an argmax-induced fourth-moment concentration estimate. This is where the argument goes beyond the derivative calculation used for BSS. There, the maximizer was used merely through first- and second-order optimality, whereas here the fourth-moment objective forces us to use global maximality of $\Phi_k$ to derive higher-moment pseudo-expectation inequalities.

Decompose the pseudo-random $x\in W$ into the $u$ direction and its orthogonal complement, say $x=\alpha u + z$, where $z=(P-uu^\top)x$.
A coordinate-wise Sum-of-Squares inequality applied to $x_i=\alpha u_i+z_i$
reduces $\wtE_w \|x\|_4^4$ to the decoupled contribution of $\wtE_w \|u\|_4^4$ and $\wtE_w \|z\|_4^4$.\footnote{For example, one can use this quartic SoS inequality: $(a+b)^4 \le(1+\delta)^3a^4+(1+\delta^{-1})^3b^4$ for any $\delta>0$ of our choice, which can decouple the contribution of $\wtE_w \|u\|_4^4$ and $\wtE_w \|z\|_4^4$ cleanly. In our actual analysis, we tried to optimize the parameters, which complicates the proof.}
% \[
%     (\alpha u_i+z_i)^4
%     \le
%     (1+\eta)\alpha^4u_i^4
%     +4\alpha^3u_i^3z_i
%     +C_\eta z_i^4 .
% \]
% The point of keeping the linear term is that it is exactly the term killed by
% the first-order optimality of the high-moment argmax after reweighting.  Thus
% the fourth-moment mass of the pseudo-random vector is bounded by the mass along
% the selected direction $u$, up to a tangent fourth-moment error.  The latter
% is precisely what the fourth-order argmax comparison controls.
Then~\eqref{eq:fourth-order-argmax-comparison} controls the contribution to $\wtE_w \|x\|^4$ from directions orthogonal to $u$. To see this, note that $z_i=\langle x,(P-uu^\top)e_i\rangle$, thus
\begin{equation}
\label{eq:z-four-norm-bound-intro}
    \wtE_w \| z\|_4^4 = \wtE_w \left[\sum_i \langle x, (P-uu^\top)e_i\rangle^4\right] \lesssim \frac{\wtE_w \langle x, u\rangle^4}{k^2} \cdot \sum_i \|(P-uu^\top)e_i\|^4.
\end{equation}
Without loss of generality, the big summation on the right hand side is bounded above by $\tau n$, simplifying
\begin{equation}
\label{eq:norm-orth-contrib}
    \wtE_w \| z\|_4^4 \lesssim \frac{\wtE_w \langle x, u\rangle^4 \cdot n\tau}{k^2}.
\end{equation}
This is because $\|(P-uu^\top)e_i\|^2\le \| P e_i \|^2=P_{ii}$. If a diagonal entry already satisfies
$P_{ii}^2\ge \tau$, then $v:=Pe_i/\sqrt{P_{ii}}$ is a good
vector in $W$, as one can verify that $\|v\|_4^4\ge v_i^4 = P_{ii}^2\ge \tau$. So we find a sparse vector trivially. Otherwise $\sum_i P_{ii}^2\le n\tau$.\footnote{We remark that one can also without loss of generality assume that $\sum_i P_{ii}^2 \le \dim (W)^2 \tau$: instead of checking $P e_i$, one can check $P g$ for a Gaussian random $g$, see for example~\cite{BKS14}. This gives faster convergence when the subspace dimension is small.} Thus, \eqref{eq:norm-orth-contrib} explains why $k$ scales like $\sqrt{n}$.

Surprisingly once again, the argmax rounding analysis can be converted into an efficient rounding algorithm.
The key observation is that: although~\eqref{eq:fourth-order-argmax-comparison} is stated for every
direction \(h\in W\cap u^\perp\) which is guaranteed only by global maximality of $u$; however in the  analysis~\eqref{eq:z-four-norm-bound-intro}, it is
applied only to \((P-uu^\top)e_i\).  Moreover, to compare $u$ and one fixed direction $(P-uu^\top)e_i$, it suffices to focus on the 1D unit circle $\mathcal C_i(u)$ in  $\mathop{\mathrm{span}} \{u, (P-uu^\top)e_i\}$.
%A unit circle can be discretized using the standard approximation theory tool like Ehlich-Zeller's grid~\cite{EZ64}.
Now, the union of $\mathcal C_i(u)$ is a much smaller search space.
Putting everything together, given a feasible pseudo-expectation, our efficient rounding algorithm would be the following iterative search starting from a good initial point $u$:
\begin{enumerate}[itemsep=-2pt]
    \item Search over the union of the discretized $\mathcal C_i(u)$ looking for a $v$ such that $\Phi_k(v) \ge (1+\epsilon) \Phi_k(u)$;
    \item If no such $v$, $u$ is good for our rounding analysis; otherwise, we can update $u\leftarrow v$, and iterate.
\end{enumerate}

\bigskip
The same argument extends to matrix \(p\!\to q\) norms, with additional technical work. We state this extension below
\begin{theorem}[Approximating $p\to q$ norms]
\label{thm:pq-convergence-intro}
Fix $1<p<q$, where $p$ is rational and $q$ is an even integer, and set
$\lambda(p):=\max\{p,2\}$.  Given a nonzero matrix
$A\in\R^{N\times d}$, a degree-$2k$ SoS based algorithm approximates
$\|A\|_{p\to q}$ to a multiplicative factor
\[
    1+
    O_{p,q}\!\left(
        \frac{N^{2/q}}{k^{2/\lambda(p)}}
        +
        \frac{N}{k^{q/\lambda(p)}}
    \right).
\]
In particular, degree
\[
    O_{p,q}\!\left(
        \frac{N^{\lambda(p)/q}}{\epsilon^{\lambda(p)/2}}
    \right)
\]
SoS gives a $(1+\epsilon)$-approximation to $\|A\|_{p\to q}$.
\end{theorem}
% Putting things together one can obtain
% \[
%     \tau^{1/4}
%     \le
%     \norm{u}_4
%     +
%     O\!\left(\frac{\sqrt n}{k}\right)\tau^{1/4}.
% \]
% Hence, for $k\ge C\sqrt n/\varepsilon$, every feasible threshold
% pseudo-expectation rounds to a vector $u\in W$ with
% \[
%     \norm{u}_4^4\ge (1-\varepsilon)\tau .
% \]
% Combining this with binary search gives, in time
% $n^{O(\sqrt n/\varepsilon)}$, an estimate $\widehat{\tau}$ satisfying
% \[
%     (1+\varepsilon)^{-1/4}\norm{P}_{2\to4}
%     \le
%     \widehat{\tau}^{1/4}
%     \le
%     (1-\varepsilon)^{-1/4}\norm{P}_{2\to4}.
% \]
% The familiar $\exp(\Ot(\sqrt n))$ constant-accuracy scale is therefore
% explained by the same one-shot calculus principle as in BSS, with the
% $1/k^2$ gain in \eqref{eq:fourth-order-argmax-comparison} replacing the $1/k$ bound in \eqref{eq:second-order-argmax-comparison}.

\medskip
\paragraph{General polynomial optimization.}
Our third application shows that the same principle is not tied to directional rounding.
For fixed-degree polynomial optimization, the argmax object is a \emph{fold} rather than a direction.
This recovers the seminal result due to Bhattiprolu, Ghosh, Guruswami, Lee, and Tulsiani~\cite{BGGLT}, which illustrates how versatile this argmax principle can be.

Let $f:\R^n \to \R$ be a degree-$d$ polynomial. The goal is to optimize
\[
    \norm{f}_\Sph:=\max_{x\in \Sph^{n-1}} |f(x)|.
\]
\begin{theorem}[Convergence of degree-$d$ polynomial optimization]
\label{thm:homo-opt-convergence-intro}
For any degree-$d$ polynomial $f$, there is a degree-$O(k)$ SoS algorithm outputting an estimate $\hat\tau$, such that
    \begin{equation}
        \norm{f}_\Sph \le \hat\tau \le O_d\left(\left(\frac{n}{k}\right)^{d/2-1}\right)\norm{f}_\Sph.
    \end{equation}
\end{theorem}

We will further assume that $f$ is homogeneous from now on. This is without loss of generality over the unit sphere by the standard trick: one can decompose $f$ into the odd and even part, then within either the odd or even part, multiply a homogeneous degree-$d'$ component by some power of $\|x\|^2$ to enforce the global homogeneity~\cite{DW12}.

Our degree-$O(k)$ SoS relaxation asks for a pseudo-expectation $\wtE$, which maximizes the linear objective $\wtE f(x)$ (or $-\wtE f(x)$), behaving as if
\[
    \|x\|_2^2 = 1.% \qquad f(x) \ge \tau.
\]

In the introduction, to keep notations simple, we focus on the quartic case. Let $T:(\R^n)^4\to\R$ be a symmetric $4$-linear form, such that $f(x)=T(x,\ldots,x).$
The folding viewpoint says that one can fix two tensor slots first.  For
$u,v\in\Sph^{n-1}$, define the quadratic fold
\[
    q_{u,v}(x)=T(u,v,x,x).
\]
This is justified by the polarization lemma that
\[
    \max_{x\in\Sph^{n-1}} |T(x,x,x,x)| \approx _d \max_{u,v,x,x \in \Sph^{n-1
    }} |T(u,v,x,x)|.
\]
Therefore, up to a loss of a fixed constant factor depending on $d$, to approximate $|f(x)|$, it suffices to look for the best $u,v$.
Once $(u,v)$ is fixed, optimizing $x\mapsto q_{u,v}(x)$ is an eigenvalue problem, reducing the hard quartic rounding
step to the computationally easy spectral step. The remaining question is how
to choose $(u,v)$.  This is the same polynomial-folding viewpoint as in many of the previous works~\cite{KN07,HLZ10,So11}, and in particular in~\cite{BGGLT}. The theorem of Bhattiprolu et al.
reaches the $O_d((n/k)^{d/2-1})$ convergence scale through the novel \emph{weak decoupling} technique they introduce, which is technically quite demanding and delicate~\cite{BGGLT}.

Our argmax principle naturally suggests that one could choose $u,v$ to be the maximizer of the following polynomial,
\[
    \Phi_k(u,v):=\wtE [ T(u,v,x,x)^{2k}].
\]
This strategy leads us to a remarkably simple analysis. Suppose $(u^*,v^*)\in \argmax \Phi_k$. Let $p(x):= \langle u\otimes v, x\otimes x\rangle.$ Key to our analysis is the following pseudo-H\"older inequality,
\begin{equation}\label{eq:quartic-holder}
\Phi_k(u^*,v^*)\ge \E_{u,v} \wtE [T(u,v,x,x)^{2k}] \ge \frac{(\E_{u,v}\wtE[ p^{2k-1} T(u,v,x,x)]) ^{2k} }{ (\E_{u,v} \wtE [ p^{2k}])^{2k-1}}.
\end{equation}
The denominator is some powered even spherical moment. In particular, let $c_{n,k}$ denote the $2k$-th moment of the first coordinate of a Haar random $v\in \Sph^{n-1}$, i.e., $c_{n,k} := \E_{v\in\Sph^{n-1}} v_1^{2k}$. Then
\begin{equation}\label{eq:quartic-holder-denominator}
    \E_{u,v} \wtE [p^{2k}] = \wtE\,\E_{u,v} [\langle u, x\rangle^{2k} \langle v, x\rangle^{2k}] = c_{n,k}^2.
\end{equation}
The numerator, as $T$ is a 4-linear form, can be simplified
\begin{equation}\label{eq:quartic-holder-numerator}
    \E_{u,v}\wtE[ p^{2k-1} T(u,v,x,x)] = \wtE\,\E_{u,v} [ p^{2k-1} T(u,v,x,x)]=c_{n,k}^2 \wtE[T(x,x,x,x)].
\end{equation}
Abbreviate $\wtE[T(x,x,x,x)]=\tau$, and plug~\eqref{eq:quartic-holder-denominator}-\eqref{eq:quartic-holder-numerator} to~\eqref{eq:quartic-holder}, one obtains that
\begin{equation*}
   \Phi_k(u^*,v^*) \ge c_{n,k}^2 \tau ^{2k}.
\end{equation*}
Let $H$ be the symmetric matrix such that $\langle x, Hx\rangle = T(u^*, v^*,x,x)$ and $\lambda$ be the operator norm of $H$, one can verify that $\lambda^{2k}\|x\|^{4k} - \langle x, Hx\rangle^{2k}$ is a degree-$4k$ sum-of-squares.
Therefore,
\[
    \max_{x\in\Sph^{n-1}} T(u^*,v^*,x,x)^{2k} = \|H\|_{\mathrm{op}}^{2k} \ge \wtE \langle x, H x\rangle^{2k} = \Phi_k(u^*,v^*)\ge c_{n,k}^2\tau^{2k}.
\]
Since the $(2k)$-th spherical moment $c_{n,k}$ scales like $(k/n)^{k}$, this shows that $\|f\|_\Sph \ge O_d(k/n)\tau$ for quartic polynomials by the polarization lemma, matching the convergence rate in~\Cref{thm:homo-opt-convergence-intro} with $d$ set to 4.
% The reason this works is already visible before any derivative calculation.
% Let
% \[
%     p_{u,v}(x)=\langle u,x\rangle\langle v,x\rangle,
%     \qquad
%     c_{n,k}:=\E_{a\sim\Sph^{n-1}}\langle a,e_1\rangle^{2k}.
% \]
% Pseudo-H\"older, applied to the positive functional
% $\E_{u,v}\wtE_x$, gives
% \[
%     \E_{u,v}\wtE q_{u,v}(x)^{2k}
%     \ge
%     \frac{
%         \left|\E_{u,v}\wtE p_{u,v}(x)^{2k-1}q_{u,v}(x)\right|^{2k}
%     }{
%         \left(\E_{u,v}\wtE p_{u,v}(x)^{2k}\right)^{2k-1}
%     } .
% \]
% The two spherical averages are explicit:
% \[
%     \E_{u,v}\wtE p_{u,v}(x)^{2k}=c_{n,k}^2,
%     \qquad
%     \E_{u,v}\wtE p_{u,v}(x)^{2k-1}q_{u,v}(x)
%     =
%     c_{n,k}^2\wtE[f(x)]
%     \ge c_{n,k}^2\tau .
% \]
% Hence
% \[
%     \Phi_k(u_*,v_*)
%     \ge
%     \E_{u,v}\wtE q_{u,v}(x)^{2k}
%     \ge
%     c_{n,k}^2\tau^{2k}.
% \]
% Since $c_{n,k}^{1/(2k)}\gtrsim\sqrt{k/n}$, this gives
% \[
%     \Phi_k(u_*,v_*)^{1/(2k)}
%     \gtrsim
%     \frac{k}{n}\,\tau .
% \]

% Here the same fixed-degree scale is recovered from the
% moment side by a direct fold-selection argument.  The quartic discussion above
% is only the simplest way to see the mechanism: in degree $d$, one folds
% $d-2$ slots, the same spherical-moment calculation contributes the factor
% $(k/n)^{(d-2)/2}$, and the remaining computation is quadratic.

\subsection{Summary and future directions}
In this paper, we formalize an argmax principle for constructing one-shot, deterministic rounding arguments from maximizers of high-moment polynomials. The auxiliary polynomials are simple: to select a direction, we use a powered overlap; to select a fold, we use a powered fold. Despite this simplicity, this principle turns out to be powerful enough to simplify the analysis of several cornerstone results, and sometimes improving the parameters to the central problem in TCS like BSS, approximating $2\to4$ norm.

At first glance, maximizing a high-degree polynomial need not be computationally tractable. Nevertheless, the resulting analysis often identifies weaker optimality conditions that can be found efficiently, thereby leading to efficient rounding algorithms.

The simplicity of our analyses suggests that the argmax principle may apply beyond the clean spherical settings treated here.

\subsection{Organization}
Section~\ref{sec:prelim} fixes notation and recalls the basic SoS and pseudo-expectation facts used throughout the paper.  Section~\ref{sec:bss} treats perfect-completeness BSS.
Section~\ref{sec:sos-2to4} proves the projector $2\to4$ norm approximation guarantee, including a finite-search rounding. Section~\ref{sec:sos-pq} extends the same method to matrix $p\!\to q$ norms.
Section~\ref{sec:degree-d-argmax-fold-sos} proves the
fixed-degree homogeneous polynomial optimization theorem using the fold-selection argument.
The appendices contain the detailed BSS rounding proof, finite polynomial descriptions for $p$-unit balls, and auxiliary inequalities used in the main text.

\section{Preliminaries}
\label{sec:prelim}

\subsection{Notation and basic linear algebra}
\label{subsec:notation-prelim}
For an integer $n\ge1$, write $[n]=\{1,\ldots,n\}$ and let
$e_1,\ldots,e_n$ denote the standard basis of $\R^n$.  The Euclidean inner
product is denoted by $\ip{x}{y}$, and the Euclidean unit sphere is
\[
    \Sph^{n-1}:=\{x\in\R^n:\norm{x}_2=1\}.
\]
If $W\subseteq\R^n$ is a linear subspace, then
$\Sph(W):=W\cap\Sph^{n-1}$. For $1\le p<\infty$, $\norm{x}_p$ denotes the
usual $\ell_p$ norm.  Vectors are columns unless stated otherwise.

For matrices, we use the Frobenius inner product
\[
    \ip{A}{B}:=\Tr(A^\top B),
\]
and the associated norm $\norm{A}_{\mathrm F}$.  The operator norm is
$\norm{A}_{\operatorname{op}}$.  For symmetric matrices, $A\succeq0$ means
that $A$ is positive semidefinite, and $0\preceq A\preceq B$ means that
$A\succeq0$ and $B-A\succeq0$.  If $W\subseteq\R^{m\times n}$ is a matrix
subspace, $\Pi_W$ denotes the Frobenius-orthogonal projector onto $W$, and
\[
    \distF(X,W):=\norm{X-\Pi_WX}_{\mathrm F}.
\]
The vectorization map is normalized by
\[
    \operatorname{vec}(xy^\top)=x\otimes y,
    \qquad x\in\R^m,
    \quad y\in\R^n.
\]

The symbols $O(\cdot)$, $\Omega(\cdot)$, and $\Theta(\cdot)$ have their usual
asymptotic meanings.  The notation $\Ot(\cdot)$, $\Omegat(\cdot)$, and
$\Thetat(\cdot)$ suppresses polylogarithmic factors. Constants hidden in
$O_{p,q}(\cdot)$ or $O_d(\cdot)$ depend only on the fixed parameters appearing
in the subscript.

\subsection{Quantum separability and \texorpdfstring{$\QMA(2)$}{QMA(2)}}
\label{subsec:quantum-prelim}
For a finite-dimensional real or complex Hilbert space $H$, let
\[
    \Dens(H):=\{\rho\succeq0:\Tr\rho=1\}
\]
be the set of density operators on $H$.  A bipartite state on
$H_A\otimes H_B$ is a product state if it has the form
$\rho\otimes\sigma$ with $\rho\in\Dens(H_A)$ and $\sigma\in\Dens(H_B)$; it is
separable if it is a convex combination of product states.

For an accept operator $0\preceq M\preceq I$ acting on $H_A\otimes H_B$, its
separable value is
\[
    h_\sep(M)
    :=
    \max_{\rho\in\Dens(H_A),\,\sigma\in\Dens(H_B)}
    \Tr\bigl(M(\rho\otimes\sigma)\bigr).
\]
Since the objective is linear in each density operator, the maximum is
attained by pure product states.  In the real case $H_A=\R^m$ and
$H_B=\R^n$, this gives
\[
    h_\sep(M)
    =
    \max_{x\in\Sph^{m-1},\,y\in\Sph^{n-1}}
    \ip{x\otimes y}{M(x\otimes y)} .
\]
All statements below are over the reals since a
complex instance can be converted to real by a standard realification; this changes
local dimensions by at most a constant factor and preserves the asymptotic
running-time and gap scales used below. Moreover, if either $m=1$ or $n=1$, then the problem is trivial. Thus we assume $m,n \geq 2$ throughout the remainder of the paper.
The Best Separable State problem, also called separable optimization or
$\HSEP$ in this paper, is to estimate or distinguish the value
$h_\sep(M)$.

A $\QMA(2)$ verifier receives two quantum proofs that are promised to be
unentangled across the two proof registers.  After fixing the verifier's circuit, the final accept/reject measurement induces an accept POVM element \(0 \preceq M \preceq I\) on the two proof
registers, obtained by tracing out the verifier's private workspace and
ancilla registers.  On proofs
$\rho\in\Dens(H_A)$ and $\sigma\in\Dens(H_B)$, the acceptance probability is
\(
    \Tr\bigl(M(\rho\otimes\sigma)\bigr),
\)
and the optimum acceptance probability of the verifier is precisely
$h_\sep(M)$.  A $\QMA(2)$ protocol with completeness $c$ and soundness
$s$ therefore induces an $\HSEP$ instance with yes-case
$h_\sep(M)\ge c$ and no-case $h_\sep(M)\le s$.  The quantity
$c-s$ is the completeness-soundness gap; perfect completeness means $c=1$.

\subsection{Sum-of-squares, pseudo-expectations and reweighing}
\label{subsec:sos-prelim}
We take the ``pseudo'' distributional formalization of SoS~\cite{BKS14, BBHKSZ12, BKS17}. Thus we sometimes omit ``pseudo'' and abuse terminologies like covariance, concentration, etc. from probability theory. Let $z=(z_1,\ldots,z_N)$ and let $\R[z]_{\le D}$ denote the vector space of
real polynomials in $z$ of degree at most $D$.  A degree-$D$
pseudo-expectation is a linear functional
\[
    \wtE:\R[z]_{\le D}\to\R
\]
satisfying
\[
    \wtE[1]=1,
    \qquad
    \wtE[h(z)^2]\ge0
    \quad\text{for every }h\in\R[z]_{\le D/2}.
\]
Equivalently, the pseudo-moment matrix indexed by monomials of degree at most
$\lfloor D/2\rfloor$ is positive semidefinite.

Polynomial equalities are imposed in the pseudo-expectation sense as follows.
For constraints $g_i(z)=0$, we require
\[
    \wtE[p(z)g_i(z)]=0
    \qquad\text{whenever }\deg(p g_i)\le D.
\]
Polynomial inequalities $h_j(z)\ge0$ are imposed by requiring
\[
    \wtE[q(z)^2h_j(z)]\ge0
    \qquad\text{whenever }\deg(q^2h_j)\le D.
\]
All feasibility regions used above are affine slices of the cone of degree-$D$
pseudo-moment matrices together with such positivity constraints, and hence are
semidefinite programs.

For $N$ variables and degree $D=2r$, the principal moment matrix has dimension
\(
    \binom{N+r}{r}.
\)
Thus a degree-$D$ pseudo-expectation can be optimized over by an SDP whose main
matrix dimension is $N^{O(D)}$ for variable degree.  Throughout the paper, SDP running times are stated in the standard
real-arithmetic model, with polynomial dependence on the input encoding length
and numerical precision suppressed unless otherwise specified.

The Cauchy–Schwarz inequality also holds in the pseudo-expectation sense.

\begin{lemma}[pseudo-Cauchy–Schwarz]\label{pseudo-Cauchy–Schwarz}
   For \(f,g\in\R[z]\) with
\(
    2\deg f\le D,
    2\deg g\le D,
\)
we have
\begin{equation}
\label{eq:pseudo-cauchy-schwarz}
    \bigl(\wtE[fg]\bigr)^2
    \le
    \wtE[f^2]\,\wtE[g^2].
\end{equation}
\end{lemma}
Indeed, for any $\alpha,\beta\in\R$, positivity of
\(
    \wtE[(\alpha f+\beta g)^2]\ge0
\)
implies that the matrix
\(
    \begin{pmatrix}
        \wtE[f^2] & \wtE[fg]\\
        \wtE[fg] & \wtE[g^2]
    \end{pmatrix}
\)
is positive semidefinite, which gives
\eqref{eq:pseudo-cauchy-schwarz}.  In particular, if
\(\wtE[f^2]=0\), then
\(
    \wtE[fg]=0
\)
for every \(g\) of compatible degree.

We also use the following standard sum-of-squares terminology.  A polynomial
\(p\in\R[z]\) is a sum of squares, abbreviated SoS, if there exist polynomials
\(g_1,\ldots,g_s\in\R[z]\) such that
\[
    p(z)=\sum_{i=1}^s g_i(z)^2 .
\]
We write \(p\in\SOS\) for this condition.  More generally, an inequality
\(p(z)\ge0\) has a degree-\(D\) SoS certificate if \(\deg(g_i^2)\le D\) for every \(i\). We will use the elementary SoS fact.
\begin{fact}\label{psd-sos}
    If \(B\succeq0\), then the quadratic form \(x^\top Bx\) is SoS.
\end{fact}
 Indeed, writing
\(B=L^\top L\) gives
\(
    x^\top Bx=\norm{Lx}_2^2=\sum_i (Lx)_i^2 .
\)

We will use normalized reweighings of pseudo-expectations.  Let
\(\wtE\) be a degree-\(D\) pseudo-expectation in variable \(z\), and let
\(w(z)\in\SOS\) have degree \(2s\).  Assume that
\(
    \wtE[w]>0.
\)
The \(w\)-reweighed pseudo-expectation is the linear functional
\[
    \wtE_{w}[f]
    :=
    \frac{\wtE[w f]}{\wtE[w]},
    \qquad
    f\in\R[z]_{\le D-2s}.
\]
Then one can verify that \(\wtE_{w}\) is a degree-\((D-2s)\) pseudo-expectation, which also preserves polynomial constraints at the degree \(D-2s\).
\subsection{Norms, duality maps, and polynomial objectives}
\label{subsec:norms-polynomials}
For a linear map $A:\R^d\to\R^N$ and $1<p,q<\infty$, define
\[
    \norm{A}_{p\to q}
    :=
    \max_{x\ne0}\frac{\norm{Ax}_q}{\norm{x}_p}.
\]
The dual exponent of $p$ is
\[
    p^*:=\frac{p}{p-1}.
\]
For $1<s<\infty$, the duality map $J_s:\R^d\to\R^d$ is
\[
    J_s(w)_j:=\operatorname{sgn}(w_j)|w_j|^{s-1}.
\]
If $\norm{w}_s=1$ and $u=J_s(w)$, then
$\norm{u}_{s^*}=1$ and $\ip{w}{u}=1$.

A $d$-linear form $T:(\R^n)^d\to\R$ is symmetric if
\[
    T(x_1,\ldots,x_d)=T(x_{\pi(1)},\ldots,x_{\pi(d)})
\]
for every permutation $\pi$ of $[d]$.  Its associated homogeneous degree-$d$
polynomial is
\[
    f_T(x):=T(x,\ldots,x),
\]
and its spherical norm is
\[
    \norm{f_T}_{\Sph}:=\max_{\norm{x}_2=1}|f_T(x)|.
\]
When the form is clear, we write $f$ for $f_T$.

% ============================================================
% Part 1: BSS
% ============================================================

\section{Best Separable State (BSS) with perfect-completeness}
\label{sec:bss}
This section develops the sum-of-squares distinguisher for the perfect-completeness BSS gap problem.
We first reduce it to a rank-one matrix subspace problem, then formulate a degree-\(4k\) SoS feasibility test for this subspace problem.
The feasibility relaxation is the standard degree-\(4k\) tensor-SDP/SoS relaxation, following the framework developed in earlier SoS rounding
works~\cite{BRS11,BKS14} and instantiated for BSS by
Barak--Kothari--Steurer~\cite{BKS17}.
The main difference lies in the soundness analysis. Here we prove its soundness by analyzing a pseudo-moment polynomial on \(\Sph^{m-1}\times\Sph^{n-1}\).  The proof extracts a product direction from a maximizer of this polynomial through a reweighing argument and shows that the resulting direction has large projection onto the target subspace.

\subsection{BSS gap formulation and subspace reduction}

Using the notation of \Cref{subsec:quantum-prelim}, the perfect-completeness BSS gap problem is the following promise problem.

\begin{definition}[\((1,1-\epsilon)\)-Best Separable State gap problem]
For \(\epsilon\in(0,1)\), the perfect-completeness Best Separable State (BSS) gap problem asks to distinguish the following two cases for an input operator \(0\preceq M\preceq I\):
\[
\Yes:\quad h_\sep(M)=1,
\qquad
\No:\quad h_\sep(M)\le 1-\epsilon .
\]
\end{definition}

Following Barak--Kothari--Steurer~\cite{BKS17}, we use the equivalent
rank-one-subspace formulation of perfect-completeness BSS.  Let \(W\subseteq \R^{m\times n}\) be a linear subspace.  The corresponding objective is
\[
\val(W)
=
\max_{x\in \Sph^{m-1},\,y\in \Sph^{n-1}}
\|\Pi_W(xy^\top)\|_{\mathrm F}^2 .
\]
This objective gives the following gap problem.

\begin{definition}[Rank-one subspace gap problem]
For \(\epsilon\in(0,1)\) and a subspace \(W\subseteq\R^{m\times n}\), the rank-one subspace gap problem \(\textsc{RankOneSubspace}_{m,n}(W,\epsilon)\) asks to distinguish
\[
\Yes:\quad \exists x\in \Sph^{m-1},\; y\in \Sph^{n-1}
\text{ such that } xy^\top\in W,
\]
from
\[
\No:\quad
\|\Pi_W(xy^\top)\|_{\mathrm F}^2\le 1-\epsilon
\quad
\text{for every }x\in \Sph^{m-1},\;y\in \Sph^{n-1}.
\]
Since \(\|xy^\top\|_{\mathrm F}=1\), the no case can equivalently be written as
\[
\No:\quad
\distF(xy^\top,W)^2\ge \epsilon
\quad
\text{for every }x\in \Sph^{m-1},\;y\in \Sph^{n-1}.
\]
\end{definition}

The deterministic distinguisher will use the following reduction.  In the perfect-completeness regime, a BSS acceptance operator can be replaced by its unit-eigenspace and then viewed, under vectorization, as a matrix subspace.

\begin{proposition}[Gap-preserving reduction to rank-one subspace gap problem]
\label{prop:bss-subspace-reduction}
Let \(0\preceq M\preceq I\) be a real BSS instance acting on \(\R^m\otimes\R^n\). Define
\[
W_M=\{X\in\R^{m\times n}:\operatorname{vec}(X)\in\ker(I-M)\}.
\]
For every \(\epsilon\in(0,1)\), the map \(M\mapsto W_M\) reduces the perfect-completeness BSS gap problem with parameter \(\epsilon\) to \(\textsc{RankOneSubspace}_{m,n}(W_M,\epsilon)\).
\end{proposition}

\begin{proof}
Let \(z=x\otimes y=\operatorname{vec}(xy^\top)\), where \(x\in\Sph^{m-1}\) and \(y\in\Sph^{n-1}\).  If \(h_\sep(M)=1\), then for some such product vector \(z\),
\[
\langle z,Mz\rangle=1.
\]
Since \(I-M\succeq0\), this implies \(\langle z,(I-M)z\rangle=0\), and therefore \(z\in\ker(I-M)\).  Thus \(xy^\top\in W_M\).

Conversely, assume \(h_\sep(M)\le 1-\epsilon\).  Then for every unit product vector \(z=x\otimes y\),
\[
\langle z,(I-M)z\rangle\ge \epsilon .
\]
Let \(P\) be the orthogonal projector onto \(\ker(I-M)\).  Since \(0\preceq I-M\preceq I\) and \(P\) projects onto the kernel of \(I-M\), the spectral decomposition of \(I-M\) gives
\[
I-M\preceq I-P .
\]
Consequently,
\[
\epsilon
\le \langle z,(I-M)z\rangle
\le \langle z,(I-P)z\rangle
=\operatorname{dist}(z,\ker(I-M))^2
=\distF(xy^\top,W_M)^2 .
\]
This is exactly the no condition in the rank-one subspace formulation.  Hence \(W_M\) is a \(\Yes\)-instance when \(h_\sep(M)=1\), and it is a \(\No\)-instance when \(h_\sep(M)\le 1-\epsilon\).
\end{proof}

\subsection{Decision algorithm}
We now describe the SoS feasibility test for the rank-one subspace gap problem.
The input to this subsection is a subspace \(W\subseteq \R^{m\times n}\) and a gap parameter
\(\epsilon\in(0,1)\).  For BSS instances, we first apply Proposition
\ref{prop:bss-subspace-reduction} and then run the same test on \(W_M\).

The degree-\(4k\) sum-of-squares relaxation for the rank-one subspace problem asks for a degree-\(4k\) pseudo-expectation on variables \(x,y\), satisfying
\[
\|x\|_2^2=1,\qquad \|y\|_2^2=1,\qquad xy^\top\in W .
\]
The last constraint means that for every \(A\in W^\perp\),
\[
\wtE q(x,y)\,\langle A,xy^\top\rangle=0
\qquad
\text{whenever } \deg q+2\le 4k .
\]
For a target gap \(\epsilon\in(0,1)\),  choose the degree parameter \(k\ge 1\) so that
\begin{equation}
\label{eq:kcondition}
2k-1\ge \frac{\sqrt{2}\,R_{m,n}}{\sqrt{\epsilon}},
\end{equation}
where
\[
R_{m,n}=((m-1)(n-1))^{1/4}.
\]
With this choice of degree, the distinguisher is the following SoS feasibility test.

\begin{algorithm}[H]
\caption{Degree-\(4k\) SoS distinguisher for the rank-one subspace gap problem}
\label{alg:sos-bss}
\begin{algorithmic}[1]
\Require A subspace \(W\subseteq \R^{m\times n}\), a gap parameter \(\epsilon\in(0,1)\), and an integer \(k\) satisfying \eqref{eq:kcondition}.
\State Form the degree-\(4k\) moment SDP for a pseudo-expectation \(\wtE\) in variables \(x\in\R^m\), \(y\in\R^n\), with constraints
\[
\|x\|_2^2=1,
\qquad
\|y\|_2^2=1,
\qquad
xy^\top\in W
\]
imposed in the pseudo-expectation sense.
\State Solve the resulting semidefinite feasibility problem.
\If{the SDP is feasible}
    \State \Return \(\Yes\).
\Else
    \State \Return \(\No\).
\EndIf
\end{algorithmic}
\end{algorithm}

The running time of Algorithm~\ref{alg:sos-bss} is governed by the size of the degree-\(4k\) moment matrix.  The number of monomials of degree at most \(2k\) in \(m+n\) variables is
\(
\binom{m+n+2k}{2k}=(m+n)^{O(k)} ,
\)
and the number of scalar moments of degree at most \(4k\) is \((m+n)^{O(k)}\).  The sphere constraints and the linear equations defining \(W^\perp\) contribute only a polynomial number of linear constraints in this moment space.  Thus, in the standard exact real-arithmetic SDP model, Algorithm~\ref{alg:sos-bss} runs in \((m+n)^{O(k)}\) arithmetic operations, up to polynomial factors in the input description of \(W\).  In the balanced case \(m=n=D\), this is \(D^{O(k)}\).

Completeness of Algorithm~\ref{alg:sos-bss} is immediate: if \(xy^\top\in W\) for some unit vectors \(x,y\),
then evaluation at this pair defines a feasible pseudo-expectation at every
degree.  The nontrivial part is soundness. The remainder of this section proves
that, for \(k\) satisfying \eqref{eq:kcondition}, feasibility of the degree-\(4k\)
relaxation implies the existence of a unit product matrix with projection at
least \(1-\epsilon\) onto \(W\).  Therefore, on \(\No\)-instances the SDP is
infeasible.

\subsection{Spherical directional moment polynomial and maximizer properties}

The soundness analysis below is organized around an argmax-principle-guided reweighing and rounding step: from any feasible pseudo-expectation, we select maximizers, a product direction $(u,v)$, of a \emph{spherical directional moment polynomial}; reweighing by the associated spherical cap, we prove that $uv^\top$ has large overlap in $W$.

Throughout this subsection, fix a feasible degree-\(4k\) pseudo-expectation satisfying
\[
\|x\|_2^2=1,\qquad \|y\|_2^2=1,\qquad xy^\top\in W .
\]
The central auxiliary object is the following \emph{spherical directional moment polynomial} on the product of the two unit spheres.  For \(p\in\Sph^{m-1}\) and \(q\in\Sph^{n-1}\), set
\begin{equation}\label{eq:spherical-directional-momemt-polynomial}
\Phi_k(p,q)=\wtE\langle x,p\rangle^{2k}\langle y,q\rangle^{2k}.
\end{equation}

For unit directions \(p\) and \(q\), \(\Phi_k(p,q)\) records the joint high-order pseudo-moment of the projections of \(x\) and \(y\) onto these directions.  We use maximizers of this polynomial to select a product direction around which the pseudo-expectation can be reweighed and rounded.
The proof only uses first- and second-order optimality of \(\Phi_k\) with respect to each sphere variable separately.  We formalize this condition as follows.

\begin{definition}[coordinate-wise second-order local maximizer]
A pair \((u,v)\in \Sph^{m-1}\times \Sph^{n-1}\) is a coordinate-wise second-order local maximizer of \(\Phi_k\) if, for every unit \(h\perp u\) and every unit \(\ell\perp v\), the curves
\[
u_h(\tau)=\frac{u+\tau h}{\sqrt{1+\tau^2}},
\qquad
v_\ell(\tau)=\frac{v+\tau \ell}{\sqrt{1+\tau^2}}
\]
satisfy
\[
\left.\frac{d}{d\tau}\Phi_k(u_h(\tau),v)\right|_{\tau=0}=0,
\qquad
\left.\frac{d}{d\tau}\Phi_k(u,v_\ell(\tau))\right|_{\tau=0}=0,
\]
and
\[
\left.\frac{d^2}{d\tau^2}\Phi_k(u_h(\tau),v)\right|_{\tau=0}\le 0,
\qquad
\left.\frac{d^2}{d\tau^2}\Phi_k(u,v_\ell(\tau))\right|_{\tau=0}\le 0 .
\]
\end{definition}

For later use, we calculate derivative identities for each sphere variable.  Let \((u,v)\in\Sph^{m-1}\times\Sph^{n-1}\), and set
\[
s=\langle x,u\rangle,
\qquad
 t=\langle y,v\rangle .
\]
For a unit vector \(h\perp u\), let \(r=\langle x,h\rangle\) and \(u_h(\tau)=(u+\tau h)/\sqrt{1+\tau^2}\).  Differentiating
\[
\Phi_k(u_h(\tau),v)
=
\wtE (s+\tau r)^{2k}(1+\tau^2)^{-k}t^{2k}
\]
at \(\tau=0\) gives
\begin{equation}
\label{eq:x-first-derivative}
\left.\frac{d}{d\tau}\Phi_k(u_h(\tau),v)\right|_{\tau=0}
=
2k\,\wtE s^{2k-1}t^{2k}r,
\end{equation}
and
\begin{equation}
\label{eq:x-second-derivative}
\left.\frac{d^2}{d\tau^2}\Phi_k(u_h(\tau),v)\right|_{\tau=0}
=
2k(2k-1)\wtE s^{2k-2}t^{2k}r^2
-
2k\wtE s^{2k}t^{2k}.
\end{equation}
Similarly, for a unit vector \(\ell\perp v\), let \(q=\langle y,\ell\rangle\) and \(v_\ell(\tau)=(v+\tau \ell)/\sqrt{1+\tau^2}\).  Then
\begin{equation}
\label{eq:y-first-derivative}
\left.\frac{d}{d\tau}\Phi_k(u,v_\ell(\tau))\right|_{\tau=0}
=
2k\,\wtE s^{2k}t^{2k-1}q,
\end{equation}
and
\begin{equation}
\label{eq:y-second-derivative}
\left.\frac{d^2}{d\tau^2}\Phi_k(u,v_\ell(\tau))\right|_{\tau=0}
=
2k(2k-1)\wtE s^{2k}t^{2k-2}q^2
-
2k\wtE s^{2k}t^{2k}.
\end{equation}

The derivative identities will be applied at a \emph{positive} maximizer of \(\Phi_k\), whose existence follows from averaging over the two spheres.

\begin{lemma}[Positive maximum]
\label{lem:positive-max}
The polynomial \(\Phi_k\) has a global maximizer \((u,v)\in\Sph^{m-1}\times\Sph^{n-1}\) with \(\Phi_k(u,v)>0\). Moreover, every global maximizer is a coordinate-wise second-order local maximizer.
\end{lemma}

\begin{proof}
Let \(\sigma_m\) and \(\sigma_n\) denote the uniform probability measures on \(\Sph^{m-1}\) and \(\Sph^{n-1}\).  For fixed \(x\), Lemma \ref{lem:degree-d-spherical-moments} gives
\[
\int_{\Sph^{m-1}}\langle x,p\rangle^{2k}\,d\sigma_m(p)
=
c_{m,k}\|x\|_2^{2k},
\]
where \(c_{m,k}>0\) is the \(2k\)-th moment of the first coordinate of a uniformly random point on \(\Sph^{m-1}\).  Similarly,
\[
\int_{\Sph^{n-1}}\langle y,q\rangle^{2k}\,d\sigma_n(q)
=
c_{n,k}\|y\|_2^{2k}.
\]
Therefore
\[
\int_{\Sph^{m-1}}\int_{\Sph^{n-1}}
\Phi_k(p,q)\,d\sigma_n(q)\,d\sigma_m(p)
=
c_{m,k}c_{n,k}\,\wtE \|x\|_2^{2k}\|y\|_2^{2k}.
\]
The sphere constraints imply
\[
\wtE \|x\|_2^{2k}\|y\|_2^{2k}=1,
\]
so the average displayed above is \(c_{m,k}c_{n,k}>0\).  Hence \(\Phi_k\) is positive at some point.  Since \(\Sph^{m-1}\times\Sph^{n-1}\) is compact and \(\Phi_k\) is continuous, a positive global maximizer exists.  At any global maximizer, the first derivative along every tangent retraction varying one sphere variable while the other is fixed vanishes, and the second derivative along the same curve is nonpositive.  Thus every global maximizer is a coordinate-wise second-order local maximizer.
\end{proof}

%The next theorem is the rounding step: given a feasible pseudo-expectation and a positive coordinate-wise second-order local maximizer, the corresponding product direction has large projection onto \(W\).
The next theorem is the rounding step: it shows that any positive coordinate-wise second-order local maximizer $(u,v)$ of the directional moment polynomial already gives a good product direction: the rank-one matrix \(uv^\top\) has almost all of its Frobenius mass inside \(W\).

\begin{theorem}[Argmax rounding for BSS]
\label{thm:main}
Let \(\wtE\) be a degree-\(4k\) pseudo-expectation satisfying
\[
\|x\|_2^2=1,\qquad
\|y\|_2^2=1,\qquad
xy^\top\in W.
\]
Let \((u,v)\in\Sph^{m-1}\times\Sph^{n-1}\) be a coordinate-wise second-order local maximizer of \(\Phi_k\) with \(\Phi_k(u,v)>0\).  Then
\[
\|\Pi_W(uv^\top)\|_{\mathrm F}^2
\ge
1-
\frac{\sqrt{(m-1)(n-1)}}{(2k-1)^2}.
\]
\end{theorem}

\begin{proof}
For the given \((u,v)\),  we decompose the pseudo-random product vector into its components parallel and orthogonal to \((u,v)\).  Set
\[
s=\langle x,u\rangle,
\qquad
 t=\langle y,v\rangle,
\qquad
\xi=(I-uu^\top)x,
\qquad
\zeta=(I-vv^\top)y .
\]
Then
\[
x=su+\xi,
\qquad
 y=tv+\zeta.
\]

\medskip
\paragraph{Reweighing and the reweighed cross-moment matrix.} The reweighing is by a cap associated with the maximizer directions $(u,v)$.  Define
\[
w(x,y)= \langle x, u\rangle^{2k-2} \langle y,v\rangle^{2k-2}.
\]
Consider the following pseudo-expectation reweighing:
\[
\wtE_w f
=
\frac{\wtE[ wf]}{\wtE w}.
\]
% We now argue that the above reweighing is well-defined. Introduce the normalizing quantities
% \[
% A=\wtE s^{2k}t^{2k}=\Phi_k(u,v)>0,
% \qquad
% Z=\wtE w.
% \]
% Since
% \[
% w-s^{2k}t^{2k}
% =
% s^{2k-2}t^{2k-2}(1-s^2t^2),
% \]
% and, modulo the two sphere constraints,
% \[
% 1-s^2t^2=(1-s^2)+s^2(1-t^2)
% =
% \|\xi\|_2^2+s^2\|\zeta\|_2^2,
% \]
% the polynomial \(s^{2k-2}t^{2k-2}(1-s^2t^2)\) is a sum of squares modulo the constraints.  Therefore
% \[
% Z-A=\wtE\bigl[s^{2k-2}t^{2k-2}(1-s^2t^2)\bigr]\ge 0,
% \]
% and hence \(Z\ge A>0\).
%Let \(A=\wtE s^{2k}t^{2k}=\Phi_k(u,v)>0\) and \(Z=\wtE w\).
Since, modulo the sphere constraints,
\[
    w-s^{2k}t^{2k}
    =
    s^{2k-2}t^{2k-2}(\|\xi\|_2^2+s^2\|\zeta\|_2^2)
\]
is a sum of squares, we have \(\wtE w \ge \wtE s^{2k}t^{2k} = \Phi_k(u,v) >0\). Thus \(\wtE_w\) is
well-defined. Since \(w=(s^{k-1}t^{k-1})^2\) has degree \(4k-4\),
\(\wtE_w\) is a degree-\(4\)
pseudo-expectation.

The central object is the following reweighed cross-moment matrix
\[
M=\wtE_w[(sx)(ty)^\top]\in W.
\]
% The reason \(M\in W\) is as follows. For every \(B\in W^\perp\),
% \[
% \langle B,M\rangle
% =
% \frac{\wtE s^{2k-1}t^{2k-1}\langle B,xy^\top\rangle}{Z}
% =
% 0,
% \]
% since the degree of the polynomial inside \(\wtE\) is \(4k\) and the constraint \(xy^\top\in W\) is imposed through degree \(4k\).
Since \(xy^\top\in W\) holds as a pseudo-expectation constraint, reweighing preserves the such containment.  Thus \(M\in W\). Indeed, by our SoS constraints for every \(B\in W^\perp\),
\[
\langle B,M\rangle
=
\frac{\wtE
s^{2k-1}t^{2k-1}\langle B,xy^\top\rangle}
{\wtE w}
=0.
\]

The bulk of our analysis concerns the singular-decomposition structure of \(M\): the first-order conditions show that \(u,v\) form a singular-vector pair of \(M\), while the second-order conditions provide a \emph{spectral gap} on singular values of $M$.
Looking ahead, we will prove
\[M = \sigma\, uv^\top + E,\]
for some singular value $\sigma>0$, and the residual term $E$ with a small total Frobenius mass relative to $\sigma$. Since $M\in W$, this would establish that $\|\Pi_W(uv^\top)\|_{\mathrm F}^2$ is large.

\paragraph{Singular-vector pair from first-order optimality.}
To relate \(M\) to the rank-one direction \(uv^\top\), expand
\[
xy^\top
=
st\,uv^\top+t\,\xi v^\top+s\,u\zeta^\top+\xi\zeta^\top.
\]
Multiplying $st$ on both sides and applying \(\wtE_w\), we obtain
% \[
% M=a\,uv^\top+m_xv^\top+u m_y^\top+c,
% \]
% where
% \[
% a=\wtE_w s^2t^2=\frac{A}{Z},
% \qquad
% m_x=\wtE_w[st^2\xi],
% \qquad
% m_y=\wtE_w[s^2t\zeta],
% \qquad
% c=\wtE_w[st\,\xi\zeta^\top].
% \]
\[
M=\wtE_w [s^2t^2 \,uv^\top + st^2 \,\xi v^\top +  s^2t\, u\zeta^\top + st\,\xi\zeta^\top],
\]
% The first-order optimality conditions eliminate the two terms in the middle in this expansion.  For every unit \(h\perp u\), write \(r=\langle x,h\rangle=\langle \xi,h\rangle\).
% by \eqref{eq:x-first-derivative},
% \[
% 0=
% 2k\,\wtE s^{2k-1}t^{2k}\langle x, h\rangle = 2k\wtE_w st^2 \langle\xi,h\rangle.
% \]
% Dividing by \(Z\) yields \(\langle m_x,h\rangle=0\) for every \(h\perp u\), and hence \(m_x=0\).  The identical argument for the second sphere variable, using \eqref{eq:y-first-derivative}, gives \(m_y=0\).  Thus
% \[
% M=a\,uv^\top+c .
% \]
The first-order optimality conditions eliminate the two terms in the middle in this expansion.  Indeed, for every unit \(h\perp u\), by \eqref{eq:x-first-derivative},
\[
2k\,\wtE s^{2k-1}t^{2k}\langle x, h\rangle = 0~ \implies~ \wtE_w st^2 \langle\xi,h\rangle = 0.
\]
Thus \(\wtE_w[st^2\xi]\in u^\perp\) has zero inner product with every
\(h\perp u\), and therefore vanishes.
The identical argument, using \eqref{eq:y-first-derivative}, kills the third term $\wtE_w[s^2tu\zeta^\top]$.  Thus $M$ simplifies to the alluded form
\begin{equation}\label{eq:M-top-SVD}
M=\sigma\,uv^\top+E,
\end{equation}
where $\sigma = \wtE_w [s^2t^2] > 0$, and $E= \wtE_w[st\xi\zeta^\top].$

Equivalently,
\[
Mv=\sigma u,\qquad M^\top u=\sigma v,
\]
so \(u,v\) form a singular-vector pair of \(M\).

\medskip
\paragraph{Bounding the remaining singular mass from second-order optimality.} It remains to bound the Frobenius mass from the residual term \(E\) relative to \(\sigma\) in~\eqref{eq:M-top-SVD}.  For this purpose, introduce the reweighed pseudo-covariance matrices
\[
X=\wtE_w[t^2\xi\xi^\top],
\qquad
Y=\wtE_w[s^2\zeta\zeta^\top].
\]
The second-order optimality condition for the first sphere variable and \eqref{eq:x-second-derivative} give, for every unit \(h\perp u\),
\begin{align*}
0
\ge
2k(2k-1)\wtE s^{2k-2}&t^{2k}\langle x,h\rangle^2
-
2k\wtE s^{2k}t^{2k}
\\
\iff\quad & (2k-1)\wtE_w t^{2}\langle x,h\rangle^2 \le \wtE_w s^{2}t^{2}.
\end{align*}
Using \(\langle x,h\rangle=\langle \xi,h\rangle\) and that $\sigma=\wtE_w s^2t^2$ derived above gives
\[
h^\top Xh\le \frac{\sigma}{2k-1}
\qquad
\text{for every unit }h\perp u.
\]
Since \(\xi\in u^\perp\), the matrix \(X\) has range contained in \(u^\perp\).  Therefore, as an operator on \(u^\perp\),
\[
X\preceq \frac{\sigma}{2k-1}I_{u^\perp}.
\]
Since \(X\succeq 0\), it follows that
\[
\|X\|_{\mathrm F}^2
\le
\|X\|_{\operatorname{op}}\Tr X
\le
\frac{\sigma}{2k-1}\cdot \frac{(m-1)\sigma}{2k-1}
=
\frac{(m-1)\sigma^2}{(2k-1)^2}.
\]
Thus
\[
\|X\|_{\mathrm F}
\le
\frac{\sqrt{m-1}\,\sigma}{2k-1}.
\]
The same argument yields
\[
\|Y\|_{\mathrm F}
\le
\frac{\sqrt{n-1}\,\sigma}{2k-1}.
\]

The Frobenius mass of the residual matrix \(E\) can be bounded by \(X\) and \(Y\) through a block-PSD argument. In particular, note that the block matrix
\[
\begin{pmatrix}
X & E\\
E^\top & Y
\end{pmatrix} \succeq 0.
\]
This is because for arbitrary \(h\in u^\perp\) and \(\ell\in v^\perp\),
\[
h^\top X h+2h^\top E\ell+\ell^\top Y\ell
=
\wtE_w
\left(
t\langle \xi,h\rangle
+
s\langle \zeta,\ell\rangle
\right)^2
\ge 0 .
\]
By the Cauchy–Schwarz inequality for block PSD matrices (Lemma~\ref{lem:block-cauchy}),
\[
\|E\|_{\mathrm F}^2
\le
\|X\|_{\mathrm F}\|Y\|_{\mathrm F},
\]
and hence
\[
\|E\|_{\mathrm F}
\le
\frac{((m-1)(n-1))^{1/4}\sigma}{2k-1}.
\]

\medskip
\paragraph{Finishing the proof.} The residual estimate now gives the required distance bound.  Since \(M\in W\) and \(W\) is linear, \(M/\sigma\in W\).  Therefore
\[
\distF(uv^\top,W)
\le
\left\|uv^\top-\frac{M}{\sigma}\right\|_{\mathrm F}
=
\frac{\|E\|_{\mathrm F}}{\sigma}
\le
\frac{((m-1)(n-1))^{1/4}}{2k-1}.
\]
Since \(\|uv^\top\|_{\mathrm F}=1\), we then have
\[
\|\Pi_W(uv^\top)\|_{\mathrm F}^2
=
1-\distF(uv^\top,W)^2\ge
1-
\frac{\sqrt{(m-1)(n-1)}}{(2k-1)^2}. \qedhere
\]
\end{proof}

\subsection{Algorithmic consequence}

We now derive the formal Yes/No guarantee for Algorithm~\ref{alg:sos-bss}.

\begin{corollary}[Yes/No guarantee]
\label{cor:algorithm-guarantee}
Let \(W\subseteq \R^{m\times n}\), let \(\epsilon\in(0,1)\), and choose \(k\) satisfying \eqref{eq:kcondition}.  Algorithm~\ref{alg:sos-bss} accepts every \(\Yes\)-instance of \(\textsc{RankOneSubspace}_{m,n}(W,\epsilon)\) and rejects every \(\No\)-instance.  Equivalently, after applying the reduction in Proposition~\ref{prop:bss-subspace-reduction}, it distinguishes perfect-completeness BSS instances with separable value \(1\) from those with separable value at most \(1-\epsilon\).

The SoS relaxation degree is
\[
4k=O\!\left(\frac{((m-1)(n-1))^{1/4}}{\sqrt{\epsilon}}\right),
\]
and the running time is \((m+n)^{O(k)}\).  For balanced local dimensions \(m=n=D\), the degree is
\[
O\!\left(\frac{\sqrt{D}}{\sqrt{\epsilon}}\right),
\]
and the running time is \(D^{O(\sqrt {D/\epsilon})}\).
\end{corollary}

\begin{proof}
Completeness is immediate.  If \(W\) is a \(\Yes\)-instance, then there exist
\(x_0\in \Sph^{m-1}\) and \(y_0\in \Sph^{n-1}\) such that
\(x_0y_0^\top\in W\).  Define a linear functional on polynomials of degree at
most \(4k\) by
\[
\wtE p := p(x_0,y_0) \text{ for every }p\in \R[x,y]_{\le 4k}.
\]
This is a degree-\(4k\) pseudo-expectation satisfying all constraints.
Thus the SDP in Algorithm~\ref{alg:sos-bss} is feasible, and the algorithm
returns \(\Yes\).

For soundness, suppose toward contradiction that \(W\) is a \(\No\)-instance and that the SDP in Algorithm~\ref{alg:sos-bss} is feasible.  Let \(\wtE\) be a feasible degree-\(4k\) pseudo-expectation.  By Lemma~\ref{lem:positive-max}, the associated moment polynomial \(\Phi_k\) has a positive global maximizer \((u,v)\).  This pair is a coordinate-wise second-order local maximizer, so the derivative identities \eqref{eq:x-first-derivative}--\eqref{eq:y-second-derivative} apply.  Applying Theorem~\ref{thm:main} gives
\[
\|\Pi_W(uv^\top)\|_{\mathrm F}^2
\ge
1-
\frac{R_{m,n}^2}{(2k-1)^2}.
\]
The choice of \(k\) in \eqref{eq:kcondition} implies
\[
\frac{R_{m,n}^2}{(2k-1)^2}
\le
\frac{\epsilon}{2}.
\]
Therefore
\[
\|\Pi_W(uv^\top)\|_{\mathrm F}^2
\ge
1-\frac{\epsilon}{2},
\]
contradicting the \(\No\)-condition
\[
\|\Pi_W(xy^\top)\|_{\mathrm F}^2\le 1-\epsilon
\qquad
\text{for every }x\in \Sph^{m-1},\;y\in \Sph^{n-1}.
\]
Thus the SDP is infeasible on every \(\No\)-instance, and Algorithm~\ref{alg:sos-bss} returns \(\No\).

The degree and running-time bounds are those established in the running-time analysis following Algorithm~\ref{alg:sos-bss}.  The BSS statement follows from Proposition~\ref{prop:bss-subspace-reduction}.
\end{proof}

\subsection{Algorithmic rounding}
\label{subsec:bss-algorithmic-rounding}
A surprising byproduct of our analysis is a constructive rounding procedure
for the \(\Yes\) case.  The rounding argument does not require computing a
global maximizer of \(\Phi_k\); it only uses a positive approximate coordinate-wise second-order local maximizer.  This
is a much weaker requirement and can be obtained by standard numerical
optimization primitives on the product sphere.  For instance, one may apply the
Riemannian trust-region method of Boumal--Absil--Cartis~\cite{BAC18}, obtaining the following theorem, whose proof we defer to~\Cref{app:bss-algorithmic-rounding}.
% This subsection turns the preceding \(\Yes\)/\(\No\) distinguisher into a constructive
% rounding procedure in the \(\Yes\) case.  Given a feasible degree-\(4k\)
% pseudo-expectation, the algorithmic task is to find an explicit rank-one matrix
% \(\widehat X=\widehat u\widehat v^\top\) with large projection onto \(W\).
% For the rounding step, it suffices to find a positive approximate second-order local maximizer of
% \[
%     \Phi_k(p,q)=\wtE\langle x,p\rangle^{2k}\langle y,q\rangle^{2k}.
% \]
% We initialize such a local search on the standard bases
% and then use a Riemannian trust-region method developed in \cite{BAC18} on the product sphere.

\begin{theorem}
\label{thm:bss-algorithmic-rounding}
Let \(\epsilon\in(0,1)\), and let \(k\) be the minimal integer satisfying
\eqref{eq:kcondition}.  Let \(\wtE\) be a feasible degree-\(4k\)
pseudo-expectation from Algorithm~\ref{alg:sos-bss}.  There is a deterministic
algorithm which, given the degree-\(4k\) moment vector of \(\wtE\), the
subspace \(W\subseteq\R^{m\times n}\), and \(\epsilon\), outputs
\[
    (\widehat u,\widehat v,\widehat X)
    \in
    \Sph^{m-1}\times\Sph^{n-1}\times\R^{m\times n}
\]
with \(\widehat X=\widehat u\widehat v^\top\) and
\begin{equation}
\label{eq:bss-algorithmic-rounding-projection}
    \|\Pi_W(\widehat X)\|_{\mathrm F}^2>1-\epsilon .
\end{equation}
The rounding time is polynomial in the size of the supplied moment table of \(\wtE\) and is
bounded, in the running-time used for Algorithm~\ref{alg:sos-bss}, by
\[
    (m+n+k)^{O(k)}
\]
arithmetic operations, up to polynomial factors in the input encoding length.
\end{theorem}

\subsection{ETH-tightness of the SoS degree}

% We end by comparing the degree scale above with the standard hardness scale arising from short-proof QMA(2) protocols.
We close the BSS section by explaining why the degree scale proved above is
essentially tight in the natural short-proof \(\QMA(2)\) regime under the exponential-time hypothesis.  This is
where the improved dependence on the gap parameter \(\epsilon\) becomes crucial:
for the inverse-linear gaps produced by logarithmic-proof \(\QMA(2)\)
protocols, our degree bound gives \(d=\widetilde O(D)\), and the corresponding
moment-SDP running time matches the ETH lower bound up to polylogarithmic
factors.  Thus, in this regime, a substantially smaller degree scale would
yield a subexponential-time algorithm for 3-SAT.

\paragraph{From $\QMA(2)$ verification to $\HSEP$.}
By \Cref{subsec:quantum-prelim}, a two-proof verifier with local proof dimension \(D\), completeness \(c\), soundness \(s\), and gap \(\Delta=c-s\) induces an accept operator \(0\preceq M\preceq I\) on \(\C^D\otimes\C^D\) whose optimum acceptance probability over unentangled proofs is \(h_\sep(M)\).  Hence yes-instances satisfy \(h_\sep(M)\ge c\), no-instances satisfy \(h_\sep(M)\le s\), and any additive approximation to \(h_\sep(M)\) with error smaller than a fixed constant multiple of \(\Delta\) distinguishes the two cases.

\paragraph{Short-proof $\QMA(2)$ protocols for 3-SAT.}
The logarithmic-proof construction of \cite{LGNN12} gives a two-proof QMA protocol for 3-SAT with perfect completeness and inverse-linear soundness gap up to polylogarithmic factors.  For a 3-SAT instance of size \(L\), the induced local proof dimension satisfies
\[
D=\Thetat(L),
\]
and the completeness-soundness gap satisfies
\[
\Delta=\Thetat(1/D).
\]
Under ETH, a distinguisher for the corresponding perfect-completeness $\HSEP$ instances, equivalently their BSS formulation, running in time
\[
2^{o(D/\polylog D)}
\]
would yield a \(2^{o(L)}\)-time algorithm for 3-SAT, contradicting the hypothesis.

\paragraph{Calibration of the SoS degree.}
For balanced local dimensions \(m,n=\Theta(D)\), Corollary~\ref{cor:algorithm-guarantee} uses degree
\[
d=O\!\left(\frac{\sqrt{D}}{\sqrt{\epsilon}}\right).
\]
At the inverse-linear gap \(\epsilon=\Thetat(1/D)\), this becomes
\[
d=\Ot(D),
\]
and the associated moment-SDP running time is \(D^{O(d)}=2^{O(d\log D)}\).  Combining this running-time form with the ETH lower bound above shows that, in this gap regime, any \(D^{O(d)}\)-time SoS family distinguishing these instances must satisfy \(d=\Omegat(D)\).  Hence the present degree scale is optimal up to polylogarithmic factors.

\bigskip
\begin{remark}
The above optimality statement is specific to the inverse-linear gap regime
\(\epsilon=\Thetat(1/D)\).  It does not preclude a different tradeoff between the
local dimension \(D\) and the precision parameter \(\epsilon\).
For example, there could be other tradeoffs between local dimension and the approximation precision $\epsilon$: currently nothing prevents a quasi-polynomial time algorithm for $\epsilon=\Omega(1)$~\cite{HM13,JLW26}.
\end{remark}
% ============================================================
% Part 2: Projector 2-to-4 Norm
% ============================================================

\section{An SoS relaxation for the projector \texorpdfstring{\(2\!\to\!4\)}{2-to-4} problem}
\label{sec:sos-2to4}
This section develops the SoS relaxation used to approximate the
projector \(2\!\to\!4\) norm.
We first reformulate the problem as maximizing
the fourth moment over the unit sphere in \(\operatorname{im}(P)\), then
introduce a fixed-threshold degree-\(2k\) SoS feasibility relaxation and use binary search
over the threshold to approximate the SoS value.

The main part of this section is the rounding analysis.  We prove that, for
sufficiently large degree, every feasible pseudo-expectation can be rounded to an
actual vector in \(\operatorname{im}(P)\).
%The rounding uses the spherical directional moment polynomial, and the argmax principle, and yields the stated relative approximation guarantee.
The analysis continues the argmax-principle viewpoint from
the BSS section, but uses it in a stronger form.  There, the selected maximizer was used merely through its first- and second-order optimality conditions.
Here, the $\ell_4$-norm objective requires consequences of global maximality itself: the maximizing direction yields higher-moment pseudo-expectation inequalities, including the fourth-moment estimate~\eqref{eq:fourth-order-argmax-comparison} that drives the rounding argument.

% The main part of the analysis proves that, for sufficiently large degree,
% every feasible pseudo-expectation can be rounded to an actual vector in
% \(\operatorname{im}(P)\) using a new argmax-principle-guided rounding procedure
% based on the spherical directional moment polynomial, which yields the stated
% relative approximation guarantee.

\subsection{Problem formulation}

Let \(P\in\R^{n\times n}\) be a nonzero orthogonal projector, and let
\[
    W=\operatorname{im}(P).
\]
For \(x\in\R^n\), define
\[
    f(x): \R^n\to \R,  \quad f: x\mapsto \sum_{i=1}^n x_i^4.
\]
The fourth power of the \(2\!\to\!4\) norm of \(P\) is
\begin{equation}
\label{eq:projector-2to4-opt}
    \OPT(P)
    :=
    \norm{P}_{2\to4}^4
    =
    \max_{\norm{y}_2=1}\norm{Py}_4^4
    =
    \max_{\substack{x\in W\\ \norm{x}_2=1}} f(x).
\end{equation}
Indeed, \(Py\in W\) and \(\norm{Py}_2\le1\) for every unit vector \(y\),
whereas every unit vector \(x\in W\) satisfies \(Px=x\). The last equality
then follows from the homogeneity of \(f\).

Throughout this section, we approximate the quartic objective
\(\OPT(P)\). Taking fourth roots yields the corresponding approximation to
\(\norm{P}_{2\to4}\).

\subsection{The fixed-threshold SoS relaxation}

Fix an integer \(k\ge2\) and a threshold \(\tau\in[0,1]\). Let
\(\mathcal{F}_{2k}(P,\tau)\) denote the set of degree-\(2k\)
pseudo-expectations \(\wtE\) satisfying
\begin{subequations}
\label{eq:fixed-threshold-relaxation}
\begin{align}
    \wtE\!\left[
        p(x)\bigl(\norm{x}_2^2-1\bigr)
    \right]
    &=0
    &&\text{for every \(p\) with \(\deg p\le 2k-2\),}
    \label{eq:sphere-constraint}\\
    \wtE\!\left[
        p(x)\bigl((I-P)x\bigr)_i
    \right]
    &=0
    &&\text{for every \(p\) with \(\deg p\le 2k-1\) and every
    \(i\in\{1,\ldots,n\}\),}
    \label{eq:subspace-constraint}\\
    \wtE\!\left[
        q(x)^2\bigl(f(x)-\tau\bigr)
    \right]
    &\ge0
    &&\text{for every \(q\) with \(\deg q\le k-2\).}
    \label{eq:threshold-constraint}
\end{align}
\end{subequations}
 For each fixed threshold
\(\tau\), the constraints in \eqref{eq:fixed-threshold-relaxation} are
affine in the degree-\(2k\) pseudo-moments and therefore define a
semidefinite feasibility problem.

Define the value of the degree-\(2k\) relaxation by
\begin{equation}
\label{eq:sos-value}
    \SOS_{2k}(P)
    :=
    \sup\left\{
        \tau\in[0,1]:
        \mathcal{F}_{2k}(P,\tau)\neq\varnothing
    \right\}.
\end{equation}
The feasible thresholds are downward closed. Indeed, suppose that
\(\wtE\in\mathcal{F}_{2k}(P,\tau)\) and \(0\le\tau'\le\tau\). Then, for
every polynomial \(q\) with \(\deg q\le k-2\),
\begin{align}
\label{eq:threshold-monotonicity}
    \wtE\!\left[
        q(x)^2\bigl(f(x)-\tau'\bigr)
    \right]
    &=
    \wtE\!\left[
        q(x)^2\bigl(f(x)-\tau\bigr)
    \right]
    +(\tau-\tau')\wtE q(x)^2
    \ge0.
\end{align}
Consequently, $\wtE\in\mathcal{F}_{2k}(P,\tau')$.
\subsection{The projector \texorpdfstring{\(2\!\to\!4\)}{2-to-4} norm approximation algorithm}

Algorithm~\ref{alg:sos-2to4} uses the fact that feasible thresholds remain
feasible after decreasing the threshold to compute a relative approximation to
\(\SOS_{2k}(P)\). Given a tolerance \(\varepsilon\in(0,1)\), it searches
over the interval \([1/n,1]\), maintains a feasible lower endpoint together
with a corresponding pseudo-expectation, and stops when the upper endpoint is
at most \(1+\varepsilon\) times the feasible lower endpoint. The validity of
the initial interval and the binary-search invariant is proved as part of
Theorem~\ref{thm:sos-main}.

\begin{algorithm}[H]
\caption{SoS approximation for the projector \(2\!\to\!4\) problem}
\label{alg:sos-2to4}
\begin{algorithmic}[1]
\Statex \textbf{Input:} A nonzero orthogonal projector \(P\), an integer
\(k\ge2\), and a tolerance \(\varepsilon\in(0,1)\).
\Statex \textbf{Output:} A threshold \(\widehat{\tau}\) and a degree-\(2k\)
pseudo-expectation \(\wtE^{\mathrm{out}}\).

\State Choose \(i_0\in\{1,\ldots,n\}\) such that \(P_{i_0i_0}>0\), and set
\(v_0\gets Pe_{i_0}/\sqrt{P_{i_0i_0}}\),
\(\tau_{\mathrm{lo}}\gets1/n\), \(\tau_{\mathrm{hi}}\gets1\), and
\(\wtE^{\mathrm{best}}\gets\wtE_{v_0}\), where
\(\wtE_{v_0}[p]:=p(v_0)\) for every polynomial \(p\) of degree at most
\(2k\).

\While{\(\tau_{\mathrm{hi}}>(1+\varepsilon)\tau_{\mathrm{lo}}\)}
    \State
    \[
        \tau_{\mathrm{mid}}
        \gets
        \frac{\tau_{\mathrm{lo}}+\tau_{\mathrm{hi}}}{2}.
    \]
    \If{\(\mathcal{F}_{2k}(P,\tau_{\mathrm{mid}})
        \neq\varnothing\)}
        \State Choose
        \[
            \wtE_{\mathrm{mid}}
            \in
            \mathcal{F}_{2k}(P,\tau_{\mathrm{mid}}).
        \]
        \State
        \[
            \tau_{\mathrm{lo}}\gets\tau_{\mathrm{mid}},
            \qquad
            \wtE^{\mathrm{best}}\gets\wtE_{\mathrm{mid}}.
        \]
    \Else
        \State
        \(\tau_{\mathrm{hi}}\gets\tau_{\mathrm{mid}}\).
    \EndIf
\EndWhile

\State
\[
    \widehat{\tau}\gets\tau_{\mathrm{lo}},
    \qquad
    \wtE^{\mathrm{out}}\gets\wtE^{\mathrm{best}}.
\]
\State \Return
\(\bigl(\widehat{\tau},\wtE^{\mathrm{out}}\bigr)\).
\end{algorithmic}
\end{algorithm}

\paragraph{Running time.}
The length of the search interval is halved at every iteration. Since
\(\tau_{\mathrm{lo}}\ge1/n\), the stopping condition is reached after
\(O(\log(n/\varepsilon))\) semidefinite feasibility calls.

The pseudo-moment matrix of a degree-\(2k\) pseudo-expectation in \(n\) variables
has dimension
\(
    \binom{n+k}{k},
\)
and the matrix associated with the threshold constraints
\eqref{eq:threshold-constraint} has dimension
\(
    \binom{n+k-2}{k-2}.
\)
Consequently, each feasibility problem has description size \(n^{O(k)}\).
Up to the standard polynomial dependence on the input encoding length and
the numerical precision of the SDP solver, the total running time of
Algorithm~\ref{alg:sos-2to4} is
\[
    n^{O(k)}.
\]
For
\[
    k=\left\lceil C_0\frac{\sqrt n}{\varepsilon}\right\rceil+1,
\]
where \(C_0\) is the universal constant in
Theorem~\ref{thm:sos-main}, the running time is
\[
    n^{O(\sqrt n/\varepsilon)}
    =
    2^{O((\sqrt n/\varepsilon)\log n)}.
\]

The algorithm above approximates the optimum of the degree-\(2k\) threshold
relaxation. The following theorem states that, with sufficiently large degree $k$, this relaxation is tight enough to yield a relative approximation to the true projector \(2\to4\) optimum.
\begin{theorem}[Approximation guarantee for the projector \(2\!\to\!4\) norm]
\label{thm:sos-main}
There exists a universal constant \(C_0>0\) such that the following holds.
Let \(P\in\R^{n\times n}\) be a nonzero orthogonal projector, let
\(0<\varepsilon<1\), and let \(k\ge2\) satisfy
\begin{equation}
\label{eq:degree-requirement}
    k-1\ge C_0\frac{\sqrt n}{\varepsilon}.
\end{equation}
Let \(\bigl(\widehat{\tau},\wtE^{\mathrm{out}}\bigr)\) be the output of
Algorithm~\ref{alg:sos-2to4} with tolerance \(\varepsilon\). Then
\begin{equation}
\label{eq:output-norm-approximation}
    (1+\varepsilon)^{-1/4}\norm{P}_{2\to4}
    \le
    \widehat{\tau}^{1/4}
    \le
    (1-\varepsilon)^{-1/4}\norm{P}_{2\to4}.
\end{equation}
\end{theorem}
\subsection{Spherical directional moment polynomial and rounding}
\label{subsec:sharp-one-shot-rounding}
For the soundness analysis of Theorem~\ref{thm:sos-main}, analogous to that of Theorem~\ref{thm:main}, we define the spherical directional moment polynomial and use the argmax-principle-guided reweighing and rounding step.
Fix \(\tau\in[0,1]\) and
\(\wtE\in\mathcal{F}_{2k}(P,\tau)\).

\paragraph{Setups.} Consider the spherical directional moment polynomial
\begin{equation}
\label{eq:projective-moment}
    \Phi_k(v):=\wtE\langle x,v\rangle^{2k},
    \qquad v\in \Sph(W).
\end{equation}
This is the restriction to the unit sphere \(\Sph(W)\) of a homogeneous degree-\(2k\) polynomial in the direction variable \(v\). It records the \(2k\)-th pseudo-moment of the projection of \(x\) along \(v\).
For a global maximizer \(u\in \Sph(W)\) of \(\Phi_k\), set
\begin{equation}
\label{eq:weighted-decomposition}
    s:=\langle x,u\rangle,
    \qquad
    Q:=P-uu^\top,
    \qquad
    z:=Qx=x-su,
    \qquad
    w(x):=s^{2k-4}.
\end{equation}
The identity \(z=x-su\) is understood modulo the subspace constraint
\(Px=x\). The weight \(w\) biases the reweighed pseudo-expectation toward points with
large overlaps with \(u\).  The following two lemmas are the one direction variants of what is established in the previous section, regarding first-order derivative information and the basic validity of reweighing by $w$.

\begin{lemma}[First-order stationarity of a projective maximizer]
\label{lem:first-order-stationarity}
Let \(u\in \Sph(W)\) be a global maximizer of \(\Phi_k\), and set
\(s:=\langle x,u\rangle\). Then, for every \(h\in W\cap u^\perp\),
\begin{equation}
\label{eq:first-order-stationarity}
    \wtE\!\left[s^{2k-1}\langle x,h\rangle\right]=0.
\end{equation}
\end{lemma}

\begin{proof}
The claim is trivial for \(h=0\), so fix \(h\in W\cap u^\perp\). Let
\(r:=\langle x,h\rangle\), and define
\[
    g(t):=
    \Phi_k\!\left(
        \frac{u+th}{\sqrt{1+t^2\norm{h}_2^2}}
    \right)
    =
    \wtE\!\left[
        \left(
        \frac{s+tr}{\sqrt{1+t^2\norm{h}_2^2}}
        \right)^{2k}
    \right].
\]
The curve inside \(\Phi_k\) lies in \(\Sph(W)\), and it passes through \(u\)
at \(t=0\). Since \(u\) is a global maximizer, \(g\) is stationary at
\(t=0\), so \(g'(0)=0\). Writing
\[
    y(t):=\frac{s+tr}{\sqrt{1+t^2\norm{h}_2^2}},
\]
we have \(y(0)=s\) and \(y'(0)=r\), because the derivative of
\((1+t^2\norm{h}_2^2)^{-1/2}\) vanishes at \(t=0\). Therefore
\[
    0=g'(0)
    =2k\,\wtE\!\left[y(0)^{2k-1}y'(0)\right]
    =2k\,\wtE\!\left[s^{2k-1}r\right].
\]
Dividing by \(2k\) proves \eqref{eq:first-order-stationarity}.
\end{proof}

\begin{lemma}[Positivity of the maximum and the reweighing]
\label{lem:positive-projective-moment}
Let \(u\in \Sph(W)\) be a global maximizer of \(\Phi_k\), and let \(s\) and
\(w\) be as in \eqref{eq:weighted-decomposition}. Then
\begin{equation}
\label{eq:positive-projective-moment}
    \wtE w(x)\ge\Phi_k(u)=\wtE s^{2k}>0.
\end{equation}
\end{lemma}

The proof of the above lemma is deferred to~\cref{sec:appendix-pos}. It justifies our reweighing
\begin{equation}
\label{eq:reweighted-functional}
    \wtE_{w}[g]
    :=
    \frac{\wtE[w(x)g(x)]}{\wtE w(x)}
\end{equation}
%\pnote{I don't really mind denoting reweighing by superscript. My impression is that subscript is more common. If we need subscript for something else, we can stick to superscript; otherwise subscript? check...}
is well defined for every polynomial \(g\) of degree at most four.
% Under this reweighing, the next lemma says that every tangent direction
% \(h\in W\cap u^\perp\) is less visible, the fourth moment in direction \(h\) is at most
% \(9/(k-1)^2\) times the fourth moment in the maximizing direction \(u\).

\paragraph{Argmax-induced concentration.}
The new ingredient for the projector \(2\!\to\!4\) rounding analysis is a fourth-moment (pseudo) concentration estimate.
After reweighing around a maximizer \(u\), for every tangent direction \(h\in W\cap u^\perp\), the fourth moment in direction \(h\) is at most
\(9/(k-1)^2\) times the fourth moment in the maximizing direction \(u\).
This is the one-sphere analogue of the spectral concentration around the top singular-vector pair used in the previous section.

% Under this reweighing, the reweighed pseudo-expectation ``concentrates'' around the maximizing direction \(u\).  The next lemma gives the fourth-moment
% instance of this concentration: for every tangent direction
% \(h\in W\cap u^\perp\), the fourth moment in direction \(h\) is at most
% \(9/(k-1)^2\) times the fourth moment in the maximizing direction \(u\).

\begin{lemma}[Argmax-induced concentration]
\label{lem:fourth-order-tangent}
Let \(u\in \Sph(W)\) be a global maximizer of \(\Phi_k\), and let \(s\) and
\(w\) be as in \eqref{eq:weighted-decomposition}. Then, for every
\(h\in W\cap u^\perp\),
\begin{equation}
\label{eq:fourth-order-tangent-bound}
    \wtE_w\!\left[\langle x,h\rangle^4\right]
    \le
    \frac{9}{(k-1)^2}
    \wtE_w\!\left[\langle x,u\rangle^4\right]
    \norm{h}_2^4.
\end{equation}
\end{lemma}

\begin{proof}
The assertion is trivial for \(h=0\), and homogeneity reduces the proof
to \(\norm{h}_2=1\). Set
\(
    r:=\langle x,h\rangle.
\)
Since \(h\perp u\), for every \(t\in\R\) the vectors
\[
    v_\pm(t):=\frac{u\pm th}{\sqrt{1+t^2}}
\]
belong to \(\Sph(W)\). Global maximality of \(u\) therefore gives
\begin{equation}
\label{eq:global-maximality-perturbation}
    \wtE(s\pm tr)^{2k}
    \le
    \wtE s^{2k}(1+t^2)^k.
\end{equation}
Adding the two inequalities in
\eqref{eq:global-maximality-perturbation}, dividing by two, and expanding
the even powers gives
\begin{equation}
\label{eq:even-binomial-expansion}
    \sum_{j=0}^k
    \binom{2k}{2j}t^{2j}
    \wtE\!\left[s^{2k-2j}r^{2j}\right]
    \le
    \wtE s^{2k}(1+t^2)^k.
\end{equation}
Every pseudo-moment on the left-hand side is nonnegative since \( \wtE s^{2k-2j}r^{2j}=\wtE(s^{k-j}r^j)^2 \geq 0\). Retaining only the term \(j=2\) in
\eqref{eq:even-binomial-expansion} and choosing \(t^2=1/k\) yields
\begin{equation}
\label{eq:raw-fourth-order-tangent-bound}
    \wtE\!\left[w(x)r^4\right]
    \le
    \frac{k^2(1+1/k)^k}{\binom{2k}{4}}
    \wtE s^{2k}.
\end{equation}
For \(k=2\), the coefficient on the right-hand side equals \(9\). If
\(k\ge3\), then \((1+1/k)^k<3\), and
\[
    \binom{2k}{4}
    =
    \frac{k(k-1)(2k-1)(2k-3)}{6}
    \ge
    \frac{k^2(k-1)^2}{3}.
\] Hence
\[
    \frac{k^2(1+1/k)^k}{\binom{2k}{4}}
    \le
    \frac{9}{(k-1)^2}.
\]
Finally,
\[
    \wtE s^{2k}
    =
    \wtE\!\left[w(x)s^4\right]
    =
    \wtE\!\left[w(x)\langle x,u\rangle^4\right].
\]
This proves \eqref{eq:fourth-order-tangent-bound} for unit \(h\) after normalization; the
general case follows by scaling.
\end{proof}

The next theorem is the main rounding estimate used in the soundness
argument. It shows that any feasible threshold pseudo-expectation can be
rounded, to the maximizer \(u\) of the spherical directional moment polynomial, whose
\(\ell_4\) norm is within an explicit tangent-error term of
\(\tau^{1/4}\). The proof reweighs the pseudo-expectation by
\(w(x)=\langle x, u\rangle^{2k-4}\), decomposes \(x=su+z\) into the maximizing direction and its orthogonal tangent component, and then uses the first-order optimality to cancel terms. The error term involves \(\sum_i P_{ii}^2\) arising by applying the preceding tangent fourth-order concentration bound Lemma \ref{lem:fourth-order-tangent} coordinate-wise to the
tangent directions obtained by projecting the coordinate directions onto the tangent space \(W\cap u^\perp\) at $u$, and then summing the
resulting estimates.
\begin{theorem}[Argmax rounding for projective $2\to4$ norm]
\label{thm:sharp-one-shot-rounding}
Assume that \(\tau>0\) and
\(\wtE\in\mathcal{F}_{2k}(P,\tau)\). Let
\[
    u
    \in
    \operatorname*{arg\,max}_{v\in \Sph(W)}\Phi_k(v).
\]
Then the following bound holds with the universal constant \(C_1=24\):
\begin{equation}
\label{eq:sharp-one-shot-bound}
    \tau^{1/4}
    \le
    \norm{u}_4
    +
    C_1
    \min\left\{
        \left(
            \frac{9}{(k-1)^2}\sum_{i=1}^nP_{ii}^2
        \right)^{1/4}\!,~~
        \tau^{-1/4}\left(\frac{9}{(k-1)^2}\sum_{i=1}^nP_{ii}^2
            \right)^{1/2}
    \right\}.
\end{equation}
\end{theorem}
\begin{proof}
For the pseudo-random vector \(x\), decompose it into its component along
\(u\) and its tangent component:
\[
    s:=\langle x,u\rangle,\qquad
    Q:=P-uu^\top,\qquad
    z:=Qx=x-su.
\]
Reweigh $\wtE$ by $w(x)=\langle u,x\rangle^{2k-4}$, which is well-defined by Lemma~\ref{lem:positive-projective-moment}, obtaining a degree-4 pseudo-expectation $\wtE_w$. Note further by the same Lemma~\ref{lem:positive-projective-moment},
\[
    \wtE_w s^4 = \frac{\wtE s^{2k}}{\wtE w(x)} \le 1.
\]
To lower bound $f(u)=\|u\|_4^4$, our proof starts with the following quartic SoS inequality (see Lemma~\ref{lem:centered-quartic-sos}),
\begin{equation}
    (a+b)^4
    \le
    (1+\eta)a^4+4a^3b+C_\eta b^4,
\end{equation}
for any $\eta>0$, and a constant $C_\eta$ that depends on $\eta$.
Apply the above inequality coordinate-wise to
\(x_i=su_i+z_i\), summing over \(i\), and taking \(\wtE_w\), we get
\begin{align}
    \tau
    &\le \wtE_w f(x)
    \nonumber\\
    &\le
    (1+\eta)\wtE_w s^4 f(u) + 4\wtE_w\left[ s^3\sum_{i=1}^n u_i^3z_i \right]+ C_\eta \wtE_w f(z),
   % \label{eq:x-4-norm-decomposition}
    \nonumber\\
    &\le (1+\eta)f(u) + 4\wtE_w\left[ s^3\sum_{i=1}^n u_i^3z_i \right] + C_\eta  \wtE_w f(z)
    \nonumber
    \\
    &\le (1+\eta)f(u)
    + C_\eta  \wtE_w f(z)
    \nonumber\\
    &\le (1+\eta)f(u)
    + C_\eta \frac{9}{(k-1)^2} \sum_{i=1}^n P_{ii}^2
    \label{eq:f-u-bound}
\end{align}
where the first inequality holds due to our SoS constraint~\eqref{eq:threshold-constraint}; the third inequality follows from $\wtE_w s^4 \le 1$. We now elaborate the last two nontrivial steps.
First, we claim the term in the middle vanishes completely in the fourth step, using the first-order stationarity of the maximizer \(u\).  Let
\(u^{\odot3}:=(u_1^3,\ldots,u_n^3)\), and put
\(h=Qu^{\odot3}\in W\cap u^\perp\).  Since \(Q\) is symmetric and \(z=Qx\),
\[
    \langle x,h\rangle
    =
    \sum_{i=1}^n u_i^3 z_i.
\]
Therefore Lemma~\ref{lem:first-order-stationarity} gives
\[
    \wtE_w\left[
        s^3\sum_{i=1}^n u_i^3z_i
    \right]
    =
    \frac{1}{\wtE w(x)}
    \wtE\left[
        s^{2k-1}\sum_{i=1}^n u_i^3z_i
    \right]
    =0.
\]
Finally, the last step is due to Lemma~\ref{lem:fourth-order-tangent}. Concretely, since \(u\in W\) and \(\norm{u}_2=1\), the matrix \(Q\) is the orthogonal
projector onto \(W\cap u^\perp\).  For each \(i\in[n]\), set
\(h_i:=Qe_i\).  Then \(h_i\in W\cap u^\perp\),
\[
    z_i=\langle x,h_i\rangle,
    \qquad
    \norm{h_i}_2^2=Q_{ii}=P_{ii}-u_i^2.
\]
Applying Lemma~\ref{lem:fourth-order-tangent} to each \(h_i\), summing over
\(i\), and using \(0\le P_{ii}-u_i^2\le P_{ii}\), gives
\begin{equation*}
%\label{eq:sharp-reweighted-tangent-bound}
    \wtE_w f(z) = \wtE_w \norm{z}_4^4
    \le
    \frac{9\,\wtE_w \langle x,u\rangle ^4}{(k-1)^2} \sum_{i=1}^n P_{ii}^2.
\end{equation*}
Together with \(\wtE_w s^4\le1\), this gives the fifth inequality in
\eqref{eq:f-u-bound}.
That completes the conceptual part of our proof. To conclude, we get that for every $\eta>0$,
\begin{equation}
\label{eq:sharp-internal-estimate}
    \tau
    \le
    (1+\eta)\norm u _4^4
    +
    C_\eta
    \frac{9}{(k-1)^2}\sum_{i=1}^n P_{ii}^2.
\end{equation}

The remaining work is to optimize over the parameter $\eta$.
Abbreviate
\[
    a:=\norm{u}_4,
    \qquad
    b:=
    \left(
        \frac{9}{(k-1)^2}\sum_{i=1}^nP_{ii}^2
    \right)^{1/4}.
\]
Since $P$ is a non-zero and $P\succeq 0$, we have $b>0$.  If
\(b\ge\tau^{1/4}\), then
\(b^2/\tau^{1/4}\ge b\), so the minimum in
\eqref{eq:sharp-one-shot-bound} is equal to \(b\). Since \(C_1=24\ge1\),
\eqref{eq:sharp-one-shot-bound} follows in this case.
It remains to consider \(0<b<\tau^{1/4}\). If \(a\ge\tau^{1/4}\), then
\eqref{eq:sharp-one-shot-bound} is immediate. Hence assume
\(a<\tau^{1/4}\), and choose
\[
    \eta:=\frac{b^2}{\tau^{1/2}}\in(0,1).
\]
Using the explicit form of \(C_\eta\), together with \(a<\tau^{1/4}\)
and \(b<\tau^{1/4}\), gives
\[
\begin{aligned}
    \eta a^4+C_\eta b^4
    &=\eta a^4+
      \left(1+\frac{18}{\eta}+3\left(\frac{2}{\eta}\right)^{1/3}\right)b^4 \\
    &\le
    \left(1+1+18+3\cdot2^{1/3}\right)b^2\tau^{1/2}
    \le
    24\,b^2\tau^{1/2}.
\end{aligned}
\]
Thus \eqref{eq:sharp-internal-estimate} gives
\begin{equation*}
    (\tau^{1/4}-a)\tau^{3/4} \le \tau-a^4 \le 24b^2\tau^{1/2},
\end{equation*}
where the first inequality uses \(a<\tau^{1/4}\). Hence
\[
    \tau^{1/4}-a
    \le
    24\,\frac{b^2}{\tau^{1/4}}.
\]
Since \(0<b<\tau^{1/4}\), this is exactly the desired bound in the
case where the minimum equals \(b^2/\tau^{1/4}\).
% Hence \eqref{eq:sharp-internal-estimate} implies
% \[
%     \tau\le a^4+24\,b^2\tau^{1/2}.
% \]
% Since \(a<\tau^{1/4}\),
% \[
%     \tau-a^4
%     \ge
%     (\tau^{1/4}-a)\tau^{3/4}.
% \]
% Therefore
% \[
%     \tau^{1/4}-a
%     \le
%     24\,\frac{b^2}{\tau^{1/4}}.
% \]
% Since \(0<b<\tau^{1/4}\), one has
% \[
%     \min\left\{b,\frac{b^2}{\tau^{1/4}}\right\}
%     =
%     \frac{b^2}{\tau^{1/4}}.
% \]
% This proves \eqref{eq:sharp-one-shot-bound} with \(C_1=24\).
\end{proof}
The preceding theorem immediately gives the following rounding corollary: when the degree is sufficiently large, every feasible threshold pseudo-expectation yields an actual unit vector in \(W\) whose $\ell_4^4$ value is within a factor \(1-\varepsilon\) of the threshold.
\begin{corollary}[Rounding]
\label{cor:sharp-ideal-rounding}
There exists a universal constant \(C_0>0\) such that the following holds.
Fix \(\tau>0\), a pseudo-expectation
\(\wtE\in\mathcal{F}_{2k}(P,\tau)\), and \(0<\varepsilon<1\). If
\[
    k-1\ge C_0\frac{\sqrt n}{\varepsilon},
\]
then there exists a unit vector \(v\in W\) satisfying
\begin{equation}
\label{eq:sharp-ideal-vector-guarantee}
    \norm{v}_4^4\ge(1-\varepsilon)\tau.
\end{equation}
More explicitly, if there exists an index \(i\) such that
\(P_{ii}^2\ge(1-\varepsilon)\tau\), one may take
\(v=Pe_i/\sqrt{P_{ii}}\). If no such index exists, one may take
\(v\in
    \operatorname*{arg\,max}_{v\in \Sph(W)}\Phi_k(v).\)

\end{corollary}

\begin{proof}
Let \(C_1\) be the universal constant from
Theorem~\ref{thm:sharp-one-shot-rounding}, and choose
\(C_0\ge\max\{12C_1,3\}\). In the first case \(P_{ii}^2\ge(1-\varepsilon)\tau\), the vector
\(Pe_i/\sqrt{P_{ii}}\) lies in \(W\) and has unit Euclidean norm, and its
\(i\)-th coordinate is \(\sqrt{P_{ii}}\). Hence $\norm {P e_i / \sqrt P_{ii}}_4^4$ is at
least \(P_{ii}^2\ge(1-\varepsilon)\tau\).

Otherwise, let
\[
    u\in\operatorname*{arg\,max}_{v\in \Sph(W)}\Phi_k(v)
\]
be the vector specified in the statement. Then
\[
    \sum_{i=1}^nP_{ii}^2
    \le
    n\max_iP_{ii}^2
    <
    n\tau.
\]
Put
\[
    B:=
    \frac{9}{(k-1)^2}\sum_{i=1}^nP_{ii}^2.
\]
The degree assumption gives
\[
    \frac{3\sqrt n}{k-1}
    \le
    \min\left\{1,\frac{\varepsilon}{4C_1}\right\}.
\]
Consequently \(B<\tau\), and
\[
    \frac{B^{1/2}}{\tau^{1/4}}
    \le
    \frac{3\sqrt n}{k-1}\tau^{1/4}
    \le
    \frac{\varepsilon}{4C_1}\tau^{1/4}.
\]
Theorem~\ref{thm:sharp-one-shot-rounding} therefore gives
\[
    \norm{u}_4
    \ge
    \left(1-\frac{\varepsilon}{4}\right)\tau^{1/4}.
\]
Since \((1-\varepsilon/4)^4\ge1-\varepsilon\) for
\(0<\varepsilon<1\), this proves
\[
    \norm{u}_4^4\ge(1-\varepsilon)\tau.\qedhere
\]
\end{proof}

\subsection{Proof of the approximation guarantee}
\label{subsec:proof-sos-main}
We now complete the proof of the main approximation guarantee. The argument sandwiches the degree-\(2k\) SoS value between \(\OPT(P)\) and \((1-\varepsilon)^{-1}\OPT(P)\), and then obtain the relative bound for \(\widehat{\tau}^{1/4}\).
\begin{proof}[Proof of Theorem~\ref{thm:sos-main}]
We first prove the lower bound
\(
    \OPT(P)\le\SOS_{2k}(P).
\)
Since \(\Sph(W)\) is compact, there exists
\(x^\star\in \Sph(W)\) such that
\(
    f(x^\star)=\OPT(P).
\)
Define the evaluation pseudo-expectation
\[
    \wtE_{x^\star}[p]:=p(x^\star).
\]
\(\wtE_{x^\star}\) satisfies all constraints \eqref{eq:sphere-constraint},~\eqref{eq:subspace-constraint} and \eqref{eq:threshold-constraint}.
Thus \(\OPT(P)\) is a feasible threshold, proving
\(\OPT(P)\le\SOS_{2k}(P)\).

We next prove the upper bound on the SoS value. Fix an arbitrary
feasible threshold \(\tau\in[0,1]\) and choose
\(\wtE\in\mathcal{F}_{2k}(P,\tau)\). The case \(\tau=0\) is trivial, so
assume \(\tau>0\). By Corollary~\ref{cor:sharp-ideal-rounding}, the
degree condition \eqref{eq:degree-requirement} implies that there exists
a unit vector \(v\in W\) satisfying
\[
    f(v)\ge(1-\varepsilon)\tau.
\]
Since \(v\) is a feasible vector for \(\OPT(P)\), we have
\[
    \OPT(P)\ge f(v)\ge(1-\varepsilon)\tau.
\]
Equivalently,
\[
    \tau\le\frac{1}{1-\varepsilon}\OPT(P).
\]
Taking the supremum over all feasible thresholds \(\tau\) gives
\[
    \OPT(P)
    \le
    \SOS_{2k}(P)
    \le
    \frac{1}{1-\varepsilon}\OPT(P).
\]
This proves the SoS-value bound used below.

It remains to verify the initialization and output guarantee of
Algorithm~\ref{alg:sos-2to4}. Since \(P\neq0\), there exists an index
\(i_0\) with \(P_{i_0i_0}>0\). Set
\begin{equation}
\label{eq:initial-vector}
    v_0:=\frac{Pe_{i_0}}{\sqrt{P_{i_0i_0}}}.
\end{equation}
Since \(P^2=P\), the vector \(v_0\) belongs to \(W\) and satisfies
\(\norm{v_0}_2=1\). Moreover, Cauchy--Schwarz gives
\begin{equation}
\label{eq:initial-lower-bound}
    f(v_0)
    =
    \sum_{i=1}^n (v_0)_i^4
    \ge
    \frac1n
    \left(\sum_{i=1}^n (v_0)_i^2\right)^2
    =
    \frac1n.
\end{equation}
Hence the evaluation pseudo-expectation \(\wtE_{v_0}\) belongs to
\(\mathcal{F}_{2k}(P,1/n)\).

The sphere constraint also implies \(\wtE f(x)\le1\). Indeed,
\begin{equation}
\label{eq:fourth-moment-upper-certificate}
    1-f(x)
    =
    \bigl(1-\norm{x}_2^2\bigr)\bigl(1+\norm{x}_2^2\bigr)
    +2\sum_{1\le i<j\le n}x_i^2x_j^2.
\end{equation}
Combining \eqref{eq:fourth-moment-upper-certificate} with the threshold
constraint for \(q=1\) shows that every feasible threshold is at most one.
Therefore
\begin{equation}
\label{eq:sos-initial-interval}
    \frac1n\le\SOS_{2k}(P)\le1.
\end{equation}

At every iteration of Algorithm~\ref{alg:sos-2to4},
\(\tau_{\mathrm{lo}}\) is feasible and
\(\SOS_{2k}(P)\le\tau_{\mathrm{hi}}\). When the algorithm terminates,
\(\tau_{\mathrm{hi}}\le(1+\varepsilon)\tau_{\mathrm{lo}}\), and therefore
\begin{equation}
\label{eq:binary-search-guarantee}
    \wtE^{\mathrm{out}}
    \in
    \mathcal{F}_{2k}(P,\widehat{\tau}),
    \qquad
    \widehat{\tau}
    \le
    \SOS_{2k}(P)
    \le
    (1+\varepsilon)\widehat{\tau}.
\end{equation}

Combining \eqref{eq:binary-search-guarantee} with the SoS-value bound above gives
\[
    \widehat{\tau}
    \ge
    \frac{1}{1+\varepsilon}\SOS_{2k}(P)
    \ge
    \frac{1}{1+\varepsilon}\OPT(P),
\]
and
\[
    \widehat{\tau}
    \le
    \SOS_{2k}(P)
    \le
    \frac{1}{1-\varepsilon}\OPT(P).
\]
Using \(\OPT(P)=\norm{P}_{2\to4}^4\) and taking fourth roots gives
\[
    (1+\varepsilon)^{-1/4}\norm{P}_{2\to4}
    \le
    \widehat{\tau}^{1/4}
    \le
    (1-\varepsilon)^{-1/4}\norm{P}_{2\to4},
\]
which is exactly \eqref{eq:output-norm-approximation}.
\end{proof}
\subsection{Algorithmic rounding by finite search}
\label{subsec:sos-2to4-algorithmic-rounding}
The ideal rounding step in Corollary~\ref{cor:sharp-ideal-rounding}
uses a global maximizer of the spherical directional moment polynomial
\(\Phi_k(v)=\wtE\langle x,v\rangle^{2k}\).  We now replace this exact
maximization by a finite deterministic search.
Throughout this subsection, fix \(\tau>0\), \(0<\varepsilon<1\), and
\(\wtE\in\mathcal{F}_{2k}(P,\tau)\).  For a unit vector $u$ under discussion, we inherit the notations from the previous sections:
\[
    s=\langle x,u\rangle,\qquad
    w=s^{2k-4},\qquad
    Q=P-uu^\top,\qquad
    M=\Phi_k(u)=\wtE s^{2k},\qquad
    A=\wtE s^{2k-4}.
\]

% The key observation is that the~\Cref{lem:fourth-order-tangent}, \todo{why theorem, some bug to fix..} where the global maximality of $u$ is used in~\eqref{eq:global-maximality-perturbation}, is in fact used in a somewhat  conservative manner in our analysis: we don't need to apply it to an arbitrary direction $h\perp u$, but only the fixed $(P-uu^\top)e_i$ for all $e_i$; and .
% This suggests the following definition
%We thus isolate the approximate maximality property that the  search will certify.

\paragraph{Coordinate-circle maximizer.}
%The rounding analysis \cref{thm:sos-main} only uses the argmax-induced concentration (Lemma~\ref{lem:fourth-order-tangent}) through a much smaller family of two-dimensional comparisons:
Although Lemma~\ref{lem:fourth-order-tangent} is stated for every tangent
direction \(h\in W\cap u^\perp\), in the proof of the rounding theorem~\cref{thm:sos-main} it is
applied only to \((P-uu^\top)e_i\);  moreover, the binomial argument proving Lemma~\ref{lem:fourth-order-tangent} uses global maximality only along the circle spanned by \(u\) and \(h\) when comparing $u$ and the direction $h$.  Therefore, it suffices to require approximate local maximality on the circles
\(\Sph(W)\cap\operatorname{span}\{u,Pe_i\}\).  This motivates the following definition.
\begin{definition}[Coordinate-circle maximizer]
\label{def:approx-coordinate-circle-maximality}
A vector \(u\in \Sph(W)\) is a \(\rho\)-coordinate-circle maximizer of
\(\Phi_k\) if \(\Phi_k(u)>0\) and, for every \(i\in[n]\) and every unit vector
\[
    v\in \Sph(W)\cap\operatorname{span}\{u,Pe_i\},
\]
one has
\[
    \Phi_k(v)\le \rho\Phi_k(u).
\]
\end{definition}
%\todo{try to prove argmax-induced concentration here}
Adapting the proof of Lemma~\ref{lem:fourth-order-tangent} to the coordinate-circle maximizer, one obtains
\begin{proposition}[Weak concentration of $\rho$-coordinate-circle maximizer]
\label{prop:coordinate-circle-concentration}
Let \(u\in \Sph(W)\) be a \(\rho\)-coordinate-circle maximizer of \(\Phi_k\).
Then, for every \(i\in[n]\),
\[
    \wtE_w[\langle x,Qe_i\rangle^4]
    \le
    \frac{9\rho}{(k-1)^2}\,
    \wtE_w[\langle x,u\rangle^4] \cdot \|Qe_i\|_2^4.
\]
Consequently,
\begin{equation}
    \wtE_w \|Qx\|_4^4
    \le
    \frac{9\rho}{(k-1)^2}\,
    \wtE_w [\langle x,u\rangle^4] \cdot \sum_i P_{ii}^2.
\end{equation}
\end{proposition}

% Now search algorithm only optimizes
% \(\Phi_k\) over the coordinate-generated circles through the current vector,
% and it terminates when no such circle gives a significant multiplicative
% improvement.  This condition is strong enough to recover the two ingredients
% used in the ideal proof: a fourth-moment tangent bound and an approximate
% first-order cancellation.

\medskip
\paragraph{Discretizing the circle.}
Now we discretize the unit circle. Recall the well-known Bernstein inequality.
\begin{lemma}[Bernstein inequalities for trigonometric polynomials]
\label{lem:trig-bernstein}
For every real trigonometric polynomial \(p\) of degree at most \(D\), and for any $\theta$,
\[
    |p'(\theta)|^2
    \le
    D^2(\|p\|_\infty^2-p(\theta)^2)
\]
In particular,
\[
    \|p'\|_\infty\le D\|p\|_\infty.
\]
\end{lemma}

The following lemma makes each one-dimensional
circle search finite, which is a self-contained version of the Ehlich-Zeller's trigonometric grid~\cite{EZ64,BV12}.
\begin{lemma}[Ehlich-Zeller's trigonometric grid]
\label{lem:finite-trigonometric-grid}
% Let \(p:\R\to\R_{\ge0}\) be a nonnegative trigonometric polynomial of
% Fourier degree at most \(D\).  If
% \(0<\beta<1\),
% \[
%     N\ge \frac{\pi D(D+1)}{\beta},
%     \qquad
%     \Theta_N:=\{2\pi j/N:0\le j<N\},
% \]
% then
% \[
%     \max_{\theta\in\Theta_N}p(\theta)
%     \ge
%     (1-\beta)\max_\theta p(\theta).
% \]
Let \(p:\mathbb R\to\mathbb R_{\ge0}\) be a nonnegative trigonometric polynomial of degree at
most \(D\). If \(0<\beta<1\) and
\[
    N\ge \frac{\pi D}{\beta},
    \qquad
    \Theta_N:=\{2\pi j/N:0\le j<N\},
\]
then
\[
    \max_{\theta\in\Theta_N}p(\theta)
    \ge
    (1-\beta)\max_\theta p(\theta).
\]
\end{lemma}

\begin{proof}
Let \(\theta^\star\) maximize \(p\). Since $p\ge 0$, $p(\theta^\star)=\|p\|_\infty$.  Choose
\(\theta_0\in\Theta_N\) whose circular distance from \(\theta^\star\) is at
most \(\pi/N\).  By Bernstein's inequality for trigonometric polynomials,
\[
    \|p'\|_\infty\le D\|p\|_\infty.
\]
Therefore
\[
    p(\theta_0)
    \ge
    p(\theta^\star)-\|p'\|_\infty\frac{\pi}{N}
    \ge
    \|p\|_\infty-D\|p\|_\infty\frac{\pi}{N}
    \ge
    (1-\beta)\|p\|_\infty.\qedhere
\]
\end{proof}

% \begin{proof}
% Write \(p(\theta)=\sum_{m=-D}^Dc_me^{im\theta}\), and let \(M_p=\max_\theta p(\theta)\).  Since \(p\ge0\),
% \[
%     |c_m|
%     \le
%     \frac1{2\pi}\int_0^{2\pi}p(\theta)\,d\theta
%     =c_0
%     \le M_p.
% \]
% Therefore \(\norm{p'}_\infty\le D(D+1)M_p\).  If \(\theta^\star\) maximizes
% \(p\) and \(\theta_0\in\Theta_N\) is within circular distance at most
% \(\pi/N\) of \(\theta^\star\), then
% \[
%     p(\theta_0)
%     \ge
%     M_p-D(D+1)M_p\frac{\pi}{N}
%     \ge
%     (1-\beta)M_p.\qedhere
% \]
% \end{proof}

\medskip
\paragraph{Iterative search.} We next give the finite local search procedure.  Its initialization guarantees
a positive objective value, while each accepted update increases this
value by a fixed multiplicative factor. The finite search is stated with an
accuracy parameter \(\alpha\); the value of \(\alpha\) needed for rounding is
specified in Theorem~\ref{thm:projector-algorithmic-rounding}.

\begin{algorithm}[H]
\caption{Iterative coordinate-circle local search}
\label{alg:projective-local-search}
\begin{algorithmic}[1]
\Require The degree-\(2k\) moment vector of \(\wtE\), the projector \(P\),
         and an accuracy parameter \(0<\alpha<1\).
\State Let \(r=\operatorname{rank}(P)\), choose an orthonormal basis
\(b_1,\ldots,b_r\) of \(W\), and set \(u\gets b_j\), where \(j\) maximizes
\(\wtE\langle x,b_j\rangle^2\).
\State Set
\[
    D\gets2k,
    \qquad
    N\gets\left\lceil\frac{\pi D}{\alpha}\right\rceil,
    \qquad
    \Theta_N\gets\{2\pi j/N:0\le j<N\}.
\]
\Repeat
    \State Set \(M\gets\Phi_k(u)\) and \(Q\gets P-uu^\top\).
    \State Form the finite candidate set
    \[
        \mathcal C_\alpha(u):=
        \left\{
        \cos\theta\,u+
        \sin\theta\,\frac{Qe_i}{\norm{Qe_i}_2}:
        Qe_i\neq0,\ \theta\in\Theta_N
        \right\}.
    \]
    \If{there exists \(w\in\mathcal C_\alpha(u)\) with
        \(\Phi_k(w)>(1+\alpha)M\)}
        \State Replace \(u\) by a vector in \(\mathcal C_\alpha(u)\)
        maximizing \(\Phi_k\) over \(\mathcal C_\alpha(u)\).
    \Else
        \State \Return \(u\).
    \EndIf
\Until{termination}
\end{algorithmic}
\end{algorithm}

\medskip
\paragraph{Analysis.}
Algorithm~\ref{alg:projective-local-search} terminates in a small number of steps.
\begin{proposition}
\label{prop:projective-local-search-termination}
Algorithm~\ref{alg:projective-local-search} terminates after at most
\[
    O\left(\frac{k\log r}{\alpha}\right)
\]
updates.
\end{proposition}

\begin{proof}
Use $\Phi_k$ as a potential function. First, we note that the initialization point $u$ in Step 1 from Algorithm~\ref{alg:projective-local-search} has not-too-small objective value.  The constraints
\(Px=x\) and \(\norm{x}_2^2=1\) imply
\(
    \sum_{j=1}^r\wtE\langle x,b_j\rangle^2=1,
\)
so the chosen basis vector satisfies
\(\wtE\langle x,u\rangle^2\ge1/r\).  By pseudo-Cauchy--Schwarz Lemma \ref{pseudo-Cauchy–Schwarz}, the even pseudo-moments
\(a_j:=\wtE\langle x,u\rangle^{2j}\) are log-convex:
\(a_j^2\le a_{j-1}a_{j+1}\). Hence \(a_k\ge a_1^k\), and therefore
\[
    \Phi_k(u)
    =\wtE\langle x,u\rangle^{2k}
    \ge
    \bigl(\wtE\langle x,u\rangle^2\bigr)^k
    \ge r^{-k}.
\]

Second, $\Phi_k$ is bounded above by 1: for every \(v\in \Sph(W)\), one has \(0\le\Phi_k(v)\le1\), because
\[
    1-\langle x,v\rangle^{2k}
    =
    \bigl(1-\langle x,v\rangle^2\bigr)
    \sum_{j=0}^{k-1}\langle x,v\rangle^{2j}
\]
is a sum of squares modulo \(\norm{x}_2^2=1\).

Each update increase the potential
\(\Phi_k\) by a multiplicative factor more than \(1+\alpha\), and hence the number of updates is
at most
\[
    \left\lceil
    \frac{k\log r}{\log(1+\alpha)}
    \right\rceil
    =O\left(\frac{k\log r}{\alpha}\right).\qedhere
\]
\end{proof}

For \(0<\alpha<1\), define
\begin{equation}
\label{eq:rho-delta-alpha-definitions}
    \rho_\alpha:=\frac{1+\alpha}{1-\alpha},
    \qquad
    \Delta_\alpha:=\sqrt{\rho_\alpha^2-1}=\frac{2\sqrt\alpha}{1-\alpha}.
\end{equation}
The next statement records the two consequences of termination that are used
in rounding: approximate maximality along all coordinate circles and a
corresponding approximate stationarity estimate.
\begin{proposition}
After termination, Algorithm~\ref{alg:projective-local-search} returns a \(\rho_\alpha\)-coordinate-circle maximizer \(u\).
% Moreover, if
% \[
%     s=\langle x,u\rangle,
%     \qquad
%     M=\Phi_k(u),
%     \qquad
%     Q=P-uu^\top,
% \]
Moreover, for every index \(i\),
% \begin{equation*}
%     \left|\wtE s^{2k-1}\sum_{i=1}^n u_i^3z_i\right| \le \Delta_\alpha M.
% \end{equation*}
\begin{equation}
\label{eq:approx-first-order-coordinate}
    \left|
    \wtE s^{2k-1}\langle x,Qe_i\rangle
    \right|
    \le
    \Delta_\alpha M\norm{Qe_i}_2.
\end{equation}
\end{proposition}

\begin{proof}

At termination, fix \(i\) with \(Qe_i\neq0\), and set
\(h_i=Qe_i/\norm{Qe_i}_2\).
%
% We first prove
% \begin{equation}
% \label{eq:approx-first-order-coordinate}
%     \left|
%     \wtE s^{2k-1}\langle x,Qe_i\rangle
%     \right|
%     \le
%     \Delta_\alpha M\norm{Qe_i}_2.
% \end{equation}
The function
\[
    p_i(\theta)=\Phi_k(\cos\theta\,u+\sin\theta\,h_i)
\]
is a nonnegative trigonometric polynomial of degree at most \(2k\).  By
Lemma~\ref{lem:finite-trigonometric-grid}, its grid maximum is at least
\((1-\alpha)\max_\theta p_i(\theta)\).  Since termination means that every
grid value is at most \((1+\alpha)\Phi_k(u)\),
\[
    \max_\theta p_i(\theta)
    \le
    \rho_\alpha\Phi_k(u).
\]
This proves the required coordinate-circle maximality for all indices with
\(Qe_i\neq0\).  If \(Qe_i=0\), then
\(\Sph(W)\cap\operatorname{span}\{u,Pe_i\}\subseteq\{u,-u\}\), and the same
claim is immediate since \(\Phi_k\) is even.

It remains to prove the approximate first-order estimate~\eqref{eq:approx-first-order-coordinate}.
% We use the Bernstein inequality for real trigonometric polynomials: if \(p\) has degree at most \(D\), then
% \[
%     |p'(\theta)|^2
%     \le
%     D^2\bigl(\norm{p}_\infty^2-p(\theta)^2\bigr).
% \]
Apply Bernstein inequality to \(p_i\) at \(\theta=0\), with \(D=2k\),
\(p_i(0)=M\), and \(\norm{p_i}_\infty\le\rho_\alpha M\), obtaining
\[
    |p_i'(0)|
    \le
    2k\Delta_\alpha M.
\]
On the other hand,
\[
    p_i'(0)
    =2k\,\wtE s^{2k-1}\langle x,h_i\rangle.
\]
Multiplying by \(\norm{Qe_i}_2\) gives
\eqref{eq:approx-first-order-coordinate}; if \(Qe_i=0\), the estimate is trivial.
%
% Therefore
% \[
% \begin{aligned}
%     \left|
%     \wtE s^{2k-1}\sum_{i=1}^n u_i^3z_i
%     \right|
%     &= \left| \wtE s^{2k-1}\sum_{i=1}^n u_i^3\langle x,Qe_i\rangle \right|\\
%     &\le
%     \sum_{i=1}^n |u_i|^3
%     \left|\wtE s^{2k-1}\langle x,Qe_i\rangle\right|  \\
%     &\le
%     \Delta_\alpha M
%     \sum_{i=1}^n |u_i|^3\norm{Qe_i}_2
%     \le
%     \Delta_\alpha M,
% \end{aligned}
% \]
% using \(\norm{Qe_i}_2\le1\) and
% \(\sum_i |u_i|^3\le\norm{u}_2^2=1\).
\end{proof}

These termination estimates allow the ideal rounding proof to be repeated
with explicit error terms.
\begin{lemma}
\label{lem:approx-coordinate-rounding}
Let \(u\) be the output of Algorithm~\ref{alg:projective-local-search}.  Then,
for every \(\eta>0\),
\begin{equation}
\label{eq:approx-coordinate-internal-estimate}
    \tau
    \le
    (1+\eta)f(u)
    +C_\eta\rho_\alpha\frac{9}{(k-1)^2}\sum_{i=1}^nP_{ii}^2
    +4\Delta_\alpha,
\end{equation}
where \(C_\eta\) is the constant from
Lemma~\ref{lem:centered-quartic-sos}.
\end{lemma}

\begin{proof}
We follow the similar proof of Theorem~\ref{thm:sharp-one-shot-rounding}, with the maximizer estimates replaced by their approximate coordinate-circle
analogues.
% Set
% \[
%     s=\langle x,u\rangle,
%     \qquad
%     M=\Phi_k(u)=\wtE s^{2k}>0,
%     \qquad
%     Q=P-uu^\top,
%     \qquad
%     z=Qx=x-su,
%     \qquad
%     A=\wtE s^{2k-4}.
% \]
As in Lemma~\ref{lem:positive-projective-moment}, \(A\ge M>0\), and the
threshold constraint with \(q=s^{k-2}\) gives
\[
    \tau A\le \wtE s^{2k-4}f(x).
\]
% For \(h_i=Qe_i\), Proposition~\ref{prop:projective-local-search-termination}
% and the binomial estimate used in Lemma~\ref{lem:fourth-order-tangent} give
% the aggregate tangent bound
By Proposition~\ref{prop:coordinate-circle-concentration}, we have
\[
    \wtE [s^{2k-4}\norm{z}_4^4]
    \le
    \frac{9\rho_\alpha}{(k-1)^2}
    M\sum_{i=1}^nP_{ii}^2,
\]
and
\begin{align*}
    \left|
    \wtE s^{2k-1}\sum_{i=1}^n u_i^3z_i
    \right|
    &= \left| \wtE s^{2k-1}\sum_{i=1}^n u_i^3\langle x,Qe_i\rangle \right|\\
    &\le
    \sum_{i=1}^n |u_i|^3
    \left|\wtE s^{2k-1}\langle x,Qe_i\rangle\right|  \\
    &\le
    \Delta_\alpha M
    \sum_{i=1}^n |u_i|^3\norm{Qe_i}_2
    \le
    \Delta_\alpha M.
\end{align*}
% \[
%     \left|
%     \wtE s^{2k-1}\sum_{i=1}^n u_i^3z_i\todo{}
%     \right|
%     \le
%     \Delta_\alpha M.
% \]
% Indeed,
% \[
%     \sum_{i=1}^n u_i^3z_i
%     =
%     \sum_{i=1}^n u_i^3\langle x,Qe_i\rangle,
% \]
% and the conclusion follows from
% \(\norm{Qe_i}_2\le1\) and
% \(\sum_i |u_i|^3\le \norm{u}_2^3=1\).

Applying Lemma~\ref{lem:centered-quartic-sos} coordinate-wise to
\(x_i=su_i+z_i\), summing over \(i\), multiplying by \(s^{2k-4}\), and then
using the preceding two estimates gives
\[
    \tau A
    \le
    M\left(
        (1+\eta)f(u)
        +C_\eta\rho_\alpha\frac{9}{(k-1)^2}\sum_{i=1}^nP_{ii}^2
        +4\Delta_\alpha
    \right).
\]
Since \(A/M\ge1\), division by \(M\) proves
\eqref{eq:approx-coordinate-internal-estimate}.
\end{proof}

We now choose the search accuracy and obtain a deterministic rounding
procedure with the same guarantee as the ideal maximizer-based rounding.

\begin{theorem}
\label{thm:projector-algorithmic-rounding}
There exists a universal constant \(C_{\mathrm{alg}}>0\) such that the
following holds.  Let \(0<\varepsilon<1\), \(\tau>0\), and
\(\wtE\in\mathcal{F}_{2k}(P,\tau)\).  Suppose
\begin{equation}
\label{eq:algorithmic-rounding-degree-condition}
    k-1\ge C_{\mathrm{alg}}\frac{\sqrt n}{\varepsilon}.
\end{equation}
The following deterministic rounding rule outputs a unit vector
\(\widehat v\in W\) satisfying
\begin{equation}
\label{eq:algorithmic-rounding-vector-guarantee}
    \norm{\widehat v}_4^4\ge(1-\varepsilon)\tau .
\end{equation}
If some index \(i\) satisfies
\(P_{ii}^2\ge(1-\varepsilon)\tau\), output
\[
    \widehat v=\frac{Pe_i}{\sqrt{P_{ii}}}.
\]
Otherwise set
\begin{equation}
\label{eq:finite-search-alpha}
    \alpha:=
    \min\left\{
        \frac1{16},
        \frac{\varepsilon^2\tau^2}{2^{16}}
    \right\},
\end{equation}
run Algorithm~\ref{alg:projective-local-search} with this value of
\(\alpha\), and output its final vector.  Consequently, for the output
\((\widehat\tau,\wtE^{\mathrm{out}})\) of Algorithm~\ref{alg:sos-2to4}, the
same rounding rule with \(\tau=\widehat\tau\) produces
\(\widehat v\in W\), \(\norm{\widehat v}_2=1\), and
\[
    \norm{\widehat v}_4^4
    \ge
    (1-\varepsilon)\widehat\tau .
\]
\end{theorem}

\begin{proof}
The $P_{ii}^2 \ge (1-\epsilon)\tau$ case is already verified in Corollary~\ref{cor:sharp-ideal-rounding}, therefore it remains to assume that
\[
    P_{ii}^2<(1-\varepsilon)\tau
    \qquad\text{for every }i\in[n].
\]  Let \(u\) be the output of
Algorithm~\ref{alg:projective-local-search}.  Then
\(
    \sum_{i=1}^n P_{ii}^2
    \le n\max_i P_{ii}^2
    < n\tau .
\)
Set
\[
    B:=
    \rho_\alpha\frac{9}{(k-1)^2}\sum_{i=1}^nP_{ii}^2,
    \qquad
    a:=\norm{u}_4 .
\]
Since \(\alpha\le1/16\), we have \(\rho_\alpha<2\).  From
\eqref{eq:algorithmic-rounding-degree-condition}, we obtain
\(
    \frac{B}{\tau}
    <
    \frac{18\varepsilon^2}{C_{\mathrm{alg}}^2}.
\)

Choosing \(C_{\mathrm{alg}}\ge 576\sqrt{2} \)  gives
\begin{equation}
\label{eq:algorithmic-B-small}
    B<\tau,
    \qquad
    24\frac{B^{1/2}}{\tau^{1/4}}
    \le
    \frac{\varepsilon}{8}\tau^{1/4}.
\end{equation}
Moreover, by the definition of \(\alpha\) and
\eqref{eq:rho-delta-alpha-definitions},
\begin{equation}
\label{eq:algorithmic-cubic-error-small}
    4\frac{\Delta_\alpha}{\tau^{3/4}}
    \le
    \frac{\varepsilon}{16}\tau^{1/4}.
\end{equation}

Lemma~\ref{lem:approx-coordinate-rounding} now gives the same scalar
inequality as in the final part of the proof of
Theorem~\ref{thm:sharp-one-shot-rounding}, with \(B\) in place of the
tangent-error term and with the additional additive error \(4\Delta_\alpha\).
If \(a\ge\tau^{1/4}\), the desired conclusion is immediate.  Otherwise,
using \(B<\tau\) and choosing
\[
    \eta=\frac{B^{1/2}}{\tau^{1/2}}\in(0,1)
\]
in Lemma~\ref{lem:approx-coordinate-rounding}, the same optimization over
\(\eta\) as in Theorem~\ref{thm:sharp-one-shot-rounding} yields
\[
    (\tau^{1/4}-a)\tau^{3/4}
    \le
    24B^{1/2}\tau^{1/2}
    +4\Delta_\alpha .
\]
Combining this with
\eqref{eq:algorithmic-B-small} and
\eqref{eq:algorithmic-cubic-error-small} gives
\[
    \tau^{1/4}-a
    \le
    \frac{\varepsilon}{4}\tau^{1/4}.
\]
Therefore
\[
    \norm{u}_4^4=a^4
    \ge
    \left(1-\frac{\varepsilon}{4}\right)^4\tau
    \ge
    (1-\varepsilon)\tau ,
\]
for \(0<\varepsilon<1\).  Taking \(\widehat v=u\) proves the claim in the
finite-search branch.  Applying the same rule to
\((\widehat\tau,\wtE^{\mathrm{out}})\) gives the final assertion.
\end{proof}

\medskip
\paragraph{Running time.}
The coefficients of \(\Phi_k\) are read directly from the degree-\(2k\)
moment vector:
\[
    \Phi_k(v)
    =
    \sum_{|\beta|=2k}
    \binom{2k}{\beta}\wtE[x^\beta]v^\beta .
\]
Algorithm~\ref{alg:projective-local-search} uses \(O(k/\alpha)\) grid
points on each coordinate circle and performs \(O(k\log r/\alpha)\) updates.
Hence the rounding time is polynomial in \(n\), \(k\), \(1/\alpha\), and the
size of the supplied moment table.  When it is applied to the output of
Algorithm~\ref{alg:sos-2to4}, one has \(\widehat\tau\ge1/n\), and therefore
\(\alpha^{-1}\le\poly(n,1/\varepsilon)\) for the value of \(\alpha\) in
\eqref{eq:finite-search-alpha}.  The rounding stage is then bounded by
\(n^{O(k)}\) arithmetic operations, which is no larger than the semidefinite
feasibility cost in Algorithm~\ref{alg:sos-2to4}.  With
\(k=\lceil C_{\mathrm{alg}}\sqrt n/\varepsilon\rceil+1\), the vector rounding
is absorbed into the overall
\[
    n^{O(\sqrt n/\varepsilon)}
\]
running time.
\section{An SoS relaxation for matrix \texorpdfstring{\(p\!\to q\)}{p-to-q} norms with even exponents}
\label{sec:sos-pq}

This section extends the preceding analysis of the projector \(2\!\to\!4\)
norm to matrix \(p\!\to q\) norms for fixed even integers
\(2\le p<q\).  In this regime, both the \(p\)-unit ball and the objective
\(\norm{Ax}_q^q\) have finite polynomial descriptions.  We first formulate
the problem and its fixed-threshold SoS relaxation, and then use binary search
over the threshold to approximate the SoS value.  By duality, the result also
applies to the corresponding conjugate exponents below \(2\).  We do not
address the mixed regime \(1<p<2<q\).

The main part of this section is the rounding analysis.  As in the projector
\(2\!\to\!4\) case, we maximize a directional moment polynomial, reweigh the
pseudo-expectation toward the maximizing direction, and round to an actual
vector.  The geometry is now that of the \(\ell_p\)-unit ball and its dual
\(\ell_{p^\ast}\)-unit sphere.  Since \(1<p^\ast\le2\), smoothness of the
dual norm yields the factor \(k^{-2/p}\) in the tangent-moment bounds.  The
projected coordinate candidates from the preceding section are replaced by
the \emph{row candidates} \(u_i\), the \(\ell_p\)-unit vectors that attain
H\"older's inequality for the individual rows of \(A\).

\subsection{Problem formulation}

Let nonzero \(A\in\R^{N\times d}\) have rows
\(a_1,\ldots,a_N\in\R^d\).  Fix even integers \(2\le p<q\), and write
\(q=2r\), so \(r\ge2\).  Define
\[
    \norm{A}_{p\to q}
    :=
    \max_{x\ne0}\frac{\norm{Ax}_q}{\norm{x}_p},
    \qquad
    p^\ast:=\frac{p}{p-1}.
\]
Since \(q\) is even, the \(q\)-th-power objective is the homogeneous polynomial
\begin{equation}
\label{eq:pq-objective-polynomial}
    F_A(x):=\norm{Ax}_q^q
    =
    \sum_{i=1}^N \langle a_i,x\rangle^q,
\end{equation}
and the corresponding optimum is
\begin{equation}
\label{eq:pq-opt}
    \OPT_{p\to q}(A)
    :=
    \norm{A}_{p\to q}^q
    =
    \max_{\norm{x}_p\le1}F_A(x)
    =
    \max_{\norm{x}_p=1}F_A(x).
\end{equation}
The final equality follows from homogeneity and the assumption \(A\ne0\).

For a nonzero row \(a_i\), set
\[
    w_i:=\frac{a_i}{\norm{a_i}_{p^\ast}},
    \qquad
    u_i:=J_{p^\ast}(w_i).
\]
By \Cref{subsec:norms-polynomials},
\[
    \norm{u_i}_p=1,
    \qquad
    \langle a_i,u_i\rangle=\norm{a_i}_{p^\ast}.
\]
Thus \(u_i\) attains the H\"older upper bound for the row functional
\(x\mapsto\langle a_i,x\rangle\).  We call \(u_i\) the \emph{row candidate}
associated with \(a_i\).

\subsection{The fixed-threshold SoS relaxation}

Fix an integer \(k\ge r\), and set
\begin{equation}
\label{eq:pq-implemented-degree}
    D_{p,k}:=2k+p,
\end{equation}
and define the polynomial
\[
    g_p(x):=1-\sum_{j=1}^d x_j^p.
\]
Because \(p\) is even, \(g_p(x)\ge0\) is exactly the \(\ell_p\)-unit ball.
For a threshold \(\tau\ge0\), let
\(\mathcal F_{D_{p,k}}^{p,q}(A,\tau)\) be the set of degree-
\(D_{p,k}\) pseudo-expectations \(\wtE\) satisfying
\begin{subequations}
\label{eq:pq-fixed-threshold-relaxation}
\begin{align}
    \wtE\!\left[h(x)^2g_p(x)\right]
    &\ge0
    &&\text{whenever }2\deg h+p\le D_{p,k},
    \label{eq:pq-ball-condition-main}\\
    \wtE\!\left[h(x)^2\bigl(F_A(x)-\tau\bigr)\right]
    &\ge0
    &&\text{whenever }2\deg h+q\le D_{p,k}.
    \label{eq:pq-objective-condition-main}
\end{align}
\end{subequations}
For each fixed \(\tau\), these constraints are affine in the degree-
\(D_{p,k}\) pseudo-moments and define a semidefinite feasibility problem in
the original \(d\) variables.  The first constraint enforces the \(p\)-unit
ball, while the second enforces the objective threshold.  The next lemma
derives the H\"older inequalities used in the rounding proof from the first
constraint.

\begin{lemma}[H\"older inequalities from the even \(p\)-ball]
\label{lem:pq-holder-from-even-ball}
Let \(p\ge2\) be even, let \(\norm{w}_{p^\ast}=1\), and put
\(s=\langle w,x\rangle\).  For every integer \(j\ge1\), the polynomial
\(1-s^{2j}\) has an SoS certificate from \(g_p(x)\ge0\) of the form
\[
    1-s^{2j}=\sigma_0(x)+\sigma_1(x)g_p(x),
    \qquad
    \sigma_0,\sigma_1\in\SOS,
\]
with total degree at most \(2j+p\).  Consequently, every
\(\wtE\in\mathcal F_{D_{p,k}}^{p,q}(A,\tau)\) satisfies
\begin{equation}
\label{eq:pq-holder-condition-main}
    \wtE\!\left[
        h(x)^2\bigl(1-\langle w,x\rangle^{2j}\bigr)
    \right]\ge0
\end{equation}
for every \(1\le j\le k\) and every \(h\) with
\(\deg h\le k-j\).
\end{lemma}

\begin{proof}
For each coordinate, Young's inequality shows that the univariate polynomial
\[
    Y_i(t)
    :=
    \frac1p t^p
    +\frac1{p^\ast}|w_i|^{p^\ast}
    -w_it
\]
is nonnegative on \(\R\).  It has even degree, and hence is a sum of squares
of real univariate polynomials.  Since
\(\sum_i|w_i|^{p^\ast}=1\),
\begin{equation}
\label{eq:pq-linear-holder-certificate}
    1-\langle w,x\rangle
    =
    \frac1p g_p(x)+\sum_{i=1}^dY_i(x_i).
\end{equation}
Replacing \(w\) by \(-w\) gives the analogous certificate for
\(1+\langle w,x\rangle\).

Writing \(s=\langle w,x\rangle\), one has the polynomial identity
\begin{equation}
\label{eq:pq-interval-holder-identity}
\begin{aligned}
    1-s^{2j}
    &=
    \frac12(1+s)^2
       \left(\sum_{\ell=0}^{j-1}s^{2\ell}\right)(1-s)\\
    &\qquad+
    \frac12(1-s)^2
       \left(\sum_{\ell=0}^{j-1}s^{2\ell}\right)(1+s).
\end{aligned}
\end{equation}
Both multipliers of \(1-s\) and \(1+s\) are sums of squares.  Combining
\eqref{eq:pq-linear-holder-certificate} and its \(-w\) version with
\eqref{eq:pq-interval-holder-identity} gives the claimed SoS certificate,
of total degree at most \(2j+p\).  Multiplying this certificate
by \(h^2\) has degree at most
\[
    2(k-j)+2j+p=D_{p,k},
\]
so the \(p\)-ball constraint \eqref{eq:pq-ball-condition-main} implies
\eqref{eq:pq-holder-condition-main}.
\end{proof}

Define the degree-\(D_{p,k}\) SoS value by
\begin{equation}
\label{eq:pq-sos-value}
    \SOS_{D_{p,k}}^{p\to q}(A)
    :=
    \sup\{\tau\ge0:
    \mathcal F_{D_{p,k}}^{p,q}(A,\tau)\ne\varnothing\}.
\end{equation}
The feasible thresholds are downward closed.
Indeed, if \(\wtE\in\mathcal F_{D_{p,k}}^{p,q}(A,\tau)\) and
\(0\le\tau'\le\tau\), then
\[
    \wtE\!\left[h(x)^2\bigl(F_A(x)-\tau'\bigr)\right]
    =
    \wtE\!\left[h(x)^2\bigl(F_A(x)-\tau\bigr)\right]
    +(\tau-\tau')\wtE h(x)^2
    \ge0
\]
for every admissible \(h\), while the \(p\)-ball constraint is unchanged.

\subsection{The matrix \texorpdfstring{\(p\!\to\!q\)}{p-to-q} norm approximation algorithm}

Algorithm~\ref{alg:sos-pq} uses the downward-closed property of the feasible
thresholds to approximate \(\SOS_{D_{p,k}}^{p\to q}(A)\).  The lower endpoint
is initialized by the best row candidate, and the upper endpoint is the H\"older
bound obtained from the row norms.  The validity of these endpoints and the
binary-search invariant is proved in Theorem~\ref{thm:pq-sos-main}.

\begin{algorithm}[H]
\caption{SoS approximation for matrix \(p\!\to\!q\) norms with even exponents}
\label{alg:sos-pq}
\begin{algorithmic}[1]
\Statex \textbf{Input:} A nonzero matrix \(A\in\R^{N\times d}\), fixed even
integers \(2\le p<q=2r\), an integer \(k\ge r\), and a tolerance
\(\varepsilon\in(0,1)\).
\Statex \textbf{Output:} A threshold \(\widehat\tau\) and a degree-
\(D_{p,k}\) pseudo-expectation \(\wtE^{\mathrm{out}}\).

\State For every nonzero row, set
\[
    w_i\gets \frac{a_i}{\norm{a_i}_{p^\ast}},
    \qquad
    u_i\gets J_{p^\ast}(w_i).
\]
\State Set
\[
    B\gets\sum_{i=1}^N \norm{a_i}_{p^\ast}^q,
    \qquad
    i_0\gets\argmax_{i:a_i\ne0}\norm{Au_i}_q^q,
\]
\[
    \tau_{\mathrm{lo}}\gets\norm{Au_{i_0}}_q^q,
    \qquad
    \tau_{\mathrm{hi}}\gets B,
    \qquad
    \wtE^{\mathrm{best}}\gets\wtE_{u_{i_0}},
\]
where \(\wtE_{u_{i_0}}[g]=g(u_{i_0})\) for \(\deg g\le D_{p,k}\).
\While{\(\tau_{\mathrm{hi}}>(1+\varepsilon)\tau_{\mathrm{lo}}\)}
    \State \(\tau_{\mathrm{mid}}\gets(\tau_{\mathrm{lo}}+\tau_{\mathrm{hi}})/2\).
    \If{\(\mathcal F_{D_{p,k}}^{p,q}(A,\tau_{\mathrm{mid}})\ne\varnothing\)}
        \State Choose
        \(\wtE_{\mathrm{mid}}\in
        \mathcal F_{D_{p,k}}^{p,q}(A,\tau_{\mathrm{mid}})\).
        \State \(\tau_{\mathrm{lo}}\gets\tau_{\mathrm{mid}}\) and
        \(\wtE^{\mathrm{best}}\gets\wtE_{\mathrm{mid}}\).
    \Else
        \State \(\tau_{\mathrm{hi}}\gets\tau_{\mathrm{mid}}\).
    \EndIf
\EndWhile
\State \(\widehat\tau\gets\tau_{\mathrm{lo}}\) and
\(\wtE^{\mathrm{out}}\gets\wtE^{\mathrm{best}}\).
\State \Return \((\widehat\tau,\wtE^{\mathrm{out}})\).
\end{algorithmic}
\end{algorithm}

\paragraph{Running time.}
For \(A\ne0\), the row-candidate initialization gives
\(0<\tau_{\mathrm{lo}}\le B\le N\tau_{\mathrm{lo}}\), so the binary search
makes \(O(\log(N/\varepsilon))\) feasibility calls.  Each feasibility
problem is a degree-\(D_{p,k}=2k+p\) moment SDP in \(d\) variables, with the
two families of constraints in
\eqref{eq:pq-fixed-threshold-relaxation}.  For fixed \(p,q\), its description
size is
\[
    (N+d)^{O_{p,q}(k)},
\]
up to polynomial factors in the input bit length and numerical precision.  At
the degree scale
\[
    k=
    O_{p,q}\!\left(
        \frac{N^{p/q}}{\varepsilon^{p/2}}
    \right),
\]
the running time is
\[
    (N+d)^{O_{p,q}\left(N^{p/q}/\varepsilon^{p/2}\right)}.
\]

The algorithm above approximates the optimum of the degree-\(D_{p,k}\)
threshold relaxation.  The following theorem shows that, for sufficiently
large \(k\), this relaxation gives a relative approximation to the true
matrix \(p\!\to q\) norm.

\begin{theorem}[Approximation guarantee for matrix \(p\!\to\!q\) norms with even exponents]
\label{thm:pq-sos-main}
Fix even integers \(2\le p<q\), and write \(q=2r\).  There exists a constant
\(C_{p,q}>0\) such that the following holds.  Let
\(A\in\R^{N\times d}\) be nonzero, let \(0<\varepsilon<1\), and let
\(k\ge r\) satisfy
\begin{equation}
\label{eq:pq-degree-requirement}
    k-1
    \ge
    C_{p,q}
    \frac{N^{p/q}}{\varepsilon^{p/2}}.
\end{equation}
Let \((\widehat\tau,\wtE^{\mathrm{out}})\) be the output of
Algorithm~\ref{alg:sos-pq}.  Then
\begin{equation}
\label{eq:pq-output-norm-approximation}
    (1+\varepsilon)^{-1/q}\norm{A}_{p\to q}
    \le
    \widehat\tau^{1/q}
    \le
    (1-\varepsilon)^{-1/q}\norm{A}_{p\to q}.
\end{equation}
Moreover, the degree-\(D_{p,k}\) SoS value satisfies
\begin{equation}
\label{eq:pq-sos-sandwich}
    \OPT_{p\to q}(A)
    \le
    \SOS_{D_{p,k}}^{p\to q}(A)
    \le
    \left(
        1+
        C_{p,q}\frac{N^{2/q}}{(k-1)^{2/p}}
        +
        C_{p,q}\frac{N}{(k-1)^{q/p}}
    \right)
    \OPT_{p\to q}(A).
\end{equation}
\end{theorem}

\begin{corollary}[Dual exponents below \(2\)]
\label{cor:pq-dual-regime}
Fix even integers \(2\le p<q\).  For every
\(A\in\R^{N\times d}\),
\[
    \norm{A}_{q^\ast\to p^\ast}
    =
    \norm{A^\top}_{p\to q}.
\]
Consequently, applying Theorem~\ref{thm:pq-sos-main} to \(A^\top\) gives
the same two-sided guarantee as \eqref{eq:pq-output-norm-approximation} for
\(\norm{A}_{q^\ast\to p^\ast}\), with \(N\) replaced by \(d\).  In
particular, replacing the theorem's tolerance by a sufficiently small
\(p,q\)-dependent multiple of \(\varepsilon\) and taking
\[
    k=
    O_{p,q}\!\left(
        \frac{d^{p/q}}{\varepsilon^{p/2}}
    \right)
\]
and hence degree \(D_{p,k}=2k+p\) gives a multiplicative
\((1+\varepsilon)\)-approximation.
\end{corollary}

The proof of Theorem~\ref{thm:pq-sos-main} is deferred to
Subsection~\ref{subsec:proof-pq-sos-main}.

\subsection{Spherical directional moment polynomial and rounding}
\label{subsec:pq-rounding}

For the soundness analysis of Theorem~\ref{thm:pq-sos-main}, we follow the
same argmax-principle-guided rounding strategy as in the preceding section.
Fix \(\tau>0\) and
\(\wtE\in\mathcal F_{D_{p,k}}^{p,q}(A,\tau)\).  Define the spherical
directional moment polynomial
\begin{equation}
\label{eq:pq-powered-moment}
    \Phi_k(w):=\wtE\langle w,x\rangle^{2k},
    \qquad
    \norm{w}_{p^\ast}=1.
\end{equation}
It records the \(2k\)-th pseudo-moment of \(x\) in the direction \(w\).
Since the \(p^\ast\)-unit sphere is compact, \(\Phi_k\) attains its maximum.
Choose
\begin{equation}
\label{eq:pq-maximizer-def}
    w\in\operatorname*{arg\,max}_{\norm{z}_{p^\ast}=1}\Phi_k(z),
    \qquad
    u=J_{p^\ast}(w),
    \qquad
    s=\langle w,x\rangle.
\end{equation}

\begin{lemma}[Tangent smoothness of the dual sphere]
\label{lem:pq-tangent-smoothness}
Fix an even integer \(p\ge2\).  There is a constant \(C_p>0\) such that,
whenever \(\norm{w}_{p^\ast}=1\),
\(\langle b,J_{p^\ast}(w)\rangle=0\), and
\(|t|\norm{b}_{p^\ast}\le1\),
\begin{equation}
\label{eq:pq-tangent-smoothness}
    \norm{w+tb}_{p^\ast}
    \le
    1+C_p|t|^{p^\ast}\norm{b}_{p^\ast}^{p^\ast}.
\end{equation}
The same estimate holds with \(t\) replaced by \(-t\).
\end{lemma}

\begin{proof}
Put \(\rho=p^\ast\in(1,2]\).  By homogeneity in \(b\), it is enough to
consider \(\norm{b}_\rho=1\).  There is a constant \(K_\rho\) such that
\[
    |\xi+\eta|^\rho
    \le
    |\xi|^\rho
    +\rho\operatorname{sgn}(\xi)|\xi|^{\rho-1}\eta
    +K_\rho|\eta|^\rho
\]
for all real \(\xi,\eta\).  Applying this inequality coordinatewise to
\(w+tb\), summing, and using
\(\langle b,J_\rho(w)\rangle=0\) gives
\[
    \norm{w+tb}_\rho^\rho\le1+K_\rho|t|^\rho.
\]
Since \((1+z)^{1/\rho}\le1+z\) for \(z\ge0\), the claimed estimate follows.
Replacing \(b\) by \(-b\) proves the second assertion.
\end{proof}

The next lemma is the analogue of the argmax-induced concentration estimate
in the preceding section.  It converts global maximality of \(w\) into
high-moment bounds in every tangent direction and also records the
first-order stationarity identity used in the rounding proof.

\begin{lemma}[Argmax-induced tangent estimates]
\label{lem:pq-dual-tangent}
Fix even integers \(2\le p<q=2r\), and let \(k\ge r\).  Let
\(\wtE\in\mathcal F_{D_{p,k}}^{p,q}(A,\tau)\), let \(w,u,s\) be as in
\eqref{eq:pq-maximizer-def}, and assume \(\wtE s^{2k}>0\).  For
\(1\le j\le k\), define
\[
    \wtE_{w,j}[g]
    :=
    \frac{\wtE s^{2k-2j}g(x)}{\wtE s^{2k-2j}}.
\]
Then, for every \(b\in\R^d\) satisfying \(\langle b,u\rangle=0\),
\begin{equation}
\label{eq:pq-dual-tangent-high-moment}
    \wtE_{w,j}\langle b,x\rangle^{2j}
    \le
    C_{p,j}(k-1)^{-2j/p}
    \wtE_{w,j}s^{2j}\norm{b}_{p^\ast}^{2j}
    \le
    C_{p,j}(k-1)^{-2j/p}\norm{b}_{p^\ast}^{2j}.
\end{equation}
In particular, for
\[
    \wtE_w[g]
    :=
    \frac{\wtE s^{2k-q}g(x)}{\wtE s^{2k-q}},
\]
one has
\begin{align}
\label{eq:pq-dual-tangent-q-moment}
    \wtE_w\langle b,x\rangle^q
    &\le
    C_{p,q}(k-1)^{-q/p}\norm{b}_{p^\ast}^q,\\
\label{eq:pq-dual-tangent-stationarity}
    \wtE_w s^{q-1}\langle b,x\rangle
    &=0,\\
\label{eq:pq-dual-tangent-quadratic}
    \wtE_w s^{q-2}\langle b,x\rangle^2
    &\le
    C_{p,q}(k-1)^{-2/p}\norm{b}_{p^\ast}^2.
\end{align}
\end{lemma}

\begin{proof}
The denominators are positive because
Lemma~\ref{lem:pq-holder-from-even-ball}, applied with
\(h=s^{k-j}\), gives
\[
    \wtE s^{2k-2j}-\wtE s^{2k}
    =
    \wtE s^{2k-2j}(1-s^{2j})
    \ge0.
\]
By homogeneity assume \(\norm{b}_{p^\ast}=1\).  Since
\(p^\ast\in(1,2]\), Lemma~\ref{lem:pq-tangent-smoothness} gives
\[
    \norm{w\pm tb}_{p^\ast}^{2k}
    \le
    (1+C_p|t|^{p^\ast})^{2k}
    \qquad (|t|\le1).
\]
The normalized vectors
\[
    \frac{w\pm tb}{\norm{w\pm tb}_{p^\ast}}
\]
lie on the \(p^\ast\)-unit sphere.  Global maximality of \(w\), together
with the homogeneity of \(\Phi_k\), therefore gives
\[
    \wtE(s\pm t\langle b,x\rangle)^{2k}
    \le
    \wtE s^{2k}\norm{w\pm tb}_{p^\ast}^{2k}.
\]
Averaging these two inequalities and retaining the \(2j\)-th term gives
\[
    \binom{2k}{2j}t^{2j}
    \wtE s^{2k-2j}\langle b,x\rangle^{2j}
    \le
    \wtE s^{2k}(1+C_p|t|^{p^\ast})^{2k}.
\]
Here every discarded even term has nonnegative pseudo-expectation because it
is the pseudo-expectation of a square.
Choosing \(t=c_{p,j}k^{-1/p^\ast}\), with \(c_{p,j}>0\) sufficiently
small, and using
\[
    1-\frac1{p^\ast}=\frac1p
\]
gives
\[
    \wtE s^{2k-2j}\langle b,x\rangle^{2j}
    \le
    C_{p,j}(k-1)^{-2j/p}\wtE s^{2k}.
\]
Dividing by \(\wtE s^{2k-2j}\) proves the first inequality in
\eqref{eq:pq-dual-tangent-high-moment}; the second follows from
Lemma~\ref{lem:pq-holder-from-even-ball}.  Taking \(j=r\) proves
\eqref{eq:pq-dual-tangent-q-moment}.

For stationarity, consider
\[
    w(t)=\frac{w+tb}{\norm{w+tb}_{p^\ast}}.
\]
The derivative of \(\norm{w+tb}_{p^\ast}\) at \(t=0\) is
\(\langle b,u\rangle=0\), so \(w'(0)=b\).  The function
\(t\mapsto\Phi_k(w(t))\) is stationary at \(t=0\), and hence
\(\wtE s^{2k-1}\langle b,x\rangle=0\).  This yields
\eqref{eq:pq-dual-tangent-stationarity} after division by
\(\wtE s^{2k-q}\).  Taking \(j=1\) in the high-moment estimate just proved
gives
\[
    \wtE s^{2k-2}\langle b,x\rangle^2
    \le
    C_{p,q}(k-1)^{-2/p}\wtE s^{2k}.
\]
Dividing by \(\wtE s^{2k-q}\) and using
\[
    \wtE_w s^q
    =
    \frac{\wtE s^{2k}}{\wtE s^{2k-q}}
    \le1
\]
proves
\eqref{eq:pq-dual-tangent-quadratic}.
\end{proof}

The following theorem is the main rounding estimate.  Starting from a
feasible threshold pseudo-expectation, it compares the moment-polynomial
maximizer with the row candidates and shows that one of these explicit vectors
has objective value comparable to the threshold.

\begin{theorem}[Rounding for matrix \(p\!\to\!q\) norms with even exponents]
\label{thm:pq-rounding}
Fix even integers \(2\le p<q=2r\), and let \(k\ge r\).  Let
\(\tau>0\) and
\(\wtE\in\mathcal F_{D_{p,k}}^{p,q}(A,\tau)\).  With \(u\) as in
\eqref{eq:pq-maximizer-def}, and with
\[
    u_i=J_{p^\ast}\!\left(\frac{a_i}{\norm{a_i}_{p^\ast}}\right)
    \qquad (a_i\ne0),
\]
there exists
\[
    v\in\{u\}\cup\{u_i:a_i\ne0\}
\]
such that \(\norm{v}_p=1\) and
\begin{equation}
\label{eq:pq-rounding-bound}
    \tau
    \le
    \left(
        1+
        C_{p,q}\frac{N^{2/q}}{(k-1)^{2/p}}
        +
        C_{p,q}\frac{N}{(k-1)^{q/p}}
    \right)
    \norm{Av}_q^q.
\end{equation}
Consequently, under \eqref{eq:pq-degree-requirement}, the same vector
satisfies
\begin{equation}
\label{eq:pq-rounded-vector-guarantee}
    \norm{Av}_q^q\ge(1-\varepsilon)\tau.
\end{equation}
\end{theorem}

\begin{proof}
The objective constraint \eqref{eq:pq-objective-condition-main} with
multiplier \(1\) gives
\(\wtE F_A(x)\ge\tau>0\), so some row has positive \(q\)-th
pseudo-moment.  For the corresponding linear form \(\ell(x)\), the even
pseudo-moments \(m_j=\wtE\ell(x)^{2j}\) are log-convex by
pseudo--Cauchy--Schwarz:
\[
    m_j^2\le m_{j-1}m_{j+1}.
\]
Since \(m_r>0\), this inequality implies successively that
\(m_{r+1},\ldots,m_k\) are positive.  After normalizing the row, we obtain a
direction \(w_i\) with \(\Phi_k(w_i)>0\).  Hence the maximum of \(\Phi_k\)
is positive.  Lemma~\ref{lem:pq-holder-from-even-ball}, applied with
\(h=s^{k-r}\) and \(j=r\), gives
\[
    \wtE s^{2k-q}(1-s^q)\ge0,
\]
so \(\wtE s^{2k-q}\ge\wtE s^{2k}>0\).  The reweighed functional
\begin{equation}
\label{eq:pq-reweighted-functional}
    \wtE_w[g]
    :=
    \frac{\wtE s^{2k-q}g(x)}{\wtE s^{2k-q}}
\end{equation}
is therefore well defined for every polynomial \(g\) of degree at most
\(q\), and
\begin{equation}
\label{eq:pq-reweighted-sq}
    \wtE_w s^q\le1.
\end{equation}

For each row set
\begin{equation}
\label{eq:pq-row-decomposition}
    \alpha_i=\langle a_i,u\rangle,
    \qquad
    b_i=a_i-\alpha_iw.
\end{equation}
Then \(\langle b_i,u\rangle=0\),
\(\langle a_i,x\rangle=\alpha_i s+\langle b_i,x\rangle\), and
\(\norm{Au}_q^q=\sum_i|\alpha_i|^q\).  Applying the objective constraint
\eqref{eq:pq-objective-condition-main} with multiplier \(s^{k-r}\) yields
\begin{equation}
\label{eq:pq-objective-after-reweighing}
    \tau\le\wtE_wF_A(x).
\end{equation}

By Lemma~\ref{lem:pq-row-binary-sos}, applied row by row with
\(a=\alpha_i s\) and \(b=\langle b_i,x\rangle\),
\[
\begin{aligned}
    \wtE_w F_A(x)
    &\le
    \norm{Au}_q^q\,\wtE_w s^q
    +q\,\wtE_w\!\left[
        s^{q-1}\left\langle\sum_i\alpha_i^{q-1}b_i,x\right\rangle
    \right]\\
    &\qquad
    +C_q\sum_i|\alpha_i|^{q-2}
        \wtE_w\!\left[s^{q-2}\langle b_i,x\rangle^2\right]
    +C_q\sum_i\wtE_w\langle b_i,x\rangle^q.
\end{aligned}
\]
For every \(i\), one has \(\langle b_i,u\rangle=0\), and therefore
\[
    \left\langle\sum_i\alpha_i^{q-1}b_i,u\right\rangle=0.
\]
The linear term thus vanishes by
\eqref{eq:pq-dual-tangent-stationarity}.  The bounds
\eqref{eq:pq-dual-tangent-q-moment} and
\eqref{eq:pq-dual-tangent-quadratic}, together with
\eqref{eq:pq-reweighted-sq}, give
\begin{equation}
\label{eq:pq-row-factorized-bound}
    \tau
    \le
    \norm{Au}_q^q
    +C_{p,q}(k-1)^{-2/p}
        \sum_i|\alpha_i|^{q-2}\norm{b_i}_{p^\ast}^2
    +C_{p,q}(k-1)^{-q/p}
        \sum_i\norm{b_i}_{p^\ast}^q.
\end{equation}

Choose \(v\) from \(\{u\}\cup\{u_i:a_i\ne0\}\) to maximize
\(\norm{Av}_q^q\).
Then \(\norm{v}_p=1\), \(\norm{Av}_q^q\ge\norm{Au}_q^q\), and for every
nonzero row,
\[
    \norm{Av}_q^q
    \ge
    \norm{Au_i}_q^q
    \ge
    |\langle a_i,u_i\rangle|^q
    =
    \norm{a_i}_{p^\ast}^q.
\]
Moreover,
\(|\alpha_i|\le\norm{a_i}_{p^\ast}\) and
\(\norm{b_i}_{p^\ast}\le2\norm{a_i}_{p^\ast}\).  Consequently,
\begin{equation}
\label{eq:pq-L-bound}
    \sum_i\norm{b_i}_{p^\ast}^q
    \le
    C_{p,q}N\norm{Av}_q^q.
\end{equation}
H\"older's inequality gives
\begin{align}
\label{eq:pq-T-bound}
    \sum_i|\alpha_i|^{q-2}\norm{b_i}_{p^\ast}^2
    &\le
    \left(\sum_i|\alpha_i|^q\right)^{(q-2)/q}
    \left(\sum_i\norm{b_i}_{p^\ast}^q\right)^{2/q}\\
    &\le
    C_{p,q}N^{2/q}\norm{Av}_q^q.
\end{align}
Substituting \eqref{eq:pq-L-bound} and \eqref{eq:pq-T-bound} into
\eqref{eq:pq-row-factorized-bound} proves
\eqref{eq:pq-rounding-bound}.  After increasing \(C_{p,q}\) if necessary,
the degree condition \eqref{eq:pq-degree-requirement} makes the sum of the
two error terms in \eqref{eq:pq-rounding-bound} at most
\(\varepsilon/(1-\varepsilon)\).  Rearranging then gives
\eqref{eq:pq-rounded-vector-guarantee}.
\end{proof}

\subsection{Proof of the approximation guarantee}
\label{subsec:proof-pq-sos-main}

We now complete the proof of the main approximation guarantee.  As in
Section~\ref{sec:sos-2to4}, we first bound the SoS value in terms of the true
optimum and then verify the binary-search invariant.

\begin{proof}[Proof of Theorem~\ref{thm:pq-sos-main}]
We first prove the lower bound
\(\OPT_{p\to q}(A)\le\SOS_{D_{p,k}}^{p\to q}(A)\).  Indeed, if
\(x^\star\) maximizes \(F_A\) over the \(p\)-unit ball, define the evaluation
pseudo-expectation by \(\wtE_{x^\star}[g]=g(x^\star)\).  It satisfies both
constraints in
\eqref{eq:pq-fixed-threshold-relaxation} at
threshold \(\OPT_{p\to q}(A)\).

For the reverse inequality, let \(\tau>0\) be feasible and choose
\(\wtE\in\mathcal F_{D_{p,k}}^{p,q}(A,\tau)\).
Theorem~\ref{thm:pq-rounding} gives
an \(\ell_p\)-unit vector \(v\) such that
\[
    \tau
    \le
    \left(
        1+
        C_{p,q}\frac{N^{2/q}}{(k-1)^{2/p}}
        +
        C_{p,q}\frac{N}{(k-1)^{q/p}}
    \right)
    \norm{Av}_q^q.
\]
Since \(\norm{Av}_q^q\le\OPT_{p\to q}(A)\), taking the supremum over all
feasible \(\tau\) proves \eqref{eq:pq-sos-sandwich}.

It remains to verify the binary-search guarantee.  For every nonzero row
\(a_i\), the row candidate satisfies
\[
    \norm{Au_i}_q^q\ge\norm{a_i}_{p^\ast}^q.
\]
The evaluation pseudo-expectation at \(u_{i_0}\) therefore shows that the
initialized lower endpoint is feasible, and \(B\le N\tau_{\mathrm{lo}}\).
To bound the SoS value from above,
Lemma~\ref{lem:pq-holder-from-even-ball}, with \(w=w_i\), \(j=r\), and
\(h=1\), gives
\[
    \wtE\langle a_i,x\rangle^q
    \le
    \norm{a_i}_{p^\ast}^q.
\]
The objective constraint \eqref{eq:pq-objective-condition-main} with
multiplier \(1\) therefore gives
\(\tau\le\wtE F_A(x)\le B\).  Hence
\(\SOS_{D_{p,k}}^{p\to q}(A)\le B\), and binary search terminates after
\(O(\log(N/\varepsilon))\) iterations with
\begin{equation}
\label{eq:pq-binary-search-guarantee}
    \widehat\tau
    \le
    \SOS_{D_{p,k}}^{p\to q}(A)
    \le
    (1+\varepsilon)\widehat\tau.
\end{equation}
The degree condition \eqref{eq:pq-degree-requirement}, with \(C_{p,q}\)
sufficiently large, makes the multiplicative factor in
\eqref{eq:pq-sos-sandwich} at most \((1-\varepsilon)^{-1}\).  Combining
this with \eqref{eq:pq-binary-search-guarantee} yields
\[
    \frac{1}{1+\varepsilon}\OPT_{p\to q}(A)
    \le
    \widehat\tau
    \le
    \frac{1}{1-\varepsilon}\OPT_{p\to q}(A).
\]
Taking \(q\)-th roots proves
\eqref{eq:pq-output-norm-approximation}.
\end{proof}

% ============================================================
% Part 3: Homogeneous Degree-d Polynomials
% ============================================================

\section{SoS approximation for degree-\texorpdfstring{$d$}{d} polynomials}
\label{sec:degree-d-argmax-fold-sos}
This section applies the argmax principle to fixed-degree polynomial optimization over the unit sphere, where the relevant auxiliary object is a \emph{fold} rather than a
direction.
%This section establishes the SoS approximation guarantee for optimizing a homogeneous degree-$d$ polynomial over the unit sphere.
This uses the standard polynomial-folding viewpoint from earlier work on homogeneous polynomial optimization~\cite{KN07,HLZ10,So11}, as well as the SoS convergence analysis of Bhattiprolu--Ghosh--Guruswami--Lee--Tulsiani~\cite{BGGLT}.%we give a new soundness analysis for the SoS relaxation.
The key difference in our analysis is that, rather than using the weak decoupling tool developed by Bhattiprolu et al., we select the fold by our argmax principle.
A pseudo-H\"older and spherical-moment
calculation lower bounds the average high moment over all folds, hence the maximizing fold has large pseudo-moment; once such a fold is selected, rounding is spectral, and polarization transfers the folded value back to the original degree-\(d\) polynomial.
This proves the desired $O_d((n/k)^{d/2-1})$ approximation guarantee for the SoS relaxation.
% The key difference is the rounding argument: we fold $d-2$ variables using random
% spherical directions, certify the existence of a quadratic fold with large
% pseudo-moment via pseudo-H\"older, and then convert this fold into an actual
% unit vector by taking a top eigenvector and applying polarization.  This
% direct rounding step proves the stated
% $O_d((n/k)^{d/2-1})$ approximation guarantee for the degree-$D_{d,k}$
% absolute-value SoS relaxation.

As is standard in sphere-constrained polynomial optimization, it suffices,
up to a constant factor depending only on \(d\), to treat homogeneous
polynomials; see, e.g.,~\cite{DW12}. Indeed, for a degree-\(d\)
polynomial \(f\), decompose \(f=f_{\mathrm{even}}+f_{\mathrm{odd}}\).
On the unit sphere, each parity component can be made homogeneous by
multiplying its lower-degree homogeneous pieces by suitable powers of
\(\|x\|_2^2\).  Moreover
\(\|f_{\mathrm{even}}\|_{\Sph},\|f_{\mathrm{odd}}\|_{\Sph}\le \|f\|_{\Sph}\)
and
\(\|f\|_{\Sph}\le
\|f_{\mathrm{even}}\|_{\Sph}+\|f_{\mathrm{odd}}\|_{\Sph}\).
Thus the homogeneous case implies the general degree-\(d\) case after losing
only a factor \(2\), absorbed in \(O_d(\cdot)\).

\subsection{Problem input and SoS relaxation}
\label{subsec:degree-d-input-algorithm}

% Fix an integer $d\ge3$.  The input consists of an integer $n$, a symmetric
% $d$-linear form
% \[
%     T:(\R^n)^d\to\R,
% \]
% given by its coefficients, and an integer level parameter $k\ge1$.  We write
% \[
%     f(x):=T(x,\ldots,x),
%     \qquad
%     \OPT(f):=\norm{f}_{\Sph}:=\max_{\norm{x}_2=1}|f(x)|.
% \]

Fix an integer \(d\ge3\).  The input is a homogeneous degree-\(d\)
polynomial \(f:\R^n\to\R\), given by its coefficients.  Let
\(T:(\R^n)^d\to\R\)
denote the associated symmetric \(d\)-linear form, so that
\[
    f(x)=T(x,\ldots,x).
\]
We write
\[
    \OPT(f):=\norm{f}_{\Sph}:=\max_{\norm{x}_2=1}|f(x)|.
\]

We use standard degree-$D$ pseudo-expectation notation.  A feasible
pseudo-expectation for the unit sphere is required to satisfy
\begin{equation}
    \wtE\big[(\norm{x}_2^2-1)q(x)\big]=0
    \quad\text{whenever }\deg q\le D-2.
    \label{eq:degree-d-sphere-constraint}
\end{equation}
For a homogeneous polynomial $f$, define the degree-$D$ absolute-value SoS
relaxation by
\begin{equation}
    \tau_D^{\sos}(f)
    :=
    \max\left\{
        \sup_{\wtE}\wtE[f(x)],
        \sup_{\wtE}\wtE[-f(x)]
    \right\},
    \label{eq:degree-d-sos-relaxation}
\end{equation}
where both suprema range over degree-$D$ pseudo-expectations satisfying
\eqref{eq:degree-d-sphere-constraint}.  The objective in
\eqref{eq:degree-d-sos-relaxation} is linear in $f$.

\begin{algorithm}[H]
\caption{Degree-$D_{d,k}$ SoS relaxation}
\label{alg:degree-d-sos}
\begin{algorithmic}[1]
\Require Coefficients of the symmetric form $T:(\R^n)^d\to\R$ associated with $f$ and level parameter $k\ge2$.
\State Set
\[
    D_{d,k}:=\max\{4k,\,2k(d-2)\}.
\]
\State Solve the two degree-$D_{d,k}$ SoS relaxations under \eqref{eq:degree-d-sphere-constraint} and set
\[
    \alpha_+:=\sup_{\wtE}\wtE[f(x)],
    \qquad
    \alpha_-:=\sup_{\wtE}\wtE[-f(x)].
\]
\State \Return
\[
    \widehat{\OPT}:=\max\{\alpha_+,\alpha_-\}.
\]
\end{algorithmic}
\end{algorithm}

\paragraph{Running time.}
For fixed $d$, the degree chosen by Algorithm~\ref{alg:degree-d-sos} is
$D_{d,k}=O_d(k)$.  The moment matrix has size
$\binom{n+O_d(k)}{O_d(k)}=n^{O_d(k)}$, so the algorithm can be implemented as
an SDP of size $n^{O_d(k)}$ in the usual real-arithmetic model.

\subsection{Soundness analysis}
\label{subsec:degree-d-one-sided-soundness}

The main technical statement is the following theorem.
It is applied to whichever of $f$ and $-f$ that realizes the larger SDP value in
Algorithm~\ref{alg:degree-d-sos}.

\begin{theorem}
\label{thm:degree-d-folded-soundness}
Let $d\ge3$ and $2\le k\le n$, and let $D_{d,k}$ be the degree used in
Algorithm~\ref{alg:degree-d-sos}.  Let $f(x)=T(x,\ldots,x)$ be a homogeneous
degree-$d$ polynomial, and let $\wtE$ be a feasible degree-$D_{d,k}$
pseudo-expectation satisfying
\begin{equation}
    \wtE[f(x)]\ge \tau>0.
    \label{eq:degree-d-positive-pseudo-value}
\end{equation}
Then there is a unit vector $y\in\Sph^{n-1}$ such that
\begin{equation}
    |f(y)|
    \ge
    \frac{d!}{d^d}\left(\frac{k}{3en}\right)^{(d-2)/2}\tau.
    \label{eq:degree-d-extraction-asymptotic}
\end{equation}
\end{theorem}

\begin{proof}
Let $m=d-2$, and let $U=(u_1,\ldots,u_m)$ be drawn from the product of the
uniform measures on $\Sph^{n-1}$.  For this tuple define
\begin{equation}
    p_U(x):=\prod_{j=1}^m\ip{u_j}{x},
    \qquad
    q_U(x):=T(u_1,\ldots,u_m,x,x),
    \qquad
    \Phi_k(U):=\wtE\big[q_U(x)^{2k}\big].
    \label{eq:degree-d-fold-definitions}
\end{equation}
Apply pseudo-H\"older (Lemma~\ref{lem:degree-d-pseudo-holder}) with $r=2k$ to the
functional $L=\E_U\wtE_x$, the polynomial $p_U$, and the quadratic fold
$q_U$.  This gives
\begin{equation}
    \left|\E_U\wtE\big[p_U(x)^{2k-1}q_U(x)\big]\right|^{2k}
    \le
    \Bigl(\E_U\wtE\big[p_U(x)^{2k}\big]\Bigr)^{2k-1}
    \E_U\wtE\big[q_U(x)^{2k}\big].
    \label{eq:degree-d-holder-application}
\end{equation}
The chosen degree is sufficient.  The terms $p_U^{2k}$ and $q_U^{2k}$ have
$x$-degrees $2k(d-2)$ and $4k$, respectively.  The intermediate squares in
Lemma~\ref{lem:degree-d-pseudo-holder} have $x$-degree
\[
    (2k-2j)(d-2)+4j,
    \qquad 0\le j\le k,
\]
whose maximum is at most $D_{d,k}$; the mixed term
$p_U^{2k-1}q_U$ also has degree at most $D_{d,k}$.

By the standard spherical moment (Lemma~\ref{lem:degree-d-spherical-moments}) and the sphere constraints \eqref{eq:degree-d-sphere-constraint},
\begin{equation}
    \E_U\wtE\big[p_U(x)^{2k}\big]
    =c_{n,k}^m\wtE\norm{x}_2^{2km}
    =c_{n,k}^m.
    \label{eq:degree-d-p-moment}
\end{equation}
For the mixed term, we have
\begin{align}
    \E_U\wtE\big[p_U(x)^{2k-1}q_U(x)\big]
    &=\wtE\left[
        \E_U\left[
          \prod_{j=1}^m \ip{u_j}{x}^{2k-1}
          T(u_1,\ldots,u_m,x,x)
        \right]
      \right]
      \label{eq:degree-d-mixed-expand}\\
    &=c_{n,k}^m\wtE\big[
        T(x,\ldots,x)\norm{x}_2^{(2k-2)m}
      \big]
      \label{eq:degree-d-mixed-average}\\
    &=c_{n,k}^m\wtE[f(x)]
      \label{eq:degree-d-mixed-sphere}\\
    &\ge c_{n,k}^m\tau .
    \label{eq:degree-d-mixed-moment}
\end{align}
Here \eqref{eq:degree-d-mixed-expand} expands the definitions of $p_U$ and
$q_U$ and uses the linearity of the pseudo-expectation to move the spherical
average inside.  To obtain \eqref{eq:degree-d-mixed-average}, average the
variables $u_1,\ldots,u_m$ successively.  For each $u_j$, the dependence on
$u_j$ is linear through the multilinear form $T$, and
\eqref{eq:degree-d-spherical-vector} replaces that coordinate by $x$ while
contributing the factor $c_{n,k}\norm{x}_2^{2k-2}$.  Finally,
\eqref{eq:degree-d-mixed-sphere} follows from $f(x)=T(x,\ldots,x)$ and the
sphere constraints, while \eqref{eq:degree-d-mixed-moment} uses
\eqref{eq:degree-d-positive-pseudo-value}.
Substituting \eqref{eq:degree-d-p-moment} and
\eqref{eq:degree-d-mixed-moment} into \eqref{eq:degree-d-holder-application}
gives the average lower bound
\begin{equation}
    \E_{U\sim(\Sph^{n-1})^m}\Phi_k(U)
    \ge c_{n,k}^m\tau^{2k},
    \label{eq:degree-d-average-lower-bound}
\end{equation}
where the expectation is with respect to the product uniform spherical
measure.  Let $U^\star$ be a maximizer of $\Phi_k$ over
$(\Sph^{n-1})^m$.  Then
\begin{equation}
    \Phi_k(U^\star)=\wtE\big[q_{U^\star}(x)^{2k}\big]
    \ge c_{n,k}^m\tau^{2k}.
    \label{eq:degree-d-good-fold}
\end{equation}
For notational simplicity, write $U^\star=(u_1,\ldots,u_m)$ in the remainder
of the proof.

Let $H_U\in\mathbb R^{n\times n}$ be the symmetric matrix representing this quadratic
fold,
\begin{equation}
    y^\top H_U y=q_U(y)=T(u_1,\ldots,u_m,y,y),
    \label{eq:degree-d-fold-matrix}
\end{equation}
and write $R=\norm{H_U}_{\mathrm{op}}$.  Since
$-RI\preceq H_U\preceq RI$, by Fact \ref{psd-sos} the quadratic forms
\[
    R\norm{x}_2^2-x^\top H_Ux,
    \qquad
    R\norm{x}_2^2+x^\top H_Ux
\]
are sums of squares.  Their product is
$R^2\norm{x}_2^4-(x^\top H_Ux)^2$, hence is a sum of squares.  Therefore
\begin{align*}
    &R^{2k}\norm{x}_2^{4k}-(x^\top H_Ux)^{2k}\\
    &\quad=
    \bigl(R^2\norm{x}_2^4-(x^\top H_Ux)^2\bigr)
    \sum_{r=0}^{k-1}
      (R^2\norm{x}_2^4)^{k-1-r}(x^\top H_Ux)^{2r}
\end{align*}
is a degree-$4k$ sum of squares.  Applying $\wtE$ and using the sphere
constraints gives
\begin{equation}
    \wtE\big[q_U(x)^{2k}\big]\le R^{2k}.
    \label{eq:degree-d-fold-upper-opnorm}
\end{equation}
Together with \eqref{eq:degree-d-good-fold}, this implies
\begin{equation}
    \norm{H_U}_{\mathrm{op}}
    \ge c_{n,k}^{m/(2k)}\tau.
    \label{eq:degree-d-opnorm-lower}
\end{equation}
Choose a unit eigenvector $w$ corresponding to an eigenvalue of $H_U$ with
largest absolute value.  Then
\begin{equation}
    |T(u_1,\ldots,u_m,w,w)|
    =|w^\top H_Uw|
    \ge c_{n,k}^{m/(2k)}\tau.
    \label{eq:degree-d-large-mixed-value}
\end{equation}
Applying the polarization lemma (Lemma~\ref{lem:degree-d-polarization}),
to the $d$ unit vectors
\(
    u_1,\ldots,u_m,w,w
\)
produces a nonzero normalized signed sum $y$ with
\begin{equation}
    |f(y)|
    \ge
    \frac{d!}{d^d}c_{n,k}^{m/(2k)}\tau.
    \label{eq:degree-d-extraction-cnk-proof}
\end{equation}
Since $m=d-2$ and $2\le k\le n$, the spherical moment (Lemma~\ref{lem:degree-d-spherical-moments})
gives \eqref{eq:degree-d-extraction-asymptotic}.
\end{proof}

\begin{corollary}[Correctness of Algorithm~\ref{alg:degree-d-sos}]
\label{cor:degree-d-algorithm-correctness}
For $2\le k\le n$, Algorithm~\ref{alg:degree-d-sos} returns a value
$\widehat{\OPT}$ satisfying
\begin{equation}
    \OPT(f)
    \le \widehat{\OPT}
    \le
    \frac{d^d}{d!}\left(\frac{3en}{k}\right)^{(d-2)/2}\OPT(f).
    \label{eq:degree-d-algorithm-correctness-asymptotic}
\end{equation}
Consequently, the degree-$D_{d,k}$ SoS relaxation gives an
$O_d((n/k)^{d/2-1})=O_d((n/D_{d,k})^{d/2-1})$ approximation to homogeneous
sphere optimization.
\end{corollary}

\begin{proof}
The lower bound $\OPT(f)\le\widehat{\OPT}$ follows because evaluation at any
unit vector is a feasible degree-$D_{d,k}$ pseudo-expectation, and one of the
two one-sided relaxations in Algorithm~\ref{alg:degree-d-sos} realizes the
absolute value at that vector.

For the upper bound, let $\tau=\widehat{\OPT}$.  If $\tau=0$, the claim is
immediate.  Otherwise, for every $\varepsilon>0$, one of the two signed
polynomials $g\in\{f,-f\}$ has a feasible degree-$D_{d,k}$ pseudo-expectation
with
\[
    \wtE[g(x)]\ge \tau-\varepsilon.
\]
Applying Theorem~\ref{thm:degree-d-folded-soundness} to $g$ gives a true unit
vector $y$ satisfying
\[
    \OPT(f)
    \ge |f(y)|
    =|g(y)|
    \ge
    \frac{d!}{d^d}\left(\frac{k}{3en}\right)^{(d-2)/2}(\tau-\varepsilon).
\]
Letting $\varepsilon\downarrow0$ proves \eqref{eq:degree-d-algorithm-correctness-asymptotic}.
\end{proof}

\begin{comment}
\begin{remark}[Algorithmic Rounding]
\label{rem:degree-d-constructive-witness}
The rounding step of Theorem~\ref{thm:degree-d-folded-soundness} is constructive.
Given a feasible degree-\(D_{d,k}\) pseudo-expectation for either \(f\) or
\(-f\) with pseudo-value at least \(\tau>0\), we can enumerate the spherical design constructed in Lemma \ref{lem:constructive-spherical-cubature}
 and find a fold tuple \(U\in\mathcal D_{n,k}^{d-2}\) satisfying
\[
    \wtE\big[q_U(x)^{2k}\big]\ge c_{n,k}^{d-2}\tau^{2k}.
\]
The subsequent eigenvector and polarization steps then output a unit vector
\(y\in\Sph^{n-1}\) with
\[
    |f(y)|
    \ge
    \frac{d!}{d^d}
    \left(\frac{k}{3en}\right)^{(d-2)/2}\tau .
\]
For fixed \(d\), this rounding stage runs in \(n^{O_d(k)}\) arithmetic
operations.
\end{remark}
\end{comment}

\bigskip
\section*{Acknowledgment}
The authors used GPT-5.5 Pro, and Claude as assistants for exploratory analysis, idea testing, and editorial polishing.
The authors reviewed and finalized all mathematical claims and proofs, and take full responsibility for the correctness, exposition, and attribution in the final manuscript.

\bibliographystyle{alpha}
\bibliography{reference}

\clearpage
\appendix
\section*{APPENDIX}

\crefalias{section}{appendix}
% ============================================================
% Merged appendices
% ============================================================

\section{BSS algorithmic rounding proof}
\label{app:bss-algorithmic-rounding}

Applying \cite{BAC18} on $\Phi_k$, we have the following lemma:
\begin{lemma}
\label{lem:bss-rtr-blackbox}
Let \((p_0,q_0)\in\Sph^{m-1}\times\Sph^{n-1}\) satisfy
\(\Phi_k(p_0,q_0)\ge A_0>0\).  For every fixed constant
\(0<\delta<1\), there is a deterministic algorithm which returns
\((u,v)\in\Sph^{m-1}\times\Sph^{n-1}\) such that, with
\(A=\Phi_k(u,v)\), one has \(A\ge A_0\) and, for every unit
\(h\perp u\) and every unit \(\ell\perp v\),
\begin{align}
\label{eq:bss-approx-local-x}
    \left|
    \left.\frac{d}{d\tau}
    \Phi_k\!\left(\frac{u+\tau h}{\sqrt{1+\tau^2}},v\right)
    \right|_{\tau=0}
    \right|
    &\le \delta A,
    &
    \left.\frac{d^2}{d\tau^2}
    \Phi_k\!\left(\frac{u+\tau h}{\sqrt{1+\tau^2}},v\right)
    \right|_{\tau=0}
    &\le \delta A, \\
\label{eq:bss-approx-local-y}
    \left|
    \left.\frac{d}{d\tau}
    \Phi_k\!\left(u,\frac{v+\tau\ell}{\sqrt{1+\tau^2}}\right)
    \right|_{\tau=0}
    \right|
    &\le \delta A,
    &
    \left.\frac{d^2}{d\tau^2}
    \Phi_k\!\left(u,\frac{v+\tau\ell}{\sqrt{1+\tau^2}}\right)
    \right|_{\tau=0}
    &\le \delta A .
\end{align}
The algorithm uses
\[
    (m+n)^{O(k)} A_0^{-O(1)}
\]
real-arithmetic operations for fixed \(\delta\).
\end{lemma}

\begin{proof}
Set
\[
    \mathcal M=\Sph^{m-1}\times\Sph^{n-1},
    \qquad
    f=-\Phi_k,
\]
and equip \(\mathcal M\) with the product Euclidean metric.  We first
establish quantitative smoothness bounds for \(f\).  Fix
\(0\le r\le3\), a point
\((p,q)\in\R^m\times\R^n\) satisfying
\(\norm{p}_2,\norm{q}_2\le1\).  Expanding the corresponding \(r\)-th
directional derivative of \(\Phi_k\), each differentiation acts on one of
the \(4k\) linear factors.  The total absolute coefficient in the resulting
expansion is therefore at most \((4k)^r\), and every term has the form
\[
    \wtE\!\left[
        \prod_{i=1}^{2k}\langle x,a_i\rangle
        \prod_{j=1}^{2k}\langle y,b_j\rangle
    \right],
    \qquad
    \norm{a_i}_2,\norm{b_j}_2\le1.
\]
For such a term, define
\[
    F
    :=
    \prod_{i=1}^{k}\langle x,a_i\rangle
    \prod_{j=1}^{k}\langle y,b_j\rangle,
    \qquad
    G
    :=
    \prod_{i=k+1}^{2k}\langle x,a_i\rangle
    \prod_{j=k+1}^{2k}\langle y,b_j\rangle.
\]
Pseudo-Cauchy--Schwarz Lemma \ref{pseudo-Cauchy–Schwarz} gives
\[
    |\wtE[FG]|
    \le
    \wtE[F^2]^{1/2}\wtE[G^2]^{1/2}.
\]
To bound \(\wtE[F^2]\), write the \(2k\) linear factors in \(F\) as
\(z_1,\ldots,z_{2k}\).  The identity
\[
    1-\prod_{s=1}^{2k}z_s^2
    =
    \sum_{s=1}^{2k}
    (1-z_s^2)\prod_{t<s}z_t^2
\]
has a degree-\(4k\) SoS certificate modulo the sphere constraints.  Namely,
for every \(\norm{a}_2\le1\),
\[
    1-\langle x,a\rangle^2
    =
    (1-\norm{x}_2^2)+x^\top(I-aa^\top)x,
\]
where \(I-aa^\top\succeq0\), and the analogous identity holds for factors
involving \(y\).  Hence \(\wtE[F^2]\le1\), and the same argument gives
\(\wtE[G^2]\le1\).  Each term in the derivative expansion consequently has
absolute value at most \(1\), so
\begin{equation}
\label{eq:bss-ambient-derivative-bound}
    \sup_{\norm{p}_2,\norm{q}_2\le1}
    \norm{D^r\Phi_k(p,q)}_{\operatorname{op}}
    \le
    (4k)^r,
    \qquad
    0\le r\le3.
\end{equation}

For \(z=(p,q)\in\mathcal M\) and
\(\eta=(\alpha,\beta)\in T_z\mathcal M\), use the normalization retraction
\[
    \operatorname{Retr}_{(p,q)}(\alpha,\beta)
    =
    \left(
        \frac{p+\alpha}{\norm{p+\alpha}_2},
        \frac{q+\beta}{\norm{q+\beta}_2}
    \right).
\]
Since \(\alpha\perp p\) and \(\beta\perp q\),
\[
    \norm{p+\alpha}_2
    =
    \sqrt{1+\norm{\alpha}_2^2},
    \qquad
    \norm{q+\beta}_2
    =
    \sqrt{1+\norm{\beta}_2^2}.
\]
Thus, on \(\norm{\eta}\le1\), direct differentiation of the two maps
\[
    \alpha\longmapsto
    (p+\alpha)(1+\norm{\alpha}_2^2)^{-1/2},
    \qquad
    \beta\longmapsto
    (q+\beta)(1+\norm{\beta}_2^2)^{-1/2}
\]
gives universal constants \(C_1,C_2,C_3\) such that
\[
    \sup_{\substack{z\in\mathcal M\\ \norm{\eta}\le1}}
    \norm{D^j\operatorname{Retr}_z(\eta)}_{\operatorname{op}}
    \le C_j,
    \qquad
    1\le j\le3.
\]
Combining these bounds with
\eqref{eq:bss-ambient-derivative-bound} and the chain rule, the pullbacks
\[
    \widehat f_z:=f\circ\operatorname{Retr}_z
\]
satisfy
\[
    \sup_{\substack{z\in\mathcal M\\ \norm{\eta}\le1}}
    \norm{D^2\widehat f_z(\eta)}_{\operatorname{op}}
    \le k^{O(1)},
    \qquad
    \sup_{\substack{z\in\mathcal M\\ \norm{\eta}\le1}}
    \norm{D^3\widehat f_z(\eta)}_{\operatorname{op}}
    \le k^{O(1)}.
\]
Taylor's theorem therefore verifies the pullback regularity assumptions
A4 and A5 of \cite[Sec.~3.1]{BAC18} with constants
\[
    1\le L_g,L_H\le k^{O(1)}.
\]
The normalization retraction is also a second-order retraction.

At every iterate \(z_j\), choose the exact pullback Hessian
\[
    H_j
    :=
    \nabla^2\widehat f_{z_j}(0)
\]
in the trust-region model.  It is linear and symmetric, and the preceding
bounds give an admissible constant
\[
    1\le c_0\le k^{O(1)}
\]
for A6.  Assumption A7 holds with \(c_1=0\).  Use the Cauchy step for
first-order iterations and the eigenstep for second-order iterations, as
in \cite[Sec.~3.3]{BAC18}; these satisfy A8 and A9 with
\[
    c_2=c_3=\frac12.
\]

Fix the algorithmic parameters
\[
    \rho'=\frac18,
    \qquad
    \overline\Delta=\Delta_0=1.
\]
With the notation of \cite[Sec.~3.4]{BAC18}, set
\[
    \lambda_g
    =
    \frac14
    \min\left\{
        \frac1{c_0},
        \frac{c_2}{L_g+c_0}
    \right\},
    \qquad
    \lambda_H
    =
    \frac34\frac{c_3}{L_H+c_1}.
\]
The quantitative bounds above imply
\[
    \lambda_g^{-1},\lambda_H^{-1}\le k^{O(1)}.
\]
Choose
\[
    \varepsilon_*
    :=
    \min\left\{
        \delta A_0,\,
        \frac{\Delta_0}{\lambda_g},\,
        \frac{\lambda_H}{\lambda_g^2},\,
        \frac1{\lambda_g}
    \right\},
    \qquad
    \varepsilon_g=\varepsilon_H=\varepsilon_*.
\]
Since \(c_2=c_3\), these choices satisfy all auxiliary tolerance conditions
in the Riemannian trust-region complexity theorem.  They also give
\begin{equation}
\label{eq:bss-epsilon-star-bound}
    \varepsilon_*^{-1}
    \le
    (\delta A_0)^{-1}+k^{O(1)}.
\end{equation}

Applying the Riemannian trust-region method of
\cite[Sec.~3.4]{BAC18} to \(f=-\Phi_k\) returns
\(z=(u,v)\in\mathcal M\) such that
\[
    \norm{\operatorname{grad}f(u,v)}
    \le\varepsilon_*,
    \qquad
    H_z\succeq-\varepsilon_* I.
\]
Because the retraction is second order and
\(H_z=\nabla^2\widehat f_z(0)\), the pullback-to-Riemannian Hessian
comparison in \cite[Sec.~3.5]{BAC18} gives
\[
    \operatorname{Hess}f(u,v)
    \succeq-\varepsilon_* I.
\]
Equivalently,
\begin{equation}
\label{eq:bss-rtr-output-bounds}
    \norm{\operatorname{grad}\Phi_k(u,v)}
    \le\varepsilon_*,
    \qquad
    \operatorname{Hess}\Phi_k(u,v)
    \preceq\varepsilon_* I.
\end{equation}

The trust-region method is a descent method for \(f=-\Phi_k\), and therefore
\[
    A:=\Phi_k(u,v)
    \ge
    \Phi_k(p_0,q_0)
    \ge A_0.
\]
Consequently,
\[
    \varepsilon_*
    \le
    \delta A_0
    \le
    \delta A.
\]
For every unit \(h\perp u\),
\[
    \operatorname{Retr}_{(u,v)}\bigl(\tau(h,0)\bigr)
    =
    \left(
        \frac{u+\tau h}{\sqrt{1+\tau^2}},
        v
    \right).
\]
Since the retraction is second order,
\[
    \left.\frac{d}{d\tau}
    \Phi_k\!\left(\frac{u+\tau h}{\sqrt{1+\tau^2}},v\right)
    \right|_{\tau=0}
    =
    \left\langle
        \operatorname{grad}\Phi_k(u,v),(h,0)
    \right\rangle
\]
and
\[
    \left.\frac{d^2}{d\tau^2}
    \Phi_k\!\left(\frac{u+\tau h}{\sqrt{1+\tau^2}},v\right)
    \right|_{\tau=0}
    =
    \left\langle
        (h,0),
        \operatorname{Hess}\Phi_k(u,v)[(h,0)]
    \right\rangle .
\]
Applying \eqref{eq:bss-rtr-output-bounds} proves
\eqref{eq:bss-approx-local-x}; the same argument with the unit tangent
vector \((0,\ell)\) proves \eqref{eq:bss-approx-local-y}.

Finally, the iteration bound in \cite[Sec.~3.4]{BAC18}, together with
\eqref{eq:bss-epsilon-star-bound} and
\(-1\le f\le0\), is
\[
    k^{O(1)}\varepsilon_*^{-3}
    =
    k^{O(1)}A_0^{-O(1)}
\]
for fixed \(\delta\).  In each iteration, \(\Phi_k\), its gradient, and its
pullback Hessian can be evaluated from the supplied degree-\(4k\) moment
table, and the Cauchy or eigenstep requires only polynomial-size linear
algebra in dimension \(m+n-2\).  Since the moment table has size
\[
    \binom{m+n+4k}{4k}
    =
    (m+n)^{O(k)},
\]
the total number of real-arithmetic operations is
\(
    (m+n)^{O(k)}A_0^{-O(1)}.
\)
\end{proof}
\begin{lemma}
\label{lem:bss-approximate-local-rounding}
Let \(\wtE\) be a feasible degree-\(4k\) pseudo-expectation from Algorithm \ref{alg:sos-bss}.
Let \((u,v)\in\Sph^{m-1}\times\Sph^{n-1}\), set
\(A=\Phi_k(u,v)>0\), and assume that
\eqref{eq:bss-approx-local-x} and \eqref{eq:bss-approx-local-y} hold with
parameter \(\delta\in(0,1)\).  Then
\begin{equation}
\label{eq:bss-approx-local-distance}
    \distF(uv^\top,W)
    \le
    \frac{((m-1)(n-1))^{1/4}}{2k-1}
    \left(1+\frac{\delta}{2k}\right)
    +\frac{\delta}{k} .
\end{equation}
\end{lemma}
\begin{proof}
We use the notation \(s,t,\xi,\zeta,w,\wtE_w,\sigma,X,Y,E,M\) from the
proof of Theorem~\ref{thm:main}.  Set
\[
    m_x:=\wtE_w[st^2\xi],
    \qquad
    m_y:=\wtE_w[s^2t\zeta].
\]
Combining the derivative identities
\eqref{eq:x-first-derivative}--\eqref{eq:y-second-derivative} with
\eqref{eq:bss-approx-local-x} and \eqref{eq:bss-approx-local-y} gives
\[
    \|m_x\|_2,\ \|m_y\|_2
    \le \frac{\delta\sigma}{2k},
    \qquad
    \|X\|_{\mathrm F}
    \le
    \frac{\sqrt{m-1}\,\sigma}{2k-1}
    \left(1+\frac{\delta}{2k}\right),
    \qquad
    \|Y\|_{\mathrm F}
    \le
    \frac{\sqrt{n-1}\,\sigma}{2k-1}
    \left(1+\frac{\delta}{2k}\right).
\]
The same block-PSD argument as in Theorem~\ref{thm:main} gives
\(\|E\|_{\mathrm F}^2\le \|X\|_{\mathrm F}\|Y\|_{\mathrm F}\).  Hence
\[
    \frac{\|E\|_{\mathrm F}}{\sigma}
    \le
    \frac{((m-1)(n-1))^{1/4}}{2k-1}
    \left(1+\frac{\delta}{2k}\right).
\]
Moreover,
\[
    M=\sigma uv^\top+m_xv^\top+um_y^\top+E\in W.
\]
Since \(M/\sigma\in W\),
\[
\begin{aligned}
    \distF(uv^\top,W)
    &\le
    \left\|uv^\top-\frac{M}{\sigma}\right\|_{\mathrm F}\\
    &\le
    \frac{\|m_x\|_2+\|m_y\|_2+\|E\|_{\mathrm F}}{\sigma}\\
    &\le
    \frac{((m-1)(n-1))^{1/4}}{2k-1}
    \left(1+\frac{\delta}{2k}\right)
    +\frac{\delta}{k},
\end{aligned}
\]
which is \eqref{eq:bss-approx-local-distance}.
\end{proof}
% \begin{proof}
% We use the same notations from Theorem~\ref{thm:main}.
% Combining \eqref{eq:x-first-derivative},
% \eqref{eq:y-first-derivative}, \eqref{eq:bss-approx-local-x} and
% \eqref{eq:bss-approx-local-y} gives
% \[
%     \|m_x\|_2\le\frac{\delta a}{2k},
%     \qquad
%     \|m_y\|_2\le\frac{\delta a}{2k},\qquad
%         \|X\|_{\mathrm F}
%     \le
%     \frac{\sqrt{m-1}\,A}{2k-1}
%     \left(1+\frac{\delta}{2k}\right),
%     \qquad
%     \|Y\|_{\mathrm F}
%     \le
%     \frac{\sqrt{n-1}\,A}{2k-1}
%     \left(1+\frac{\delta}{2k}\right).
% \]

% The same block-PSD argument as in Theorem~\ref{thm:main} gives
% \(\|C\|_{\mathrm F}^2\le\|X\|_{\mathrm F}\|Y\|_{\mathrm F}\).  Hence
% \[
%     \frac{\|c\|_{\mathrm F}}{a}
%     =
%     \frac{\|C\|_{\mathrm F}}{A}
%     \le
%     \frac{((m-1)(n-1))^{1/4}}{2k-1}
%     \left(1+\frac{\delta}{2k}\right).
% \]
% Since \(M/a\in W\),
% \[
% \begin{aligned}
%     \distF(uv^\top,W)
%     \le
%     \left\|uv^\top-\frac{M}{a}\right\|_{\mathrm F}
%     \le
%     \frac{\|m_x\|_2}{a}
%     +
%     \frac{\|m_y\|_2}{a}
%     +
%     \frac{\|c\|_{\mathrm F}}{a},
% \end{aligned}
% \]
% which proves \eqref{eq:bss-approx-local-distance}.
% \end{proof}

\begin{proof}[Proof of Theorem~\ref{thm:bss-algorithmic-rounding}]
Choose a universal constant \(2\sqrt2\,\delta_0<1-\frac1{\sqrt2}\).
Let \(e_1,\ldots,e_m\) and \(f_1,\ldots,f_n\) be the standard bases of
\(\R^m\) and \(\R^n\).  Choose \((i_0,j_0)\) maximizing
\(\Phi_k(e_i,f_j)\) over all coordinate pairs, and set
\[
    (p_0,q_0)=(e_{i_0},f_{j_0}),
    \qquad
    A_0=\Phi_k(p_0,q_0).
\]
By the elementary SoS power-mean certificate, modulo
\(\|x\|_2^2=1\) and \(\|y\|_2^2=1\),
\(
    \sum_{i,j}x_i^{2k}y_j^{2k}
    \ge
    m^{1-k}n^{1-k}
\) has a degree-\(4k\) SoS certificate.
Applying \(\wtE\) gives
\[
    \sum_{i=1}^m\sum_{j=1}^n \Phi_k(e_i,f_j)
    =
    \wtE\sum_{i=1}^m\sum_{j=1}^n x_i^{2k}y_j^{2k}
    \ge
    m^{1-k}n^{1-k}.
\]
Hence,
\begin{equation}
\label{eq:bss-a0-lower-bound}
    A_0^{-1}\le (mn)^k\le (m+n)^{O(k)} .
\end{equation}
Apply Lemma~\ref{lem:bss-rtr-blackbox} with this starting point \((p_0,q_0)\) and
\(\delta=\delta_0\).  It returns
\((\widehat u,\widehat v)\in\Sph^{m-1}\times\Sph^{n-1}\) satisfying
\eqref{eq:bss-approx-local-x} and \eqref{eq:bss-approx-local-y}.  Put
\(\widehat X=\widehat u\widehat v^\top\).  By
Lemma~\ref{lem:bss-approximate-local-rounding},
\[
    \distF(\widehat X,W)
    \le
    \frac{R_{m,n}}{2k-1}
    \left(1+\frac{\delta_0}{2k}\right)
    +
    \frac{\delta_0}{k} .
\]
The degree condition \eqref{eq:kcondition} gives
\(
    \frac{R_{m,n}}{2k-1}\le\frac{\sqrt\epsilon}{\sqrt2},
    \frac1k\le\sqrt2\sqrt\epsilon,
\)
where the second inequality uses \(R_{m,n}\ge1\).  Therefore
\[
    \distF(\widehat X,W)
    \le
    \frac{\sqrt\epsilon}{\sqrt2}
    +2\sqrt2\,\delta_0\sqrt\epsilon
    <
    \sqrt\epsilon.
\] Hence
\(
    \|\Pi_W(\widehat X)\|_{\mathrm F}^2
    =
    1-\distF(\widehat X,W)^2
    >
    1-\epsilon .
\)
By
\eqref{eq:bss-a0-lower-bound}, the trust-region stage has the same asymptotic
cost, up to a polynomial factor in the supplied moment-table size.  Hence the running time for rounding is absorbed into the
SoS feasibility time scale used in Algorithm~\ref{alg:sos-bss}.
\end{proof}

\section{Positivity of the maximum and reweighing}
\label{sec:appendix-pos}
In this section, we prove Lemma~\ref{lem:positive-projective-moment}
\begin{proof}
Let \(r=\operatorname{rank}(P)\ge1\), and let \(\sigma_W\) denote the
uniform probability measure on \(\Sph(W)\). Since \(v\in W\), \(\langle x,v\rangle=\langle Px,v\rangle\). Identifying \(W\) isometrically
with \(\mathbb R^r\), Lemma \ref{lem:degree-d-spherical-moments} applied in dimension \(r\) gives
\[
    \int_{\Sph(W)} \langle x,v\rangle^{2k}\,d\sigma_W(v)
    =
    c_{r,k}\|Px\|_2^{2k},
\]
where \(c_{r,k}>0\) is the \(2k\)-th moment of the first coordinate of a
uniform point on \(S^{r-1}\).  Therefore
\[
    \int_{\Sph(W)} \Phi_k(v)\,d\sigma_W(v)
    =
    c_{r,k}\,\wtE \norm{Px}_2^{2k}.
\]

We claim that
\[
    \wtE \norm{Px}_2^{2k}=1.
\]
Indeed, since \(P\) is an orthogonal projector,
\[
    \norm{Px}_2^2-\norm{x}_2^2
    =
    -x^\top(I-P)x
    =
    -\sum_{i=1}^n x_i((I-P)x)_i .
\]
Hence, for every polynomial \(R\) with \(\deg R\le 2k-2\),
\[
    \wtE\!\left[
        R(x)(\norm{Px}_2^2-\norm{x}_2^2)
    \right]=0
\]
by the subspace constraint, because each multiplier \(R(x)x_i\) has
degree at most \(2k-1\).  Taking
\[
    R(x)=\sum_{j=0}^{k-1}
        \norm{Px}_2^{2(k-1-j)}\norm{x}_2^{2j}
\]
gives
\[
    \wtE \norm{Px}_2^{2k}
    =
    \wtE \norm{x}_2^{2k}.
\]
The sphere constraint implies
\(\
    \wtE \norm{x}_2^{2k}=1,
\)
since
\(
    \norm{x}_2^{2k}-1
    =
    (\norm{x}_2^2-1)
    \sum_{j=0}^{k-1}\norm{x}_2^{2j},
\)
and the multiplier has degree \(2k-2\).  Thus
\[
    \int_{\Sph(W)} \Phi_k(v)\,d\sigma_W(v)
    =
    c_{r,k}>0.
\]
Since \(u\) is a global maximizer of \(\Phi_k\) on \(\Sph(W)\), it follows that
\[
    \Phi_k(u)\ge \int_{\Sph(W)}\Phi_k(v)\,d\sigma_W(v)>0.
\]
It remains to verify that the reweighing has positive total mass. Since
\(w(x)=s^{2k-4}\),
\begin{align}\label{eq:weighted-mass-comparison}
    \wtE w(x)-\wtE s^{2k}=
    \wtE\!\left[s^{2k-4}(1+s^2)(1-s^2)\right].
\end{align}
Let \(\ell(x):=(I-uu^\top)x\). Modulo the sphere constraint,
\(
    1-s^2=\norm{\ell(x)}_2^2.
\)
Consequently, the right-hand side of
\eqref{eq:weighted-mass-comparison} equals
\[
    \wtE\!\left[
        \sum_{j=1}^n(s^{k-2}\ell_j(x))^2
        +
        \sum_{j=1}^n(s^{k-1}\ell_j(x))^2
    \right]
    \ge0.
\]
Each squared polynomial has degree at most \(k\), and the multiplier of
the sphere equality has degree at most \(2k-2\). This proves
\eqref{eq:positive-projective-moment}.
\end{proof}

\section{Auxiliary lemmas}

\begin{lemma}[PSD block inequality~\cite{Bha07}]
\label{lem:block-cauchy}
If
\[
\begin{pmatrix}
X & C\\
C^\top & Y
\end{pmatrix}
\succeq 0,
\]
then
\[
\|C\|_{\mathrm F}^2\le \|X\|_{\mathrm F}\|Y\|_{\mathrm F}.
\]
\end{lemma}

\begin{proof}
When \(X\) and \(Y\) are positive definite, the Schur complement implies
\[
Y-C^\top X^{-1}C\succeq 0 .
\]
Equivalently, with
\[
K=X^{-1/2}CY^{-1/2},
\]
one has \(K^\top K\preceq I\).  Thus
\[
C=X^{1/2}KY^{1/2}.
\]
The singular case follows by replacing \(X,Y\) with \(X+\varepsilon I\), \(Y+\varepsilon I\), applying the argument above, and sending \(\varepsilon\downarrow 0\).

Using the contraction factorization,
\[
\|C\|_{\mathrm F}^2
=
\Tr(C^\top C)
=
\Tr\!\left(Y^{1/2}K^\top XKY^{1/2}\right)
=
\Tr\!\left((K^\top XK)Y\right).
\]
By Cauchy--Schwarz for the Frobenius inner product,
\[
\Tr\!\left((K^\top XK)Y\right)
\le
\|K^\top XK\|_{\mathrm F}\|Y\|_{\mathrm F}.
\]
Since \(K^\top K\preceq I\),
\[
\|K^\top XK\|_{\mathrm F}\le \|X\|_{\mathrm F}.
\]
Combining the inequalities gives the claim.
\end{proof}

\begin{lemma}
\label{lem:centered-quartic-sos}
For every \(\eta>0\), define
\[
    C_\eta:=1+\frac{18}{\eta}+3\left(\frac{2}{\eta}\right)^{1/3}.
\]
Then, for all \(a,b\in\R\),
\begin{equation}
\label{eq:centered-quartic-sos}
    (a+b)^4
    \le
    (1+\eta)a^4+4a^3b+C_\eta b^4.
\end{equation}
Moreover, \eqref{eq:centered-quartic-sos} has a quartic
sum-of-squares certificate.
\end{lemma}

\begin{proof}
Expanding both sides, it suffices to certify
\[
    \eta a^4-6a^2b^2-4ab^3+(C_\eta-1)b^4\ge0.
\]
The following identity gives an explicit sum-of-squares certificate:
\[
\begin{aligned}
&\eta a^4-6a^2b^2-4ab^3+(C_\eta-1)b^4 \\
&\qquad =
\frac{\eta}{2}
\left(a^2-\frac{6}{\eta}b^2\right)^2
+
\frac{\eta}{2}
\left(a-\left(\frac{2}{\eta}\right)^{1/3}b\right)^2 \\
&\qquad\qquad\cdot
\left(
    \left(a+\left(\frac{2}{\eta}\right)^{1/3}b\right)^2
    +2\left(\frac{2}{\eta}\right)^{2/3}b^2
\right).
\end{aligned}
\]
Thus \eqref{eq:centered-quartic-sos} is a quartic SoS inequality.
\end{proof}

\begin{lemma}
\label{lem:pq-row-binary-sos}
Let \(q\ge4\) be even.  There is a constant \(C_q>0\) such that, for all
\(a,b\in\R\),
\begin{equation}
\label{eq:pq-row-binary-sos}
    (a+b)^q
    \le
    a^q+qa^{q-1}b+C_q\bigl(a^{q-2}b^2+b^q\bigr).
\end{equation}
Moreover, after moving the left-hand side of
\eqref{eq:pq-row-binary-sos} to the right, the resulting homogeneous binary
form is nonnegative and hence is a sum of squares of real polynomials.
\end{lemma}

\begin{proof}
Expanding \((a+b)^q\) and retaining the constant and linear terms leaves
only monomials \(a^{q-\ell}b^{\ell}\) with \(2\le \ell\le q\).  If
\(|b|\le|a|\), each such monomial is bounded in absolute value by a constant
multiple of \(|a|^{q-2}b^2\); if \(|b|>|a|\), it is bounded by a constant
multiple of \(|b|^q\).  Since \(q\) is even, this proves
\eqref{eq:pq-row-binary-sos} with \(C_q\) chosen sufficiently large.  The right-hand side
minus the left-hand side is then a nonnegative homogeneous binary form of even
degree, by
\cite[Theorem~2.1(a)]{Naldi2014} every such form is a sum of squares over \(\R\).
\end{proof}
\begin{lemma}[Spherical moments]
\label{lem:degree-d-spherical-moments}
Let $u$ be uniform on $\Sph^{n-1}$, and define
\begin{equation}
    c_{n,k}:=\E_u\ip{u}{e_1}^{2k}.
    \label{eq:degree-d-cnk-definition}
\end{equation}
Then
\begin{equation}
    c_{n,k}
    =\frac{(2k-1)!!}{n(n+2)\cdots(n+2k-2)},
    \qquad
    c_{n,k}^{1/(2k)}\ge\sqrt{\frac{k}{3en}}\quad(2\le k\le n).
    \label{eq:degree-d-cnk-formula-lower}
\end{equation}
For every fixed $x,a\in\R^n$,
\begin{align}
    \E_u\ip{u}{x}^{2k}
        &=c_{n,k}\norm{x}_2^{2k},
        \label{eq:degree-d-spherical-even}\\
    \E_u\ip{u}{x}^{2k-1}\ip{u}{a}
        &=c_{n,k}\norm{x}_2^{2k-2}\ip{x}{a}.
        \label{eq:degree-d-spherical-vector}
\end{align}
\end{lemma}

\begin{proof}
The formula for $c_{n,k}$ follows from the standard Gaussian representation
$u=g/\norm{g}_2$, where $g\sim N(0,I_n)$.  Equivalently,
$u_1^2$ has the beta distribution with parameters $1/2$ and $(n-1)/2$, so
\[
    \E u_1^{2k}
    =\frac{(1/2)(3/2)\cdots(k-1/2)}{(n/2)(n/2+1)\cdots(n/2+k-1)}
    =\frac{(2k-1)!!}{n(n+2)\cdots(n+2k-2)}.
\]
For $2\le k\le n$,
\[
    (2k-1)!!\ge k!\ge(k/e)^k,
    \qquad
    n(n+2)\cdots(n+2k-2)\le(3n)^k.
\]
Substituting these two bounds into the formula for $c_{n,k}$ proves the
lower bound in \eqref{eq:degree-d-cnk-formula-lower}.
Rotational invariance gives \eqref{eq:degree-d-spherical-even}.  For the
vector identity, the vector
$\E_u\ip{u}{x}^{2k-1}u$ must be parallel to $x$.  Taking its inner product
with $x$ and using \eqref{eq:degree-d-spherical-even} gives
\[
    \E_u\ip{u}{x}^{2k-1}u
    =c_{n,k}\norm{x}_2^{2k-2}x,
\]
and \eqref{eq:degree-d-spherical-vector} follows by pairing with $a$.
\end{proof}

\begin{lemma}[Pseudo-H\"older]
\label{lem:degree-d-pseudo-holder}
Let $r\ge2$ be even.  Let $L$ be a linear functional defined on the
polynomials below and assume that
\[
    L[h^2]\ge0
    \quad\text{for every }h\in
    \operatorname{span}\{p^{r/2-j}q^j:0\le j\le r/2\}.
\]
Then
\begin{equation}
    |L[p^{r-1}q]|^r
    \le L[p^r]^{r-1}L[q^r].
    \label{eq:degree-d-pseudo-holder}
\end{equation}
\end{lemma}

\begin{proof}
For $0\le j\le r/2$, put
\[
    c_j:=L[p^{r-2j}q^{2j}]=L[(p^{r/2-j}q^j)^2]\ge0.
\]
Cauchy--Schwarz for the positive semidefinite bilinear form induced by $L$ on
the displayed span gives
\[
    c_j^2\le c_{j-1}c_{j+1}
    \qquad (1\le j<r/2).
\]
Thus the sequence $(c_j)_{j=0}^{r/2}$ is log-convex.  Consequently,
\begin{equation}
    c_1\le c_0^{1-2/r}c_{r/2}^{2/r}.
    \label{eq:degree-d-log-convex-step}
\end{equation}
A final Cauchy--Schwarz inequality gives
\[
    |L[p^{r-1}q]|^2
    \le c_0c_1
    \le c_0^{2-2/r}c_{r/2}^{2/r}.
\]
Raising both sides to the power $r/2$ proves the claim.  If one of the
endpoints $c_0,c_{r/2}$ is zero, the same Cauchy--Schwarz inequalities give
$L[p^{r-1}q]=0$, and the inequality is immediate.
\end{proof}

The following mixed-to-diagonal estimate is the standard consequence of the
real polarization identity of Banach~\cite{Ban38}; the same form is also used in~\cite{HLZ10, So11, BGGLT}.

\begin{lemma}[Polarization lemma]
\label{lem:degree-d-polarization}
Let $T:(\R^n)^d\to\R$ be symmetric and let $g(x):=T(x,\ldots,x)$.  If
$x_1,\ldots,x_d\in \Sph^{n-1}$ and
\[
    |T(x_1,\ldots,x_d)|\ge\beta,
\]
then one of the nonzero normalized signed sums
\[
    y_\varepsilon
    :=\frac{\sum_{i=1}^d\varepsilon_i x_i}
            {\norm{\sum_{i=1}^d\varepsilon_i x_i}_2},
    \qquad \varepsilon\in\{\pm1\}^d,
\]
satisfies
\begin{equation}
    |g(y_\varepsilon)|\ge\frac{d!}{d^d}\,\beta.
    \label{eq:degree-d-polarization-loss}
\end{equation}
\end{lemma}

\end{document}